\documentclass[11pt]{article}
\usepackage[margin=.8in,left=.8in]{geometry}
\usepackage{amsthm}
\usepackage{mathtools}
\usepackage{xcolor}
\usepackage{stmaryrd}

\usepackage{bbm}

\usepackage{enumitem} 
\setlist{
  topsep=1pt, 
  itemsep=0pt,
  parsep=1pt
}

\usepackage{adjustbox}

\usepackage{hyperref}
\hypersetup{
  colorlinks = true,
  allcolors=.  
}

\let\PLAINthebibliography\thebibliography
\renewcommand\thebibliography[1]{
  \PLAINthebibliography{#1}
  \setlength{\parskip}{1pt}
  \setlength{\itemsep}{1pt plus .3ex}
}

\usepackage{tikz}
\usetikzlibrary{decorations.pathmorphing}
\usetikzlibrary{arrows}
\usetikzlibrary{cd}

\DeclareMathAlphabet{\mathpzc}{OT1}{pzc}{m}{it} 

\newcommand\mathscr[1]{\scalebox{1.1}{$\mathpzc{#1}$}}

\newcommand{\proofstep}[1]{\scalebox{.85}{ #1 }}

\usepackage[titletoc]{appendix}

\usepackage[safe]{tipa} 

\definecolor{darkblue}{rgb}{0.05,0.25,0.65}
\definecolor{greenii}{RGB}{20,140,10}
\definecolor{lightgray}{rgb}{0.9,0.9,0.9}
\definecolor{orangeii}{RGB}{200,100,5}

\usepackage{multirow}

\usepackage{stmaryrd}    

\usepackage{amssymb,amsmath}

\newcounter{sqindex}

\usepackage{mathptmx}
\usepackage{amsmath}
\usepackage{graphicx}
\DeclareRobustCommand{\coprod}{\mathop{\text{\fakecoprod}}}
\newcommand{\fakecoprod}{%
  \sbox0{$\prod$}%
  \smash{\raisebox{\dimexpr.9625\depth-\dp0}{\scalebox{1}[-1]{$\prod$}}}%
  \vphantom{$\prod$}%
}

\newcommand{\HomotopyQuotient}[2]{
  #1 \!\sslash\! #2
}

\usepackage[new]{old-arrows}   

\def\acts{\raisebox{1.4pt}{\;\rotatebox[origin=c]{90}{$\curvearrowright$}}\hspace{.5pt}}

\makeatletter
\newif\if@sup
\newtoks\@sups
\def\append@sup#1{\edef\act{\noexpand\@sups={\the\@sups #1}}\act}%
\def\reset@sup{\@supfalse\@sups={}}%
\def\mk@scripts#1#2{\if #2/ \if@sup ^{\the\@sups}\fi \else%
  \ifx #1_ \if@sup ^{\the\@sups}\reset@sup \fi {}_{#2}%
  \else \append@sup#2 \@suptrue \fi%
  \expandafter\mk@scripts\fi}
\def\tensor#1#2{\reset@sup#1\mk@scripts#2_/}
\def\multiscripts#1#2#3{\reset@sup{}\mk@scripts#1_/#2%
  \reset@sup\mk@scripts#3_/}
\makeatother

\makeatletter
\newbox\slashbox \setbox\slashbox=\hbox{$/$}
\def\itex@pslash#1{\setbox\@tempboxa=\hbox{$#1$}
  \@tempdima=0.5\wd\slashbox \advance\@tempdima 0.5\wd\@tempboxa
  \copy\slashbox \kern-\@tempdima \box\@tempboxa}
\def\slash{\protect\itex@pslash}
\makeatother

\def\clap#1{\hbox to 0pt{\hss#1\hss}}
\def\mathllap{\mathpalette\mathllapinternal}
\def\mathrlap{\mathpalette\mathrlapinternal}
\def\mathclap{\mathpalette\mathclapinternal}
\def\mathllapinternal#1#2{\llap{$\mathsurround=0pt#1{#2}$}}
\def\mathrlapinternal#1#2{\rlap{$\mathsurround=0pt#1{#2}$}}
\def\mathclapinternal#1#2{\clap{$\mathsurround=0pt#1{#2}$}}

\let\oldroot\root
\def\root#1#2{\oldroot #1 \of{#2}}
\renewcommand{\sqrt}[2][]{\oldroot #1 \of{#2}}

\DeclareSymbolFont{symbolsC}{U}{txsyc}{m}{n}
\SetSymbolFont{symbolsC}{bold}{U}{txsyc}{bx}{n}
\DeclareFontSubstitution{U}{txsyc}{m}{n}

\DeclareSymbolFont{stmry}{U}{stmry}{m}{n}
\SetSymbolFont{stmry}{bold}{U}{stmry}{b}{n}

\DeclareFontFamily{OMX}{MnSymbolE}{}
\DeclareSymbolFont{mnomx}{OMX}{MnSymbolE}{m}{n}
\SetSymbolFont{mnomx}{bold}{OMX}{MnSymbolE}{b}{n}
\DeclareFontShape{OMX}{MnSymbolE}{m}{n}{
    <-6>  MnSymbolE5
   <6-7>  MnSymbolE6
   <7-8>  MnSymbolE7
   <8-9>  MnSymbolE8
   <9-10> MnSymbolE9
  <10-12> MnSymbolE10
  <12->   MnSymbolE12}{}

\usepackage{cleveref}

\crefformat{section}{\S#2#1#3} 
\crefformat{subsection}{\S#2#1#3}
\crefformat{subsubsection}{\S#2#1#3}

\theoremstyle{plain}
\newtheorem{theorem}{Theorem}[section]
\newtheorem{lemma}[theorem]{Lemma}
\newtheorem{proposition}[theorem]{Proposition}

\theoremstyle{definition}
\newtheorem{definition}[theorem]{Definition}

\newtheorem{example}[theorem]{Example}

\newtheorem{remark}[theorem]{Remark}

\newcommand{\defneq}{\equiv}

\begin{document}

\title{
Entanglement of Sections: 

The pushout of entangled and parameterized quantum information}

\author{
  Hisham Sati${}^{\ast, \dagger}$
\qquad 
  Urs Schreiber${}^\ast$
}

\maketitle

\begin{abstract}
A question raised by Freedman \& Hastings \cite{FreedmanHastings23} still stands: To produce a mathematical theory that would unify quantum entanglement/tensor-structure 
with parameterized/bundle-structure via their amalgamation (a hypothetical pushout) along bare quantum (information) 
theory  --- a question motivated by the role that vector bundles of spaces of quantum states play in the K-theoretic classification 
of topological phases of matter. 

\smallskip

Here we produce a possible answer to this question. To that end, first we make precise a form of the relevant pushout diagram in monoidal category theory. 
With the question thus formalized, we proceed to compute this pushout and prove that it gives what is known as the {\it external} tensor product on vector bundles/K-classes, or rather on 
flat such bundles (flat K-theory), i.e., those equipped with monodromy encoding topological Berry phases. The external tensor product 
was recently highlighted in the context of topological phases of matter in \cite{Mera20} and through our work in quantum programming 
theory \cite{QS} but has not otherwise found due attention in quantum theory yet.
\end{abstract}

\vspace{1.5cm}

 \begin{center}
 \begin{minipage}{11.5cm}
 \tableofcontents
 \end{minipage}
 \end{center}

\vfill

\hrule
\vspace{5pt}

{
\footnotesize
\noindent
\def\arraystretch{1}
\tabcolsep=0pt
\begin{tabular}{ll}
${}^*$\,
&
Mathematics, Division of Science; and
\\
&
Center for Quantum and Topological Systems,
\\
&
NYUAD Research Institute,
\\
&
New York University Abu Dhabi, UAE.  
\end{tabular}
\hfill
\adjustbox{raise=-15pt}{
\includegraphics[width=3cm]{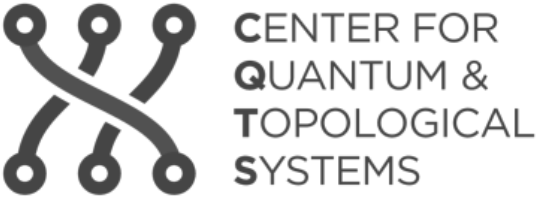}
}

\vspace{1mm} 
\noindent ${}^\dagger$The Courant Institute for Mathematical Sciences, NYU, NY

\vspace{.05cm}

\noindent
The authors acknowledge the support by {\it Tamkeen UAE} under the 
{\it NYU Abu Dhabi Research Institute grant} {\tt CG008}.
}

\newpage

\section{Introduction and outline}
\label{IntroductionAndOutline}

\noindent
{\bf Entangled quantum processes and Tensor categories.}
The natural logical framework for  pure quantum information theory (e.g. \cite{NielsenChuang00}) is (this goes back to \cite[p. 7]{Girard87}\cite{Yetter90}\cite{Pratt92}, has been put to practical use since \cite{AbramskyCoecke04}\cite{AbramskyDuncan06}, further exposition may be found in \cite{Slavnov05}\cite{Baez06}\cite{BaezStay11}\cite{CoeckeKissinger17}\cite{HeunenVicary19}) the ``internal logic'' -- in fact the ``internal type theory'' -- of the closed monoidal 
$\dagger$-category of (finite-dimensional) complex Hilbert spaces, with respect to their usual linear tensor product  ``$\otimes$'':
\smallskip 
\begin{itemize}[leftmargin=.5cm]

\item[1.] ({\bf Superposition}) The fact that for finite-dimensional Hilbert spaces the classical logical connective of conjunction
(interpreted as the cartesian product) merges with the logical disjunction (the coproduct) into a single {\it biproduct} (the direct sum) effectively encodes the superposition principle of quantum physics.

\item[2.] ({\bf No-cloning}) the appearance of another ``multiplicative conjunction'' interpreted as the {\it tensor} product 
reflects, due to its {\it non-cartesian} nature, the no-cloning/no-deletion constraints on pure quantum processes (interpreted 
as the absence of diagonal- and projection morphisms).

\item [3.] ({\bf Entanglement}) and in combination, these give rise to the all-important phenomenon of entanglement 
of quantum states (namely the superposition of product states in a tensor product).
\end{itemize}

\smallskip 
\noindent In fact, all of these phenomena are governed already by  (the internal linear logic of, cf. \cite{ValironZdancewic14}\cite{Murfet14}) the underlying tensor category of complex vector spaces:
$$
  \adjustbox{raise=8pt}{
  \scalebox{.8}{
  \begin{tabular}{l}
    Backdrop for {\it entangled} quantum 
    \\
    processes featuring:
    \\
    superposition, no-cloning,
    \\
    and quantum entanglement
  \end{tabular}
  }
  }
  \hspace{1.5cm}
  \overset{
    \mathclap{
    \raisebox{7pt}{
      \scalebox{.7}{
        \color{darkblue}
        \bf
       \def\arraystretch{.9}
       \begin{tabular}{c}
        {\color{purple} non-cartesian monoidal} category
        \\
        of complex vector spaces
        \end{tabular}
      }
    }
    }
  }{
    (\mathrm{Mod}_{\mathbb{C}}, \otimes)
  }
  \qquad \quad 
  =
  \quad
  \left\{\!\!\!\!
  \adjustbox{raise=4pt}{
  \begin{tikzcd}[row sep=-5pt, column sep=50pt]
    \overset{
      \mathclap{
        \rotatebox{24}{
          \hspace{-15pt}
          \rlap{
            \scalebox{.7}{
              \color{darkblue}
              \bf
              vector space
            }
          }
        }
      }
    }{
      \mathscr{V} 
    }
      \ar[
        rr,
        "{ \phi }"{description},
        "{
          \scalebox{.7}{
            \color{greenii}
            \bf
            linear map
          }
        }"{yshift=3pt},
      ] 
    && 
    \mathscr{W}
    \\
    \otimes 
    \mathrlap{
      \scalebox{.7}{
        \color{orangeii}
        \bf tensor product
      }
    }
    && 
    \otimes
    \\
    \mathscr{V}' 
    \ar[
      rr,
      "{ \phi' }"{description}
    ] 
    && \mathscr{W}'
  \end{tikzcd}
  }
\!\!\!\!  \right\}
$$
The refinement of this situation to Hilbert spaces serves to provide the: 
\smallskip 
\begin{itemize}[leftmargin=.5cm]
\item[4.] ({\bf Born rule}) 
The further $\dagger$-structure on complex {\it Hilbert} spaces (sending linear operators to their adjoints) reflects the 
Hermitian inner product in quantum states and hence the probabilistic nature of quantum physics.
\end{itemize}

So far this concerns coherent quantum processes on pure states, undisturbed by interaction with (such as in quantum measurement) 
or control by (such as in quantum state preparation) a classical environment. Such classical$\leftrightarrow$quantum interactions 
are instead reflected in {\it parameterized} quantum systems:

\medskip

\noindent
{\bf Parameterized quantum processes and Bundle categories.} More recent developments \cite[\S 6.1]{Schreiber14}\cite{RiosSelinger18} \cite{FuKishidaSelinger20} show that the logic of 
interaction between quantum systems and classical environments is essentially the ``internal logic'' of categories of Hilbert 
space-{\it bundles} over varying classical base spaces \cite[\S\S 3-4]{QS}:

\smallskip 
\begin{itemize}[leftmargin=.5cm]
  \item [1.] ({\bf Many worlds}) The fact that different quantum states, or even different Hilbert spaces, may be seen in 
  different classical worlds (in different classical parameter configurations) is reflected in these states being {\it sections}
  of {\it bundles} of state spaces over classical parameter spaces.
  \item [2.] ({\bf Quantum compulsion}) in which picture the above superposition principle re-appears as the fact that along 
  maps of parameter base spaces with finite pre-images, the left- and right-pushforward of vector bundles coincide (``ambidexterity'').
  \item [3.] ({\bf Quantum measurement}) In particular, the push-pull of vector bundles along a map with finite pre-image $B$  
  is a (co)monadic operation whose (co)unit reflects exactly the collapse of quantum states in the Hilbert space $\mathrm{Q}B$ 
  branched according to the classical measurement outcome in $B$.
  \item [4.] ({\bf Quantum state preparation}) and, dually, classical quantum state preparation is reflected in the operation 
  where dependent on a classical parameter $b \in B$ we pick in the linear fiber $\mathbb{C}B$ the corresponding basis vector.
\end{itemize}

\smallskip 
In particular, in this category of linear bundle types a classical logical conjunction is restored, the cocartesian coproduct operation $\sqcup$ on bundles, 
which interprets the logical connective of having some quantum states in some possible world {\it or} other quantum states in another possible world:
\vspace{-2mm}
$$
  \adjustbox{raise=8pt}{
  \scalebox{.8}{
  \begin{tabular}{l}
    Backdrop for {\it parameterized} quantum 
    \\
    processes featuring:
    many worlds, 
    \\
    quantum measurement \& state collapse,
    \\
    quantum state preparation, etc.
  \end{tabular}
  }
  }
  \hspace{1.5cm}
  \overset{
    \mathclap{
      \raisebox{7pt}{
        \scalebox{.7}{
          \color{darkblue}
          \bf
          \def\arraystretch{.8}
          \begin{tabular}{c}
            {\bf \color{purple}cocartesian monoidal} 
            category
            \\
            of complex vector bundles
          \end{tabular}
        }
      }
    }
  }{
    \big(
    \mathrm{Bun}_{\mathbb{C}},
    \sqcup
    \big)
  }
  \;\;\;\;
  =
  \;\;\;\;
  \left\{
  \adjustbox{raise=3pt}{
  \begin{tikzcd}[column sep=50pt]
    \overset{
      \mathclap{
        \scalebox{.7}{
          \color{darkblue}
          \bf
          \def\arraystretch{.7}
          \begin{tabular}{c}
            vector
            \\
            bundle
          \end{tabular}
        }
      }
    }{
      \mathscr{V}
    }
    \ar[
      rr,
      "{ \phi }"{description},
      "{
        \scalebox{.7}{
          \color{greenii}
          \bf
          \def\arraystretch{.9}
          \begin{tabular}{c}
            fiberwise 
            \\
            linear map
          \end{tabular}
        }
      }"{yshift=3pt}
    ]
    \ar[
      d,
      "{
        \scalebox{.7}{
          \color{orangeii}
          \bf
          \def\arraystretch{.8}
          \begin{tabular}{l}
            bundle
            \\
            projection
          \end{tabular}
        }
      }"{xshift=-5pt}
    ]
    && 
    \mathscr{V}'
    \ar[d]
    \\
    \underset{
      \mathclap{
        \raisebox{-4pt}{
          \scalebox{.7}{
          \color{darkblue}
          \bf
          \def\arraystretch{.8}
          \begin{tabular}{c}
            base 
            \\
            space
          \end{tabular}
          }
        }
      }
    }{
      X
    }
    \ar[
      rr,
      "{ f }"{description},
      "{  
        \scalebox{.7}{
          \color{greenii}
          \bf
          classical map
        }
      }"{swap, yshift=-3pt}
    ]
    &&
    X'
  \end{tikzcd}
   \!\!\!\!  }
  \right\}
$$
\vspace{-2mm} 

In this context, a quantum state is no longer just an element of a fixed Hilbert space, but is a {\it section} of a bundle of such spaces (in the language of type theory: a ``dependent term of a dependent linear type'') assigning to each parameter value (to each possible classical world) the quantum state seen for that parameter value (in that world).
$$
  \mathllap{
    \scalebox{.7}{
      \begin{tabular}{l}
        In parameterized quantum theory
        \\
        a quantum state is a {\it section}
        \\
        of a bundle of quantum state spaces
        \\
        assigning to each parameter value
        \\
        (to each possible classical world)
        \\
        the quantum state for that value
        \\
        (as seen in that particular world).
      \end{tabular}
    }
  }
  \hspace{1cm}
  \begin{tikzcd}
    &[-26pt]
    &
    \mathbb{C}^{n_0}
    \ar[dd]
    &
    [-26pt]
    \\[-29pt]
    &&&
    \mathbb{C}^{n_1}
    \ar[dd]
    \\[+15pt]
    \mathllap{\big\{}0
    \ar[dr, phantom, "{ , }"]
    \ar[
      rr, equals, shorten >=5pt
    ]
    \ar[
      uurr,
      crossing over,
      "{
        \vert \psi_0 \rangle
      }"{sloped}
    ]
    &[-26pt]
    &
    \mathllap{\big\{}0
    \ar[dr, phantom, "{ , }"]
    &[-26pt]
    \\[-29pt]
    & 
    1\mathrlap{\big\}}
    \ar[
      rr, equals, shorten <=5pt
    ]
    \ar[
      uurr,
      crossing over,
      "{
        \vert \psi_1 \rangle
      }"{sloped}
    ]
    &
    &
    1\mathrlap{\big\}}
  \end{tikzcd}
$$

\noindent Crucially, in the category of parameterized bundles of quantum states, {\it quantum measurement} 
and the ensuing {\it quantum state collapse} is  naturally reflected by linear projection maps parameterized over the set of possible
measurement outcomes, and dually for quantum state preparation:
\vspace{-2mm} 
$$
  \scalebox{.7}{
    \begin{tabular}{l}
      Process of quantum measurement of
      \\
      qubits 
      $q_0 \vert0\rangle + q_1 \vert 1 \rangle \,\in\, \mathbb{C}^2 \simeq \mathbb{C}[\{0,1\}]$
      \\
      as a morphism of 
      Hilbert space bundles 
      \\
      over spaces of possible classical worlds
    \end{tabular}
  }
  \hspace{.4cm}
  \begin{tikzcd}[row sep=15pt,
    column sep=50pt
  ]
    \mathbb{C}^2
    \ar[
      rr,
      "{
        \vert 0 \rangle \, \langle 0 \vert
      }"{description}
    ]
    \ar[dd]
    &[-45pt]
    &
    \mathbb{C}
    \ar[dd]
    &[-45pt]
    \\[-20pt]
    & 
    \mathbb{C}^2
    \ar[
      rr, 
      crossing over, 
      "{
        \vert 1 \rangle \, \langle 1 \vert
      }"{description, pos=.45}
    ]
    &&
    \mathbb{C}
    \ar[
      dd
    ]
    \\[+20pt]
    \mathllap{\big\{}0
    \ar[rr, shorten >=6pt, |->]
    \ar[dr, phantom, "{,}"]
    &&
    \mathllap{\big\{}0
    \ar[dr, phantom, "{,}"]
    \\[-20pt]
    & 
    1\mathrlap{\big\}}
    \ar[rr, shorten <=5pt, |->]
    \ar[
      from=uu,
      crossing over
    ]
    &&
    1\mathrlap{\big\}}
  \end{tikzcd}
  \hspace{.5cm}
  \begin{tikzcd}[row sep=15pt,
    column sep=50pt
  ]
    \mathbb{C}
    \ar[dd]
    \ar[
      rr, 
      "{
        \vert\psi_0\rangle
      }"{description, pos=.35}
    ]
    &[-45pt]
    &
    \mathbb{C}^2
    \ar[dd]
    \\[-20pt]
    &
    \mathbb{C}
    \ar[
      ur,
      end anchor={[yshift=-2pt]},
      "{
        \vert\psi_1\rangle
      }"{description, sloped, pos=.45}
    ]
    \\[+20pt]
    \mathllap{\big\{}0
    \ar[dr, phantom, "{,}"]
    \ar[rr, |->]
    &&
    \{\ast\}
    \\[-20pt]
    &
    1\mathrlap{\big\}}
    \ar[
      from=uu,
      crossing over
    ]
    \ar[
      ur, 
      |->,
      shorten <=5pt, 
      end anchor={[yshift=-2pt]}
    ]
  \end{tikzcd}
  \scalebox{.7}{
    \begin{tabular}{l}
      State preparation
      \\
      conditioned on 
      \\
      classical parameters
    \end{tabular}
  }
$$

\vspace{-2mm}
\noindent 
In fact (\cite[\S 4]{QS}): the quantum measurement process is the {\it counit} of the {\it base change comonad}
$\Box_{\{0,1\}} \,:\defneq\, p^* p_*$ on the {\it slice category} of $\mathrm{Bun}_{\mathbb{C}}$ over the 0-bundle $0_{\{0,1\}}$:
\vspace{-2mm}
$$
  \begin{tikzcd}[row sep=small]
  &
  0_{\{0,1\}}
  \ar[rr, "{ p }"]
  &&
  0_{\{\ast\}}
  &
  \scalebox{.7}{
    \color{greenii}
    \def\arraystretch{.9}
    \begin{tabular}{c}
      map of vector bundles
      \\
      (here: 0-bundles over sets)
    \end{tabular}
  }
  \\
  \scalebox{.7}{
    \def\arraystretch{.7}
    \color{orangeii}
    \bf
    \begin{tabular}{c}
      comonad expressing
      \\
      quantum measurement
      \\
      as a computational effect
    \end{tabular}
  }
  &
  \big(
    \mathrm{Bun}_{\mathbb{C}}
  \big)_{\!\!\!\big/0_{\{0,1\}}}
    \ar[
        out=180+38, 
        in=180-38, 
        decorate,
        looseness=3.5,
        shift left=10pt,
        "{ \Box_{\{0,1\}} }"{description}
    ]
  \ar[
    rr,
    shift left=12pt,
    "{ p_! }"
  ]
  \ar[
    from=rr,
    "{ p^\ast }"{description}
  ]
  \ar[
    rr,
    shift right=12pt,
    "{ p_\ast }"{swap}
  ]
  \ar[
    rr,
    phantom,
    shift left=7pt,
    "{
      \scalebox{.6}{$\bot$}
    }"
  ]
  \ar[
    rr,
    phantom,
    shift right=7pt,
    "{
      \scalebox{.6}{$\bot$}
    }"
  ]
  &&
  \big(
    \mathrm{Bun}_{\mathbb{C}}
  \big)_{\!\!\big/0_{\{\ast\}}}
  &
  \scalebox{.7}{
    \color{greenii}
    \def\arraystretch{.9}
    \begin{tabular}{c}
      induced base change adjunction
      \\
      between slice categories (of
      \\
      bundles of vector bundles!)
    \end{tabular}
  }
  \end{tikzcd}
$$

\vspace{-5mm}
\noindent 
It is such base change adjunctions between slices of $\mathrm{Bun}_{\mathbb{C}}$
which interpret the dependent linear type inference rules of quantum programming languages like {\tt Quipper} \cite{RiosSelinger18}\cite{FuKishidaSelinger20} and {\tt LHoTT} \cite[\S 3]{QS}.

\medskip

\noindent
{\bf Unification of entangled and parameterized quantum information.}
Therefore one is naturally led to wonder about {\it amalgamating} these two fragments of quantum information theory:
\begin{itemize}[leftmargin=.5cm]
 \item[1.]  the non-cartesian monoidal tensor structure on plain vector spaces encoding pure quantum phenomena such as entanglement,
 \item[2.] the cocartesian monoidal structure of vector bundles encoding quantum/classical phenomena such as state collapse upon quantum measurement,
\end{itemize}
by coupling these two theory sectors along their common core of vector spaces of quantum states.

\medskip
\noindent
In the language of category theory, such an amalgamation of two objects along a common core would be called a {\it pushout}
(to be abbreviated ``po''), which in the present case would mean to ask for the {\it universal} way of completing the 
following (for the moment: schematic) diagram to a commuting square, in a suitable sense:
\begin{equation}
  \label{TheSoughtAfterPushout}
  \hspace{-2cm} 
  \begin{tikzcd}[column sep=large]
    \mathllap{
      \scalebox{.8}{
        \def\arraystretch{.8}
        \begin{tabular}{l}
          {\color{purple}\bf  Pure quantum phenomena}:
          \\
          no-cloning, {\bf entanglement},...
        \end{tabular}
      }
    }
    \big(
      \mathrm{Mod}_{\mathbb{C}}
      ,
      \otimes
    \big)
    \ar[rr, dashed]
    \ar[
      drr,
      phantom,
      "{
        \scalebox{.7}{(po)}
      }"
    ]
    &&
   \colorbox{lightgray}{ (??)}
     \mathrlap{  
       \hspace{12pt}
        \scalebox{.8}{
          \bf
          \begin{tabular}{l}
            {\color{purple}
             Parameterized quantum phenomena
            }
            \\
            entanglement of sections
          \end{tabular}
        }
      }
    \\
    \mathllap{
      \raisebox{0pt}{
      \scalebox{.8}{
        \def\arraystretch{.8}
        \begin{tabular}{l}
          {\color{purple}\bf Core quantum phenomena}:
          \\
          superposition principle
        \end{tabular}
      }
    }
    }
    \mathrm{Mod}_{\mathbb{C}}
    \ar[rr]
    \ar[u]
    &&
    {
      \big(
        \mathrm{Bun}_{\mathbb{C}}, \sqcup\big)
    }
    \ar[u, dashed]
     \mathrlap{  
        \scalebox{.8}{
          \begin{tabular}{l}
            {\color{purple}\bf Classical $\leftrightarrow$ Quantum phenomena}:
            \\
            parameterized quantum states:
            {\bf sections}
            \\
            state collapse/preparation in many worlds.
          \end{tabular}
        }
      }
  \end{tikzcd}
\end{equation}

\noindent
Essentially the following natural {\bf Question} was recently raised in \cite[p. 1]{FreedmanHastings23}:

\noindent
\begin{itemize}[leftmargin=.8cm]
\item [{\bf (i)}] How to make this precise?
\item [{\bf (ii)}] What then is the pushout? 
\item [{\bf (iii)}] What is its import on quantum information theory?
\end{itemize}

\medskip

\noindent
Here we offer an {\bf Answer.} Informally, our answer says:

\begin{center}
\fboxrule=1pt \colorbox{lightgray}{  
\small 
\begin{tabular}{c}
{\bf The amalgamation of} 
\\
    {\bf the entanglement tensor product} structure on Hilbert spaces 
  \\
  {\bf with the parameterized coproduct} structure on Hilbert bundles 
  \\
  {\bf is the} {\it external} {\bf tensor product} structure on Hilbert bundles.
  \end{tabular} 
  }
\end{center}

\medskip

\noindent
{\bf External tensor product.}
Here the {\it external tensor product of vector bundles} or that induced on their K-theory classes
\cite[\S 2.6]{Atiyah67}\cite[p. 19]{Bott69}\cite[\S 13.51]{Switzer75}\cite[\S 4.9]{Karoubi78}
(cf. also \cite[p. 84]{GHV73}\cite[p. 2]{Lyubashenko01}\cite[p. 7]{Shulman13}) \newline 
forms the cartesian product of parameter base spaces and over each pair of parameter values assigns 
the tensor product of the corresponding Hilbert spaces:
$$
  \begin{tikzcd}[column sep=huge]
    \mathllap{
      \scalebox{.7}{
        \color{darkblue}
        \bf
        \def\arraystretch{.9}
        \begin{tabular}{c}
          External tensor of
          \\
          vector bundles
        \end{tabular}
      }
    }
    \mathscr{V} \boxtimes \mathscr{W}
    \ar[
      d,
      shorten <=-5pt
    ]
    &
    \mathscr{V}_{\! x} 
    \otimes
    \mathscr{W}_{\! y}
    \ar[
      l, 
      hook', 
      "{ 
        \scalebox{.7}{
          \color{greenii}
          \bf
          fiber
        } 
      }"{swap}
    ]
    \ar[
      d,
      shorten <=-5pt
    ]
    \mathrlap{
      \scalebox{.7}{
        \color{darkblue}
        \bf
        \def\arraystretch{.9}
        \begin{tabular}{c}
          Tensor product
          \\
          of vector spaces
        \end{tabular}
      }
    }
    &[1cm]
      \scalebox{.7}{
        \color{orangeii}
        \bf
        \def\arraystretch{.9}
        \begin{tabular}{c}
          Entanglement of
          \\
          quantum states
        \end{tabular}
      }
    \\
    \mathllap{
      \scalebox{.7}{
        \color{darkblue}
        \bf
        \def\arraystretch{.9}
        \begin{tabular}{c}
          Classical product 
          \\
          of base spaces
        \end{tabular}
      }
    }
    X \times Y
    &
    \{(x,y)\}
    \ar[
      l, 
      hook',
      "{ 
        \scalebox{.7}{
          \color{greenii}
          \bf
          base point
        } 
      }"{swap},
    ]
    \mathrlap{
      \scalebox{.7}{
        \color{darkblue}
        \bf
        \def\arraystretch{.9}
        \begin{tabular}{c}
          Pair of 
          \\
          base points
        \end{tabular}
      }
    }
    &
      \scalebox{.7}{
        \color{orangeii}
        \bf
        \def\arraystretch{.9}
        \begin{tabular}{c}
          Pairing of 
          \\
          classical parameters
        \end{tabular}
      }
  \end{tikzcd}
$$
While, as a mathematical construction, the external tensor product (specifically in topological K-theory) is well-known, its 
relevance to quantum theory has been considered only recently \cite{Mera20}\cite[\S 3.1]{QS}.

A key point of the external tensor product is that it {\it distributes} over the disjoint union of bundles (the cocartesian product):
$$
  \mathscr{V}_X
  \boxtimes
  \big(
    \mathscr{W}_{Y}
    \sqcup
    \mathscr{W}'_{Y'}
  \big)
  \;\;
  \simeq
  \;\;
  \big(
  \mathscr{V}_X
  \,\boxtimes\,
  \mathscr{W}_Y
  \big)
  \,\sqcup\,
  \big(
  \mathscr{V}_X
  \,\boxtimes\,
  \mathscr{W}'_{Y'}
  \big)
  \,.
$$
In this way, $\big(\mathrm{Bun}_{\mathbb{C}},\,\sqcup,\,\boxtimes\big)$ forms what is called a {\it distributive category} (Def. \ref{CategoriesOfCategories}):
$$
  \adjustbox{raise=8pt}{
  \scalebox{.8}{
  \begin{tabular}{l}
    Backdrop for 
    \\
    {\it parameterized entangled}
    \\
    quantum    
        processes, featuring 
        \\
    entanglement of sections
   \end{tabular}
  }
  }
  \hspace{1.3cm}
  \overset{
    \mathclap{
      \raisebox{4pt}{
        \scalebox{.7}{
          \color{darkblue}
          \bf
          \def\arraystretch{.9}
          \begin{tabular}{c}
            {\color{purple}Distributive monoidal} category 
            \\
            of complex vector bundles
            \\
            with external tensor product
          \end{tabular}
        }
      }
    }
  }{
  \big(
    \mathrm{Bun}_{\mathbb{C}}
    ,\,
    \sqcup
    ,\,
    \boxtimes
  \big)
  }
  \qquad
  \defneq
  \;\;\;
  \left\{\!\!\!
  \adjustbox{raise=5pt}{
  \begin{tikzcd}[
    row sep=-7pt, column sep=large
  ]
    &
    \overset{
      \mathclap{
        \raisebox{4pt}{
          \scalebox{.7}{
            \color{darkblue}
            \bf
            \def\arraystretch{.8}
            \begin{tabular}{c}
              vector
              \\
                            bundle
            \end{tabular}
          }
        }
      }
    }{
      \mathscr{V}
    }
    &&
    \mathscr{V'}
    \\
    &
    \mathllap{
      \scalebox{.7}{
        \color{orangeii}
        \bf
        external tensor
      }
      \,
    }
    \boxtimes
    \ar[dddl]
    \ar[
      rr,
      "{
        \scalebox{.7}{
          \color{greenii}
          \bf
          \def\arraystretch{.8}
          \begin{tabular}{c}
            fiberwise
            \\
            linear map
          \end{tabular}
        }
      }"
    ]
    &&
    \boxtimes
    \ar[dddl]
    \\
    &
    \mathscr{W}
    &&
    \mathscr{W}'
    \\[25pt] 
    X 
    &&
    X'
    \\
    \times
    \ar[
      rr,
      "{
        \scalebox{.7}{
          \color{greenii}
          \bf
          classical map
        }
      }"{swap}
    ]
    &&
    \times
    \mathrlap{
      \,
      \scalebox{.7}{
        \color{orangeii}
        \bf
        classical product
      }
    }
    \\
    Y && Y'
  \end{tikzcd}
  }
\!\!\!\!  \right\}.
$$

\medskip

\noindent
{\bf Universal property of external tensor product of vector bundles over discrete spaces.} Consider the special case of vector 
bundles over discrete spaces, i.e., over plain sets. (While just a small special case of the general mathematical notion, 
this already captures all  parameterization of quantum processes  considered in contemporary quantum information theory.)
It is useful to understand this category as the unification (``Grothendieck construction'', Def. \ref{GrothendieckConstruction}) of all categories of vector 
bundles over fixed sets, which in turn are usefully understood via their fiber-assigning functors:
 
\begin{equation}
  \hspace{-1.5cm} 
  \begin{tikzcd}[sep=0pt]
  \mathrm{Mod}_{\mathbb{C}}
  \ar[rr, hook, "{ \iota }"]
  &&
  \overset{
    \mathrm{Fam}_{\mathbb{C}}
  }{
  \overbrace{
  \int_{X \in \mathrm{Set}}
  \mathrm{Mod}^X_{\mathbb{C}}
  }
  }
  \\
 \scalebox{0.85}{$ \mathscr{V} $}
  &\mapsto&
\scalebox{0.85}{$  \left[\!\!\!
  \begin{array}{c}
    \mathscr{V}
    \\
    \downarrow
    \\
    \mathrm{pt}
  \end{array}
  \!\!\!\right]
  $}
  \end{tikzcd}
  \hspace{1cm}
  \mbox{where}
  \hspace{1cm}
  \begin{tikzcd}[row sep=-3pt, column sep=small]
    \mathrm{Set}^{\mathrm{op}}
    \ar[rr]
    &&
    \mathrm{Cat}
    \\
    X 
    \ar[d, "{f}"]
      &\mapsto&
    \mathrm{Func}(
      X,
      \,
      \mathrm{Mod}_{\mathbb{C}}
    )
    \ar[
      from=d,
      "{
        f^\ast 
      }"{swap}
    ]
    \mathrlap{
    \;\,
    \defneq
    \;\,
    \mathrm{Mod}^X_{\mathbb{C}}
    }
    \\[20pt]
    Y 
      &\mapsto&
    \mathrm{Func}(
      Y,
      \,
      \mathrm{Mod}_{\mathbb{C}}
    )
    \mathrlap{
    \;\,
    \defneq
    \;\,
    \mathrm{Mod}^Y_{\mathbb{C}}\,.
    } 
  \end{tikzcd}
\end{equation}

\vspace{-1mm} 
\noindent Notice that every bundle over a discrete space $X$ is the coproduct of its 
restrictions to the points in the base space:
$$
  \mathscr{V}_X
  \;\;\;
  =
  \;\;\;
  \begin{array}{c}
    \mathscr{V}_x
    \\
    \downarrow
    \\
    \{x\}
  \end{array}
  \;\;
    \sqcup
  \;\;
  \begin{array}{c}
    \mathscr{V}_y
    \\
    \downarrow
    \\
    \{y\}
  \end{array}
  \;\;
    \sqcup
  \;\;
  \begin{array}{c}
    \mathscr{V}_z
    \\
    \downarrow
    \\
    \{z\}
  \end{array}
  \;\;
    \sqcup
  \;\;
  \cdots
  \,.
$$
In fact, vector bundles over sets are the {\it free coproduct completion} (Ex. \ref{CategoriesOfIndexedSetsOfObjects})
of the plain category of vector spaces.
 But this means that on such 
 bundles the external tensor product is {\it completely characterized} by these two properties:
 \vspace{.1mm} 
\begin{itemize}[leftmargin=.5cm]
  \item[\bf 1.] 
  Over singletons, it reduces to the ordinary tensor product:
  $\quad 
    \begin{array}{c}
      \mathscr{V}
      \\
      \downarrow
      \\
      \mathrm{pt}
    \end{array}
      \;
      \boxtimes
      \;
    \begin{array}{c}
      \mathscr{W}
      \\
      \downarrow
      \\
      \mathrm{pt}
    \end{array}
    \;
      \simeq
    \;
    \begin{array}{c}
      \mathscr{V} 
      \otimes
      \mathscr{W}
      \\
      \downarrow
      \\
      \mathrm{pt}
    \end{array}
  $
  
  \vspace{.3cm}
  
  \item[\bf 2.] It distributes over coproducts:
  $\qquad 
   \begin{array}{c}
      \mathscr{V}
      \\
      \downarrow
      \\
      \{x\}
    \end{array}
    \;
    \boxtimes
    \;
    \left(\!\!
    \begin{array}{c}
      \mathscr{W}
      \\
      \downarrow
      \\
      \{y\}
    \end{array}
    \;
      \sqcup
    \;
    \begin{array}{c}
      \mathscr{W}'
      \\
      \downarrow
      \\
      \{y'\}
    \end{array}
    \!\!\right)
    \;
      =
    \;
    \left(\!\!
    \begin{array}{c}
      \mathscr{V}
      \\
      \downarrow
      \\
      \{x\}
    \end{array}
    \,\boxtimes\,
    \begin{array}{c}
      \mathscr{W}
      \\
      \downarrow
      \\
      \{y\}
    \end{array}
   \!\! \right)
    \;\sqcup\;
    \left(\!\!
    \begin{array}{c}
      \mathscr{V}
      \\
      \downarrow
      \\
      \{x\}
    \end{array}
    \,\boxtimes\,
    \begin{array}{c}
      \mathscr{W}'
      \\
      \downarrow
      \\
      \{y'\}
    \end{array}
    \!\!\right).
  $
\end{itemize}

\medskip

This characterization may be reformulated (we make this precise in \cref{PushoutPropertyForVectorBundlesOverDiscreteSpaces}) 
as saying that the following diagram is a {\it pushout} in a suitable category of plain, monoidal, cocartesian and distributive 
categories, whose morphisms are strong monoidal functors with respect to the monoidal structure present on their domain category
(Thm. \ref{ThePushoutTheoremOverSets}):
\vspace{-2mm} 
\begin{equation}
  \label{TheFirstPushoutDiagramInIntroduction}
  \begin{tikzcd}[row sep=small, column sep=large]
    \big(
      \mathrm{Mod}_{\mathbb{C}}
      ,\,
      \otimes
    \big)
    \ar[r, hook, "{ \iota }"]
    \ar[dr, phantom, "{ \scalebox{.7}{(po)} }"]
    &
    \colorbox{lightgray}{$    
    \big(\,
    \mathrm{Fam}_{\mathbb{C}}
    ,\,
    \sqcup
    ,\,
    \boxtimes
    \big)
    $}
    \\
    \mathrm{Mod}_{\mathbb{C}}
    \ar[u]
    \ar[r, hook, "{\iota}"]
    &
    \big(\,
    \mathrm{Fam}_{\mathbb{C}}
    ,\,
    \sqcup
    \big)
    \ar[u]
  \end{tikzcd}
\end{equation}
\vspace{-2mm} 

\noindent Therefore, this is a first answer to the question \eqref{TheSoughtAfterPushout} for the case of discrete parameter spaces.
While this subsumes all of the contemporary quantum information theory, we naturally want to go further.

\medskip
\noindent
{\bf External tensor product of group representations.}
Some parameters in quantum physics are known not to form discrete sets but to form homotopy 1-types, namely 
groupoids (exposition in \cite{Weinstein96}); and the Hilbert bundles over these parameter spaces have {\it monodromy}. 
A famous example are bundles of conformal blocks over configuration spaces of points and equipped with 
the Knizhnik-Zamolodchikov connection, which are thought to be the bundles of Hilbert spaces for
{\it anyons} (see \cite{TQP} for extensive references to the literature and further discussion related to our perspective here):
\vspace{-2mm} 
$$
  \begin{tikzcd}[row sep=3pt]
    \overset{
      \mathclap{
        \raisebox{3pt}{
        \scalebox{.7}{
          \color{darkblue}
          \bf
          \def\arraystretch{.8}
          \begin{tabular}{c}
            Hilbert bundle
            \\
            of anyon states
            \\
            (``conformal blocks'')
          \end{tabular}
        }
        }
      }
    }{
      \mathscr{H}
    }
    \ar[d]
    \\
    \underset{
      \mathclap{
        \scalebox{.7}{
          \color{darkblue}
          \bf
          \def\arraystretch{.8}
          \begin{tabular}{c}
            configuration
            \\
            space
          \end{tabular}
        }
      }
    }{
    \underset{
      \mathclap{
        \scalebox{.7}{$
          \{1,\! \cdots,\! n\}
        $}
      }
    }{\mathrm{Conf}}\big(\mathbb{R}^2\big)
    }
    \ar[
      r, 
      phantom, 
      "{ \simeq }",
      "{ \scalebox{.7}{$\mathrm{whe}$} }"{yshift=-6pt}
    ]
    &[-15pt]
    \underset{
      \mathclap{
        \raisebox{-12pt}{
        \scalebox{.7}{
          \color{darkblue}
          \bf
          \def\arraystretch{.8}
          \begin{tabular}{c}
            Eilenberg-MacLane
            \\
            space
          \end{tabular}
        }
        }
      }
    }{
    K\big(
     \mathrm{PBr}(N)
     ,\,
     1
    \big)
    }
    \;\;
    \underset{
      \mathrm{whe}
    }{\simeq}
    \;\;
    \underset{
      \mathclap{
        \hspace{-4pt}
        \raisebox{-13pt}{
        \scalebox{.7}{
          \color{darkblue}
          \bf
          \def\arraystretch{.8}
          \begin{tabular}{c}
            classifying
            \\
            space
          \end{tabular}
        }
        }
      }
    }{
    B\big(
      \overset{
        \mathclap{
          \scalebox{.7}{
            \color{darkblue}
            \bf
            \def\arraystretch{.8}
            \begin{tabular}{c}
              pure braid
              \\
              group
            \end{tabular}
          }
        }
      }{
        \mathrm{PBr}(N)
      }
    \big)
    \mathrlap{\,.}
    }
  \end{tikzcd}
$$
An elegant way to encode a {\it flat connection} on such bundles is to consider the {\it fundamental groupoid} 
\eqref{FundamentalGroupoid} of their base space -- which for connected spaces $X$ is equivalently the delooping groupoid  $\mathbf{B}\pi_1(X)$ \eqref{DeloopingGroupoids} with a single object and the elements of the fundamental group $\pi_1(X)$ as morphisms. 
Then a flat connection over $X$ is equivalently its monodromy functor $\mathbf{B}\pi_1(X) \to \mathrm{Mod}_{\mathbb{C}}$ (\cite{SchreiberWaldorf09}\cite{Dumitrascu10}), also known as a ``local system'' 
(\cite[p. 58]{Spanier66}\cite[\S I.1]{Deligne70}\cite{LibgoberYuzvinsky00}\cite[\S I 9.2.1]{Voisin02}\cite[\S 2.5]{Dimca04}):
\vspace{-2mm} 
$$
  \begin{tikzcd}[
    row sep=0pt,
    column sep=13pt
  ]
    \mbox{\bf{B}}\big(
      \mathrm{PBr}(N)
    \big)
    \ar[
      rr,
      "{
        \def\arraystretch{.9}
        \scalebox{.7}{
          \color{greenii}
          \bf
          \def\arraystretch{.9}
          \begin{tabular}{c}
            Knizhnik-Zamolodchikov
            \\
            monodromy
          \end{tabular}
        }
      }"
    ]
    &&
    \mathrm{Mod}_{\mathbb{C}}
    \\[+10pt]
    \scalebox{.65}{$\{1, \cdots, N\}$}
    \ar[
        out=180+38, 
        in=180-38, 
        decorate,
        looseness=3.5,
        shift right=-4pt,
        "{ \mathrm{braid} }"{description}
    ]
    &\hspace{-17pt}\mapsto&
    \mathscr{H}
    \ar[
        out=180+38, 
        in=180-38, 
        decorate,
        looseness=3.5,
        shift right=4pt,
        "{ \mathrm{unitary} }"{description}
    ]
  \end{tikzcd}
  \hspace{1cm}
  \raisebox{-1.2cm}{
    \includegraphics[width=7cm]{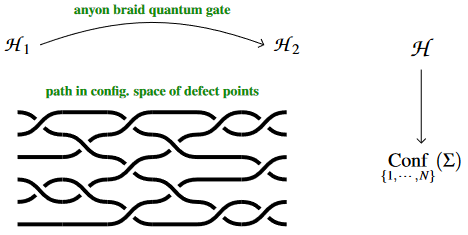}
  }
$$
Generally/equivalently, for $G$ any group, a complex linear $G$-representation is a functor from $\mathbf{B}G$ to $\mathrm{Mod}_{\mathbb{C}}$:
\vspace{-2mm} 
\begin{equation}
  \label{GRepAsModBGInIntroduction}
  G \mathrm{Rep}_{\mathbb{C}}
  \;\;
  \simeq
  \;\;
  \mathrm{Mod}_{\mathbb{C}}^{\mathbf{B}G}
  \;\;
  =
  \;\;
  \Bigg\{\!\!\!\!
  \adjustbox{raise=6pt}{
  \begin{tikzcd}[row sep=small]
    \mbox{\bf{B}}G
    \ar[rr, "{ \mathscr{V}_{(\mbox{-})} }"]
    &&
    \mathrm{Mod}_{\mathbb{C}}
    \\[-10pt]
 \scalebox{0.8}{$   \mathrm{pt}   
    \ar[
        out=180+38, 
        in=180-38, 
        decorate,
        looseness=3.5,
        shift right=4pt,
        "{ g }"{description}
    ]
    $}
    &\mapsto&
\scalebox{0.8}{$    \mathscr{V}
    \ar[
        out=180+38, 
        in=180-38, 
        decorate,
        looseness=3.5,
        shift right=7pt,
        "{ \rho_g }"{description}
    ]    
    $}
  \end{tikzcd}
  }
  \!\!\!\!\Bigg\}.
\end{equation}

\vspace{-2mm} 
\noindent There is a classical notion of {\it external tensor product of group representations} (e.g. \cite[Exer. 2.36]{FultonHarris91}), 
which in this groupoid-picture is the {\it cup-tensor product}
$$
  G,\, G'
  \,\in\,
  \mathrm{Grp}
  \hspace{1cm}
    \vdash
  \hspace{1cm}
  \begin{tikzcd}[row sep=20pt]
    G \mathrm{Rep}_{\mathbb{C}}
    \,\times\,
    G' \mathrm{Rep}_{\mathbb{C}}
    \ar[rr, "{ \boxtimes }"]
    \ar[d, phantom, "{\simeq}"{rotate=-90}]
    &&
    (G \times G') \mathrm{Rep}_{\mathbb{C}}
    \ar[d, phantom, "{\simeq}"{rotate=-90}]
    \\[-15pt]
    \mathrm{Mod}_{\mathbb{C}}^{\mathbf{B}G}
    \,\times\,
    \mathrm{Mod}_{\mathbb{C}}^{\mathbf{B}G'}
    \ar[rr]
    &&
    \mathrm{Mod}_{\mathbb{C}}^{\mathbf{B}(G \times G')}
  \end{tikzcd}
$$
$$
  \begin{tikzcd}[column sep=40pt]
    \mathllap{
      \mathscr{V}
      \,\boxtimes\,
      \mathscr{W}
      \;:\;\;
    }
    \mbox{\bf{B}}(G \times G')
    \;\simeq\;
    \mbox{\bf{B}}G
    \times
    \mbox{\bf{B}}G'
    \ar[rr, " \mathscr{V}_{(\mbox{-})} \times \mathscr{W}_{(\mbox{-})} "]
    &&
    \mathrm{Mod}_{\mathbb{C}}
    \times
    \mathrm{Mod}_{\mathbb{C}}
    \ar[r, "{ \otimes }"]
    &
    \mathrm{Mod}_{\mathbb{C}}
  \end{tikzcd}
$$
This system of group-wise external tensor products is again unified 
on the Grothendieck construction
$$
  \mathrm{Rep}_{\mathbb{C}}
  \;:\defneq\;
  \int_{G \in \mathrm{Grp}}
  \mathrm{Mod}_{\mathbb{C}}^{\mathbf{B}G}
$$
into a single functor
$$
  \begin{tikzcd}
    \mathrm{Rep}_{\mathbb{C}}
    \times
    \mathrm{Rep}_{\mathbb{C}}
    \ar[r, "{ \boxtimes }"]
    &
    \mathrm{Rep}_{\mathbb{C}}
    \,.
  \end{tikzcd}
$$
analogous to the external tensor product of bundles over sets, from \eqref{TheFirstPushoutDiagramInIntroduction}.

\medskip

\noindent
{\bf Universal property of the external tensor product of flat vector bundles.}
To better understand this situation of group representations over varying groups, we may without restriction focus here on {\it skeletal groupoids} \eqref{SkeletalGroupoidDeclaration}, namely disjoint unions of delooping groupoids, which we may think of as sets each of whose elements is equipped 
with a group of automorphisms:
\vspace{-2mm} 
$$
  \mathrm{Grpd}_{\mathrm{skl}}
  \;\;\simeq\;\;
  \int_{ S \in \mathrm{Set} }
  \mathrm{Grp}^{S}
  \;\;\;\;
  \simeq
  \;\;\;
  \big\{
    \mathbf{B}G_1
    \,\sqcup\,
    \mathbf{B}G_2
    \,\sqcup\,
    \mathbf{B}G_3
    \,\sqcup\,
   \cdots
  \big\}
  \,.
$$
Flat vector bundles
on arbitrary base spaces $X$ are equivalently functors on such skeletal groupoids (\cite{SchreiberWaldorf09}\cite{Dumitrascu10}),
namely on the skeletization of their fundamental groupoid, 
which has one component $\mathbf{B}\pi_1(X_i)$ for each connected component $X_i$ of $X$:
\vspace{-2mm} 
\begin{equation}
  \label{CategoryOfFlatComplexVectorBundlesOverGroupoids}
  \scalebox{.7}{
    \def\arraystretch{.9}
    \color{darkblue}
    \bf
    \begin{tabular}{c}
      flat complex vector bundle
      \\
      over base space $X$ 
    \end{tabular}
  }
  \mathrm{Loc}_{\mathbb{C}}(X)
  \;\;
  \simeq
  \;\;
  \mathrm{Mod}_{\mathbb{C}}^{
    \coprod_{i \in \pi_0(X)} \mathbf{B} \pi_1(X_i)
  }
  \;\;
    :\defneq
  \;\;
  \mathrm{Func}
  \big(
    \mathbf{B}\pi_1(X_1)
    \,\sqcup\,
    \mathbf{B}\pi_1(X_2)
    \,\sqcup\,
    \cdots
    ,\,
    \mathrm{Mod}_{\mathbb{C}}
  \big)
  \,.
\end{equation}

Hence replacing, in the previous discussion, bare sets of points with sets of points-with-automorphisms
(i.e.: groupoids), 
we obtain the following category of flat vector bundles over varying base spaces:
\begin{equation}
  \label{CategoryOfLocalSystemsInIntroduction}
  \hspace{-1.5cm}
  \begin{tikzcd}[row sep=-2pt, column sep=small]
  G\mathrm{Rep}_{\mathbb{C}}
  \ar[rr, hook, "{ \iota }"]
  &&
  \qquad 
  \overset{
    \mathclap{
      \raisebox{7pt}{
        \scalebox{.7}{
          \color{darkblue}
          \bf
          \def\arraystretch{.8}
          \begin{tabular}{c}
            flat complex vector bundles
            \\
            over varying base spaces
          \end{tabular}
        }
      }
    }
  }{
  \overset{
    \mathrm{Loc}_{\mathbb{C}}
    \,\, :\defneq
  }{
  \overbrace{
  \underset{
    \mathclap{
    \scalebox{.7}{$
    \def\arraystretch{.9}
    \begin{array}{c}
    \mathcal{X} \in
        \mathrm{Grpd}
    \end{array}
    $}
    }
  }{\textstyle{\int}}  
    \quad   
    \mathrm{Mod}^{\mathcal{X}}_{\mathbb{C}}
  }
  }
  }
  \\
  (\mathscr{V}, \rho)
  &\longmapsto&
  \mathscr{V}_{\mathbf{B}G}
  \end{tikzcd}
  \hspace{.5cm}
  \mbox{where}
  \hspace{.5cm}
  \begin{tikzcd}[row sep=-3pt, column sep=small]
    \mathrm{Grpd}^{\mathrm{op}}
    \ar[rr]
    &&
    \mathrm{Cat}
    \\
    \mathcal{X} 
    \ar[d, "{f}"]
      &\mapsto&
    \mathrm{Func}\big(
      \mathcal{X},
      \,
      \mathrm{Mod}_{\mathbb{C}}
    )
    \ar[
      from=d,
      "{
        \mathrm{Func}(f,\,\mathrm{Mod}_{\mathbb{C}})
      }"{swap}
    ]
    \ar[
      from=d,
      shift right=70pt,
      "{
        f^\ast 
      }"{swap}
    ]
    \mathrlap{
    \;\,
    \defneq
    \;\,
    \mathrm{Mod}^\mathcal{X}_{\mathbb{C}}
    }
    \\[20pt]
    \mathcal{Y}
      &\mapsto&
    \mathrm{Func}\big(
      \mathcal{Y},
      \,
      \mathrm{Mod}_{\mathbb{C}}
    )
    \mathrlap{
    \;\,
    \defneq
    \;\,
    \mathrm{Mod}^{\mathcal{Y}}_{\mathbb{C}} \;.
    }    
  \end{tikzcd}
\end{equation}
Observe that in this category, every $G$-representation $\mathscr{V}_{\mathbf{B}G}$ sits in a cartesian square of the following form
$$
  \adjustbox{scale=.9}{
  \begin{tikzcd}[row sep=small, column sep=40pt]
    \mathscr{V}
    \ar[rrr]
    \ar[dr]
    \ar[dd]
    &[-25pt]
    &
    &
    \mathscr{V} \!\sslash\! G
    \ar[dd]
    \ar[dr]
    &[-30pt]
    \\[-10pt]
    & 
    \mathrm{pt}
    \ar[rrr, crossing over]
    &&&
    \mathbf{B}G
    \ar[dd]
    \\[+10pt]
    0
    \ar[dr]
    \ar[rrr]
    &&&
    0
    \ar[dr]
    \\[-10pt]
    & 
    \mathrm{pt}
    \ar[from=uu, crossing over]
    \ar[rrr]
    &&&
    \mathbf{B}G
  \end{tikzcd}
  }
  \hspace{1.6cm}
  \begin{tikzcd}[row sep=15pt, column sep=50pt]
    \mathscr{V}_{\mathrm{pt}}
    \ar[r]
    \ar[d]
    \ar[
      dr, 
      phantom, 
      "{
        \scalebox{.6}{(pb)}
      }"
    ]
    &
    \mathscr{V}_{\mathbf{B}G}
    \ar[d]
    \\
    0_{\mathrm{pt}}
    \ar[r]
    &
    0_{\mathbf{B}G}
  \end{tikzcd}
$$
and that the external tensor product of representations {\it preserves} the Cartesianness of these squares:
\vspace{-2mm} 
$$
  \begin{tikzcd}[column sep=50pt, row sep=20pt]
    \mathscr{V}_{\mathrm{pt}}
    \boxtimes
    \mathscr{W}_{\mathbf{B}H}
    \ar[r]
    \ar[d]
    &
    \mathscr{V}_{\mathbf{B}G}
    \boxtimes
    \mathscr{W}_{\mathbf{B}H}
    \ar[d]
    \\
    0_{\mathrm{pt}}
    \boxtimes
    \mathscr{W}_{\mathbf{B}H}    
    \ar[r]
    &
    0_{\mathbf{B}G}
    \boxtimes
    \mathscr{W}_{\mathbf{B}H}    
  \end{tikzcd}
  \hspace{1cm}
  =
  \hspace{1cm}
  \begin{tikzcd}[column sep=50pt, row sep=small]
    \big(\mathscr{V} \otimes \mathscr{W}
    \big)_{\mathbf{B}(1 \times H)}
    \ar[r]
    \ar[dd]
    &
    \big(\mathscr{V} \otimes \mathscr{W}
    \big)_{\mathbf{B}(G \times H)}
    \ar[dd]
    \\
    \\
    0_{\mathbf{B}(1 \times H)}
    \ar[r]
    &
    0_{\mathbf{B}(G \times H)}
  \end{tikzcd}
$$

\vspace{-2mm} 
\noindent But this means that on flat vector bundles, the external tensor products of vector bundles and of group representations unify 
into an external tensor product which is uniquely characterized by these properties:
\begin{itemize}[leftmargin=.5cm]
  \item[\bf 1.] over singletons, it reduces to the ordinary tensor product:
  $\quad 
    \begin{array}{c}
      \mathscr{V}
      \\
      \downarrow
      \\
      \mathrm{pt}
    \end{array}
      \;
      \boxtimes
      \;
    \begin{array}{c}
      \mathscr{W}
      \\
      \downarrow
      \\
      \mathrm{pt}
    \end{array}
    \;\;
      \simeq
    \;
    \begin{array}{c}
      \mathscr{V} 
      \otimes
      \mathscr{W}
      \\
      \downarrow
      \\
      \mathrm{pt}
    \end{array}
  $
  
  \vspace{.3cm}
  
  \item[\bf 2.] It distributes over
  \begin{itemize}[leftmargin=.4cm]
  \item[(i)] coproducts:
  $\qquad \qquad
    \begin{array}{c}
      \mathscr{V}\!\!\sslash\!\! G
      \\
      \downarrow
      \\
      \mathbf{B}G
    \end{array}
        \boxtimes
    \;
    \left(\!\!\!
    \begin{array}{c}
      \mathscr{W}\!\!\sslash\!\! H
      \\
      \downarrow
      \\
      \mathbf{B}H
    \end{array}
          \sqcup
    \begin{array}{c}
      \mathscr{W}'\!\!\sslash\!\! H'
      \\
      \downarrow
      \\
      \mathbf{B}H'
    \end{array}
    \!\!\right)
          =
    \left(\!\!\!
    \begin{array}{c}
      \mathscr{V} \!\!\sslash\!\! G
      \\
      \downarrow
      \\
      \mathbf{B}G
    \end{array}
    \,\boxtimes\,
    \begin{array}{c}
      \mathscr{W} \!\!\sslash\!\! H
      \\
      \downarrow
      \\
      \mathbf{B}H
    \end{array}
   \!\! \right)
    \;\sqcup\;
    \left(\!\!\!
    \begin{array}{c}
      \mathscr{V} \!\!\sslash\!\! G
      \\
      \downarrow
      \\
      \mathbf{B}G
    \end{array}
    \,\boxtimes\,
    \begin{array}{c}
      \mathscr{W}' \!\!\sslash\!\! H'
      \\
      \downarrow
      \\
      \mathbf{B}H'
    \end{array}
    \!\!\! \right).
  $
  \item[(ii)]
  homotopy quotients:
  $$
    \begin{array}{c}
      \mathscr{V}\!\!\sslash\!\! G
      \\
      \downarrow
      \\
      \mathbf{B}G
    \end{array}
    \,\boxtimes\,
    \begin{array}{c}
      \mathscr{W}\!\!\sslash\!\! H
      \\
      \downarrow
      \\
      \mathbf{B}H
    \end{array}
    \;\;\;
    =
    \;\;\;
    \begin{array}{c}
      \big(
      \mathscr{V}
      \otimes
      \mathscr{W}
      \big)
      \!\!\sslash\!\! (G \times H)
      \\
      \downarrow
      \\
      \mathbf{B}(G \times H)
    \end{array}
  $$
  \end{itemize}
  which jointly means that it distributes over
  {\it homotopy quasi-coproducts}
  (\cite{HT95}, Def. \ref{HomotopyQuasiCoproducts}), to be denoted
  $\sqcup^{{}^{hq}}$.
\end{itemize}

\noindent
In homotopy-theoretic generalization of the previous discussion \eqref{TheFirstPushoutDiagramInIntroduction}, this property of the external tensor product may then be re-expressed (we make this precise in \cref{PushoutForLocalSystemsOverGeneralSpaces}) again as a pushout, now understanding cocartesian categories as {\it homotopy quasi-cocartesian categories} (Thm. \ref{PushoutCharacterizationOfExternalTensorProductOnLocalSystems}):

\begin{equation}
  \label{PushoutPropertyForFlatVectorBundlesOverSkeletalGroupoids}
  \begin{tikzcd}[column sep=40pt, row sep=18pt]
    \big(
      \mathrm{Mod}_{\mathbb{C}}
      ,\,
      \otimes
    \big)
    \ar[r, hook, "{ \iota }"]
    \ar[dr, phantom, "{ \scalebox{.7}{(po)} }"]
    &
\colorbox{lightgray}{ $\big(
    \mathrm{Loc}_{\mathbb{C}}
    ,\,
    \sqcup^{{}^{hq}}
    ,\,
    \boxtimes
    \big)
   $ }
    \\
    \mathrm{Mod}_{\mathbb{C}}
    \ar[u]
    \ar[r, hook, "{\iota}"]
    &
    \big(
    \mathrm{Loc}_{\mathbb{C}}
    ,\,
    \sqcup^{{}^{hq}}
    \big)
    \ar[u]
  \end{tikzcd}
  \hspace{.7cm}
\end{equation}

This result constitutes a satisfactory answer to the question \eqref{TheSoughtAfterPushout}.
With the structures in \eqref{PushoutPropertyForFlatVectorBundlesOverSkeletalGroupoids} captureing all of contemporary quantum information theory {\it and} currently understood topological quantum phenomena. But one can still go further, we comment on this in the outlook section \S\ref{ConclusionAndOutlook}.

\section{The external tensor product as a pushout}
\label{ExternalTensorProductAsAPushout}

Here we make precise the formulation and proof of the claim for vector bundles over discrete spaces (\cref{PushoutPropertyForVectorBundlesOverDiscreteSpaces}) and for flat vector bundles over general spaces (\cref{PushoutForLocalSystemsOverGeneralSpaces}). The discussion here involves some basic category theory, relevant background on which we have compiled in \S\ref{SomeDefinitionsAndFacts}.

\subsection{For vector bundles over discrete spaces}
\label{PushoutPropertyForVectorBundlesOverDiscreteSpaces}

We demonstrate -- in the comparatively simple special case of discrete parameter spaces (the default in quantum 
information theory) -- a precise sense in which there is an amalgamation of the theories of entangled and of parameterized 
quantum processes, and that it is encoded in an ``external tensor product'' on bundles of parameterized quantum state 
spaces (Thm. \ref{ThePushoutTheoremOverSets} below). 

\begin{definition}[Categories of monoidal categories]
\label{CategoriesOfCategories}
Consider the following very large categories (cf. Def. \ref{GroupoidsAndCategories}):

\begin{itemize}[leftmargin=.65cm]
\setlength\itemsep{-4pt}
  \item[{\bf (i)}] \fbox{$\mathrm{Cat}$} of categories

    with morphisms all functors,
    \\

  \item[{\bf (ii)}]  \fbox{$\mathrm{MonCat}$} of {\it monoidal categories} (e.g. \cite[\S 1.1]{Kelly82}\cite[\S VII.1]{MacLane97}\cite[\S 2]{EGNO15})

    with morphisms 
    functors that admit the structure of 
    (strong) monoidal functors (e.g. \cite[\S XI.2]{MacLane97}\cite[\S 2.4]{EGNO15}),
    \\
    
  \item[{\bf (iii)}]  \fbox{$\mathrm{CoCartCat}$} of {\it cocartesian categories}
   i.e., monoidal categories whose monoidal operation is the coproduct $\sqcup$

   with morphisms functors that
   admit coproduct-preserving structure,
   \\

 \item[{\bf (iv)}]  \fbox{$\mathrm{DistMonCat}$} of {\it distributive monoidal categories} (e.g. \cite[p. 1]{BJT97}\cite{Labella03}), i.e., of monoidal 
 categories $(\mathcal{C}, \otimes)$ with (set-indexed) coproducts $\coprod$ whose tensor product distributes 
 over the coproduct in each variable, in that 
 for any index set $I$ and indexed set $(A_i)_{i \in I}$ of objects, and any other object $B$,
 the canonical comparison maps are isomorphisms
 \vspace{-2mm} 
 \begin{equation}
   \label{DistributivityOfMonoidalStructure}
   \begin{tikzcd}
   \underset{i \in I}{\coprod}
   \big(
     A_i
     \otimes
     B
    \big)
   \ar[
     rr, 
     "{ (  q_i \,\otimes\, \mathrm{id}_B )_{i \in I} }",
     "{ \sim }"{swap}
   ]
   &&
   \Big(\,
     \underset{i \in I}{\coprod}
     A_i
   \Big)
   \otimes
   B
   \,,
   \end{tikzcd}
   \hspace{1cm}
   \begin{tikzcd}
   \underset{i \in I}{\coprod}
   \big(
     B
     \otimes
     A_i
    \big)
   \ar[
     rr, 
     "{ 
        (  
           \mathrm{id}_B 
             \,\otimes\, 
           q_i 
        )_{i \in I} }",
     "{ \sim }"{swap}
   ]
   &&
   B
   \otimes
   \Big(\,
     \underset{i \in I}{\coprod}
     A_i
   \Big)
   \mathrlap{\,,}
   \end{tikzcd}
 \end{equation}

\vspace{-2mm} 
\noindent  and with morphisms in $\mathrm{DistMonCat}$ being functors that admit (strong) monoidal structure for both products.
\end{itemize}
\end{definition}

\noindent
We are interested for now in the following quadruple of examples:

\begin{example}[Category of complex vector spaces]
  \label{CategoryOfComplexVectorSpaces}
  We write $\mathrm{Mod}_{\mathbb{C}} \,\in\, \mathrm{Cat}$ for the usual category whose objects are complex vector spaces and whose morphisms are complex-linear maps between these.
\end{example}

\begin{example}[Tensor category of complex vector spaces]
  \label{TensorCategoruOfComplexVectorSpaces}
  We write $\big(\mathrm{Mod}_{\mathbb{C}}, \otimes_{\mathbb{C}}\big) \,\in\, \mathrm{MonCat}$
  for the category of complex vector spaces from Ex. \ref{CategoryOfComplexVectorSpaces}, but now regarded as a monoidal category by equipping it with the usual linear tensor product $\otimes_{\mathbb{C}}$ of complex vector spaces (whose tensor unit is $\mathbb{C}$ regarded as a vector space over itself).
\end{example}

The following Ex. \ref{SetAsDistributiveMonoidalCategory} serves to prepare concepts and notation for the main Ex. \ref{CategoryOfComplexVectorBundlesOverSets} and Ex. \ref{DistributiveMonoidalCategoryofVectorBundles} further below.

\begin{example}[Set as distributive monoidal category]
 \label{SetAsDistributiveMonoidalCategory}
  We write $\big(\mathrm{Set},\, \sqcup,\, \times \big)
  \,\in\, \mathrm{DistMonCat}$ for the category of sets regarded as a distributive cartesian monoidal category.

 \item An abstract way to see that the cartesian product distributes over the coproduct 
  is to notice that the product functors $Y \times (-),\,\;(-) \times Y : \mathrm{Set} \longrightarrow \mathrm{Set}$ 
  have a right adjunction (forming function sets $(-)^Y$), which implies that they preserve all colimits and hence, in
  particular, the coproducts involved in distributivity.

 \item  Also notice that every set is isomorphic to the coproduct indexed by its elements, of the singleton set
  \begin{equation}
    \label{SetIsCoproductOfItsElements}
    X \,\in\, \mathrm{Set}
    \;\;\;\;\;\;\;\;
    \vdash
    \;\;\;\;\;\;\;\;
    X \,\simeq\, \underset{x \in X}{\coprod} \ast
    \,.
  \end{equation}
\end{example}

\begin{example}[Category of complex vector bundles over discrete spaces]
\label{CategoryOfComplexVectorBundlesOverSets}
We write 
$$
  \mathrm{Fam}_{\mathbb{C}}
  \;\;
  \coloneqq
  \;\;
  \int_{X \in \mathrm{Set}} \mathrm{Mod}^{X}_{\mathbb{C}}
$$ 
for the category 
of complex vector {\it bundles} over varying sets (i.e., over varying discrete topological spaces), hence for the
Grothendieck construction (Def. \ref{GrothendieckConstruction}) on the following pseudofunctor (Def. \ref{Pseudofunctor})
\vspace{-2mm} 
$$
  \begin{tikzcd}[row sep=-8pt]
    \mathllap{
      \mathrm{Mod}^{(\mbox{-})}_{\mathbb{K}}
      \;:\;\;
    }
    \mathrm{Set}^{\mathrm{op}}
    \ar[rr]
    &&
    \mathrm{Cat}
    \\
    \phantom{a} & \phantom{b} 
    \\
    X
    \ar[d, "{f}"]
      &\longmapsto& 
    \mathrm{Func}(X,\,\mathrm{Mod}_{\mathbb{C}})
    \ar[
      from=d,
      "{ f^\ast \;:=\; (-)\circ f }"{swap}
    ]
    \\[+30pt]
    Y
    &\longmapsto&
    \mathrm{Func}(Y,\,\mathrm{Mod}_{\mathbb{C}})    
  \end{tikzcd}
$$

\vspace{-2mm} 
\noindent and regarded as a cocartesian monoidal category (in fact, this is the {\it free coproduct completion} of $\mathrm{Mod}_{\mathbb{C}}$; 
cf. Ex. \ref{CategoriesOfIndexedSetsOfObjects}).  Explicitly this means the following, where on the right we show the corresponding construction 
of topological vector bundles:

\begin{itemize}[leftmargin=.4cm]
\setlength\itemsep{-3pt}
  \item its objects are pairs $\mathscr{V}_X$ consisting of a base $X \in \mathrm{Set}$ and a functor $\mathscr{V}_{(-)}$ 
  from $X$, regarded as a discrete groupoid, to the category $\mathrm{Mod}_{\mathbb{C}}$ of complex vector spaces, 
  hence equivalently a vector bundle (necessarily and uniquely flat) over $X$:
\vspace{-3mm} 
  $$
    \begin{tikzcd}[row sep=-2pt]
      \mathllap{ \mathscr{V}_{(-)} \;: \;\;}
      X \ar[rr] 
      && \mathrm{Mod}_{\mathbb{C}}
      \\
      x &\mapsto& \mathscr{V}_x
    \end{tikzcd}
    \hspace{1.6cm}
    \longleftrightarrow
    \hspace{1.6cm}
    \begin{tikzcd}[row sep=small]    
      \mathscr{V}_x
      \ar[rr]
      \ar[d]
      \ar[drr, phantom, "\scalebox{.6}{(pb)}"]
      &&
      \mathscr{V}_X
      \ar[d]
      \\
      \{x\}
      \ar[rr, hookrightarrow]
      &&
      X
    \end{tikzcd}
  $$

  \vspace{-2mm} 
  \item
  morphisms $\phi_f \,\colon\, \mathscr{V}_X \longrightarrow \mathscr{W}'_{Y}$ are pairs consisting of a map $f : X \longrightarrow Y$ of base spaces 
  and a natural transformation from $\mathscr{V}_X$ to $f^\ast \mathscr{W}_{Y}$, hence equivalently morphisms of vector bundles covering maps of base spaces:
  \vspace{-2mm} 
  $$
    \begin{tikzcd}
      x \;\;\;\;\mapsto\;\;\;\; 
      \mathscr{V}_x
      \ar[rr, "{\phi_x}"]
      && 
      \mathscr{W}_{f(x)}
    \end{tikzcd}
    \hspace{1.6cm}
    \longleftrightarrow
    \hspace{1.6cm}
    \begin{tikzcd}[row sep=small] 
      \mathscr{V}_X
      \ar[rr, "{ \phi_f }"]
      \ar[d]
      &&
      \mathscr{W}_Y
      \ar[d]
      \\
      X 
      \ar[rr, "{f}"]
      && 
      Y
    \end{tikzcd}
  $$
  \item 
  the cocartesian pairing
  $\mathscr{V}_X \sqcup \mathscr{V}'_{X'}$
  of a pair of objects is
  \vspace{-2mm} 
  \begin{equation}
    \label{CoCartesianPairingOfBundles}
    \begin{tikzcd}[sep=-2pt]
      X &[-6pt]\sqcup&[-6pt] X'
      \ar[rr, "{ \mathscr{V}_X \sqcup \mathscr{V}'_{X'} }"]
      && \mathrm{Mod}_{\mathbb{C}}
      \\
      x& & &\;\;\;\;\;\mapsto\;\;\;\;\;& \mathscr{V}_x
      \\
      && x' &\mapsto& \mathscr{V}'_{x'}
    \end{tikzcd}
    \hspace{1.6cm}
     \longleftrightarrow
    \hspace{1.6cm}
    \begin{tikzcd}[sep=0pt]
      \mathscr{V}_X 
      \ar[d]
        &[-6pt]\sqcup&[-6pt] 
      \mathscr{V}'_{X'}
      \ar[d]
      \\[+20pt]
      X &\sqcup& X'
    \end{tikzcd}
  \end{equation}
\end{itemize}

\vspace{-1mm} 
\noindent
The evident full inclusion of the category of plain vector spaces (Ex. \ref{CategoryOfComplexVectorSpaces}) into bundles 
of vector spaces, given by regarding the former as the bundles over the singleton set $\{\ast\} \in \mathrm{Set}$, we denote as follows:
  \vspace{-2mm} 
\begin{equation}
  \label{FullInclusionOfVectorSpacesIntoVectorBundles}
  \begin{tikzcd}[sep=0pt]
    \mathrm{Mod}_{\mathbb{C}}
    \ar[rr, hook, "{ \iota }"]
    &&
    \mathrm{Fam}_{\mathbb{C}}
    \\
    \mathscr{V} 
    &\mapsto&
    \mathscr{V}_{\{\ast\}}
  \end{tikzcd}
\end{equation}
\end{example}

\begin{remark}[Abstract characterization of the construction]
  \label{AbstractCharacterizationOfTheConstruction}
  The Grothendieck construction $\mathrm{Fam}_{\mathbb{C}} \,:=\, \int_{S \in \mathrm{Set}} \mathrm{Mod}_{\mathbb{K}}^S$ in Ex. \ref{CategoryOfComplexVectorBundlesOverSets}
  may be understood as a {\it free coproduct completion} (Ex. \ref{CategoriesOfIndexedSetsOfObjects}), here applied to $\mathrm{Mod}_{\mathbb{C}}$,
  but the construction exists for any category. Applied to any symmetric closed monoidal category and then regarded as categorical 
  semantics for dependent linear-typed quantum programming languages; this has been considered in
  \cite[\S3.2]{RiosSelinger18}\cite[Def. 2.10]{FuKishidaSelinger20}, see \cite[\S 2.1]{QS}.
\end{remark}

\begin{remark}[Coproducts of bundles over singletons]
  Since every $X \in \mathrm{Set}$ is the disjoint union of the singleton sets $\{x\}$ on its elements $x \in X$
  (Ex. \ref{SetAsDistributiveMonoidalCategory}),
  it follows that every object in $\int_{S \in \mathrm{Set}}\mathrm{Mod}^{S}_{\mathbb{C}}$
  (Ex. \ref{CategoryOfComplexVectorBundlesOverSets})
  is the coproduct \eqref{CoCartesianPairingOfBundles} 
  of its restrictions $\mathscr{V}_{\{x\}}$ to these singletons:
  \vspace{-2mm} 
  \begin{equation}
    \label{VectorBundleOverSetIsCoproductOfItsRestrictionToSingletons}
    \mathscr{V}_X
    \,\in\,
    \mathrm{Mod}^{\mathrm{Set}}_{\mathbb{C}}
    \;\;\;\;\;\;\;\;
      \vdash
    \;\;\;\;\;\;\;\;
    \mathscr{V}_X
    \;\;\simeq\;\;
    \underset{
      x \in X
    }{\coprod}
    \mathscr{V}_{\{x\}}
    \,.
  \end{equation}
\end{remark}

\begin{example}[Distributive monoidal category of vector bundles]
  \label{DistributiveMonoidalCategoryofVectorBundles}
  We write $\big( \mathbb{C}\mathrm{ModMod}_{\mathrm{Set}},\, \sqcup,\, \boxtimes \big) \,\in\,\mathrm{DistMonCat}$ 
  for the cocartesian monoidal category of vector bundles over sets, from Ex. \ref{CategoryOfComplexVectorBundlesOverSets},
  but now in addition equipped with a further monoidal structure given by the following {\it external tensor product} of vector 
  bundles, defined as the result of pulling back to the cartesian product of bases spaces and there forming the usual
  fiberwise tensor product of bundles:
  \vspace{-3mm} 
  \begin{equation}
    \label{ExternalTensorProductOfBundlesOverSet}
    \hspace{1.3cm}
    \begin{tikzcd}[row sep=-3pt, column sep=0pt]
      \mathllap{
      \mathscr{V}_X 
        \boxtimes
      \mathscr{V}'_{X'}
      \;\;
      :
      \;\;
      }
      X \times X'
      \ar[rr]
      &&
      \mathrm{Mod}_{\mathbb{C}}
      \\
      (x,x')
      &\longmapsto&
      \mathscr{V}_x 
        \otimes
      \mathscr{V}'_{x'}
    \end{tikzcd}
    \hspace{.5cm}
     \longleftrightarrow
    \hspace{.5cm}
    \begin{tikzcd}[row sep=small] 
      & 
      &[-40pt]
      \mathclap{
      \big(
      (\mathrm{pr}_X)^\ast 
      \mathscr{V}_X
      \big)
      \otimes
      \big(
      (\mathrm{pr}_{X'})^\ast 
      \mathscr{V}'_{X'}
      \big)
      }
      \ar[dd]
      &[-40pt]
      \\
      \mathscr{V}_X
      \ar[from=r]
      \ar[ddr]
      &
      (\mathrm{pr}_X)^\ast \mathscr{V}_X
      \ar[dr]
      & 
      &
      (\mathrm{pr}_{X'})^\ast \mathscr{V}'_{X'}
      \ar[dl]
      \ar[r]
      &
      \mathscr{V}'_{X'}
      \ar[ddl]
      \\[20pt]
      && 
      X \times X'
      \ar[
        dl,
        "{ \mathrm{pr}_X }"
        {swap}
      ]
      \ar[
        dr,
        "{ \mathrm{pr}_{X'} }"
      ]
      \\
      & X && X'
    \end{tikzcd}
  \end{equation}

  \vspace{-2mm} 
\noindent   This external tensor product indeed distributes over the cocartesian product \eqref{CoCartesianPairingOfBundles} 
in each variable, in the sense required in \eqref{DistributivityOfMonoidalStructure}, in a fiberwise covering of how
the cartesian product of base sets distributes over the disjoint union of sets
  \vspace{-3mm} 
  $$
    \begin{tikzcd}[row sep=6pt, column sep=10pt]
      &&
      \mathscr{V}^{\, i}_{x_i}
      \otimes
      \mathscr{V}'_{x'}
      \ar[rr, equals]
      &&
      \mathscr{V}^{\, i}_{x_i}
      \otimes
      \mathscr{V}'_{x'}      
      \\
      (x_i, x')
      &\mapsto&
      \Big(
      \big(
        \mathscr{V}^1_{X_1}
        \sqcup
        \mathscr{V}^2_{X_2}
      \big)
        \boxtimes
        \mathscr{V}'_{X'}      
      \Big)_{(x_i, x')}
      \ar[u, equals]
      \ar[rr]
      &&
      \Big(
      \big(
        \mathscr{V}^1_{X_1}
        \boxtimes
        \mathscr{V}'_{X'}
      \big)
      \sqcup
      \big(
        \mathscr{V}^2_{X_2}
        \boxtimes
        \mathscr{V}'_{X'}
      \big)
      \Big)_{(x_i, x')}
      \ar[u, equals]
      \\[-6pt]
      &&
      \big(
        X_1 \sqcup X_2
      \big)
        \times
      X'
      \ar[rr, "{ \sim }"]
      &&
      \big(
        X_1 \times X'
      \big)
      \sqcup 
      \big(
        X_2 \times X'
      \big)
      \\[-12pt]
      &&
      (x_i, \, x') &\longmapsto& (x_i,\, x')
    \end{tikzcd}
  $$
\end{example}

Also the converse statement holds:
\begin{proposition}[Characterization of external tensor product]
  \label{CharacterizationOfExternalTensorProduct}
  Up to isomorphism,
  the external tensor product \eqref{ExternalTensorProductOfBundlesOverSet}
  is the unique functor
  \vspace{-1.5mm} 
  $$
    (-) \boxtimes (-)
    \;:\;
    \begin{tikzcd}
    \mathrm{Mod}^{\mathrm{Set}}_{\mathbb{C}}
    \times
    \mathrm{Mod}^{\mathrm{Set}}_{\mathbb{C}}
    \ar[rr]
    &&
    \mathrm{Mod}^{\mathrm{Set}}_{\mathbb{C}}    
    \end{tikzcd}
  $$

  \vspace{-2mm} 
\noindent  such that:
  \begin{itemize}[leftmargin=.65cm]
   \item[{\bf (i)}] It distributes over coproducts 
   \eqref{CoCartesianPairingOfBundles}
   in each
   variable, in the sense of \eqref{DistributivityOfMonoidalStructure}.
   \item[{\bf (ii)}] Restricted to plain vector spaces via 
   \eqref{FullInclusionOfVectorSpacesIntoVectorBundles},
   it coincides with the ordinary tensor product (Ex. \ref{TensorCategoruOfComplexVectorSpaces}).
  \end{itemize}
\end{proposition}

\begin{proof}
  Let for the moment $\boxtimes$ denote any monoidal product satisfying the above two assumptions. 
  Then it is fixed, up to isomorphism, by the following formula:
  \vspace{-1.5mm} 
  \begin{equation}
    \label{ExternalTensorProductAsCoproductOfFiberwiseTensor}
    \def\arraystretch{1.5}
    \begin{array}{lll}
      \mathscr{V}_X 
        \boxtimes
      \mathscr{V}'_{X'}
      &
      \;\simeq\;
      \Big(\,
        \underset{x \in X}{\coprod}
        \mathscr{V}_{\{x\}}
      \Big)
      \boxtimes
      \Big(\,
        \underset{x' \in X'}{\coprod}
        \mathscr{V}'_{\{x'\}}
      \Big)
      &
      \proofstep{by \eqref{VectorBundleOverSetIsCoproductOfItsRestrictionToSingletons}}
      \\
   &   \;\simeq\;
      \underset{x \in X}{\coprod}
      \bigg(\!\!
      \mathscr{V}_{\{x\}}
      \boxtimes
      \Big(\,
        \underset{x' \in X'}{\coprod}
        \mathscr{V}'_{\{x'\}}
      \Big)
      \!\! \bigg)
      &
      \proofstep{by assumption (i) in first variable}
      \\
    &  \;\simeq\;
      \underset{x \in X}{\coprod}
      \;
      \underset{x' \in X'}{\coprod}
      \Big(
        \mathscr{V}_{\{x\}}
        \boxtimes
        \mathscr{V}'_{\{x'\}}
      \Big)
      &
      \proofstep{by assumption (i) in second variable}
      \\
    &  \;\simeq\;
      \underset{x \in X}{\coprod}
      \;
      \underset{x' \in X'}{\coprod}
      \Big(
        \iota
        \big(
        \mathscr{V}_x
        \otimes
        \mathscr{V}_{x'}
        \big)
      \Big)
      &
      \proofstep{by assumption (ii)}
      \\
     & \;\simeq\;
      \underset{
        (x,x') \,\in\, X \times X'
      }{\coprod}
      \Big(
        \iota
        \big(
        \mathscr{V}_x
        \otimes
        \mathscr{V}_{x'}
        \big)
      \!\Big)
      &
      \proofstep{
        just to make the base space manifest,
      }
    \end{array}
  \end{equation}

    \vspace{-1mm} 
\noindent 
  which is manifestly isomorphic to the 
  operation of the external tensor product according to \eqref{ExternalTensorProductOfBundlesOverSet}.
\end{proof}

We proceed to show that the content of Prop. \ref{CharacterizationOfExternalTensorProduct} is equivalently 
exhibited by a pushout of the form requested in \eqref{TheSoughtAfterPushout}.

\begin{remark}[Incrementally forgetting distributive monoidal structure]
The forgetful functors between the categories-of-categories from Def. \ref{CategoriesOfCategories}, i.e.,
those which act as the identity on the underlying categories $\mathcal{C}$ but forget the presence of 
either or any monoidal structure, evidently arrange into a commuting square as follows:
\vspace{-2mm} 
\begin{equation}
  \label{CommutingSquareOfCategoriesOfCategories}
  \begin{tikzcd}
    (\mathcal{C}, \otimes)
    \ar[rrr, phantom, "{\xmapsfrom{\qquad \qquad \qquad \qquad \quad}}"]
    \ar[ddd, phantom, "{ \xmapsto{\qquad \quad} }"{sloped}]
    &[-35pt]
    &
    &[-45pt]
    (\mathcal{C}, \sqcup, \otimes)
    \ar[ddd, phantom, "{ \xmapsto{\qquad \quad} }"{sloped}]
    \\[-20pt]
    &
    \mathrm{MonCat}
    \ar[from=r]
    \ar[d]
    &
   \colorbox{lightgray}{$ \mathrm{DistMonCat}$}
    \ar[d]
    \\
    &
    \mathrm{Cat}
    \ar[from=r]
    &
    \mathrm{CoCartCat}
    \\[-20pt]
    \mathcal{C}
    \ar[rrr, phantom, "{\xmapsfrom{\qquad \qquad \qquad \qquad \qquad}}"]
    &&&
    (\mathcal{C}, \sqcup)
  \end{tikzcd}
\end{equation}
\end{remark}

\begin{definition}
  \label{AnyMonCat}
  Write \fbox{$\mathrm{AnyMonCat}$} for the Grothendieck construction (Def. \ref{GrothendieckConstruction}) on the square \eqref{CommutingSquareOfCategoriesOfCategories}, hence for the very large category-of-categories whose

  \begin{itemize}[leftmargin=.4cm]
  \item objects are categories $\mathcal{C}$ equipped 
  {\it either} with no monoidal structure {\it or} with any monoidal structure $\otimes$ {\it or} with cocartesian monoidal structure $\sqcup$ {\it or} with cocartesian monoidal structure and any further monoidal structure $\otimes$ distributing over it;

  \item morphisms are functors whose codomain category carries at least the kind of monoidal structures that the domain carries and 
  which are strong monoidal with respect to the kind of monoidal structures that the domain carries.
  \end{itemize}
\end{definition}

\begin{example}[The candidate commuting diagram for a pushout]
 \label{TheCommutingDiagramToBeShownAPushout}
  We have the following commuting diagram in $\mathrm{AnyMonCat}$ (Def. \ref{AnyMonCat}):
  \begin{equation}
    \label{CommutingDiagramToBeShownAPushout}
    \begin{tikzcd}
      \mathllap{
        \scalebox{.7}{
          \color{gray}
          Ex. \ref{TensorCategoruOfComplexVectorSpaces}
        }
        \;\;\;
      }
      \big(
        \mathrm{Mod}_{\mathbb{C}}
        ,\,
        \otimes
      \big)
      \ar[
        rr,
        "{ \iota }"
      ]
      &&
      \colorbox{lightgray}{$      
      \big(
        \mathrm{Fam}_{\mathbb{C}}
        ,\,
        \sqcup
        ,\,
        \boxtimes
      \big)
      $}
      \mathrlap{
        \;\;\;
        \scalebox{.7}{
          \color{gray}
          Ex. \ref{DistributiveMonoidalCategoryofVectorBundles}
        }
      }
      \\
      \mathllap{
        \scalebox{.7}{
          \color{gray}
          Ex. \ref{CategoryOfComplexVectorSpaces}
        }
        \;\;\;
      }
      \mathrm{Mod}_{\mathbb{C}}
      \ar[
        rr,
        "{ \iota }"
      ]
      \ar[u]
      &&
      \big(
        \mathrm{Fam}_{\mathbb{C}}
        ,\, 
        \sqcup
      \big)
      \ar[u]
      \mathrlap{
        \;\;\;
        \scalebox{.7}{
          \color{gray}
          Ex. \ref{CategoryOfComplexVectorBundlesOverSets}
        }
      }
    \end{tikzcd}
  \end{equation}
  where 
  \begin{itemize}
    \item the underlying functor of the vertical morphisms is the respective identity functor,
    \item the underlying functor 
    of both horizontal morphisms is 
    \eqref{FullInclusionOfVectorSpacesIntoVectorBundles},
    \item the top horizontal morphism is 
    strong monoidal essentially by construction 
    
    (or alternatively as a special case of Prop. \ref{CharacterizationOfExternalTensorProduct}),
    \item the right identity functor is tautologically strong monoidal.
  \end{itemize}  
  Therefore, the underlying diagram of functors clearly commutes and there is no non-trivial composition of strong-monoidal 
  structure involved, hence the diagram commutes in $\mathrm{AnyMonCat}$.
\end{example}

\begin{remark}[On morphisms in  AnyMonCat]
  \label{OnMorphismsInAnyMonCat}
$\,$
  
\noindent  {\bf (i)} By Definition \ref{AnyMonCat}, even if the underlying functors are identities, when regarded as morphisms in $\mathrm{AnyMonCat}$
  they must point in a direction such that no monoidal structure is ``forgotten'' along the way. 
  
  \noindent {\bf (ii)} For instance, 
  for none of the morphisms shown in the diagram 
  \eqref{CommutingDiagramToBeShownAPushout} does there exist a reverse morphism in $\mathrm{AnyMonCat}$. 
  
\noindent {\bf (iii)}  More importantly:
If the left and bottom part of the diagram \eqref{CommutingDiagramToBeShownAPushout} is given, then the only possibility to 
  complete it to a square in $\mathrm{AnyMonCat}$ is by having a distributive monoidal category in the top right corner, 
  because only such a structure can receive morphisms in $\mathrm{AnyMonCat}$ from both a monoidal category (top left) 
  and a cocartesian category (bottom right).
\end{remark}

We may now state and prove the conclusion of this discussion:

\begin{theorem}[Pushout characterization of the external tensor product of vector bundles over sets]
  \label{ThePushoutTheoremOverSets} $\,$
  \newline
  The diagram in Ex. \ref{TheCommutingDiagramToBeShownAPushout} is a pushout.
\end{theorem}
\begin{proof}
  We check the defining universal property of the pushout. To that end, consider any extension of the square to a cocone diagram as shown
  by solid arrows in the following
  (where the tip of the cocone is necessarily a distributive monoidal category, as shown, by Rem. \ref{OnMorphismsInAnyMonCat}):
  \vspace{-2mm} 
  \begin{equation}
    \label{CoconeDiagram}
    \begin{tikzcd}[row sep=small]
      && &
      \big(
        \mathcal{C}, \sqcup, \otimes_{{}_{\mathcal{C}}}
      \big)   \mathrlap{\,.}
      \\
      \big(
        \mathrm{Mod}_{\mathbb{C}}
        ,\,
        \otimes
      \big)
      \ar[
        rr,
        "{ \iota }"
      ]
      \ar[urrr, bend left=10]
      &&
      \big(
        \mathrm{Fam}_{\mathbb{C}}
        ,\,
        \sqcup
        ,\,
        \boxtimes
      \big)
      \ar[ur, dashed]
      \\
      \mathrm{Mod}_{\mathbb{C}}
      \ar[
        rr,
        "{ \iota }"
      ]
      \ar[u]
      &&
      \big(
        \mathrm{Fam}_{\mathbb{C}}
        ,\, 
        \sqcup
      \big)
      \ar[u]
      \ar[uur, bend right=14, "{ F }"{swap}]
    \end{tikzcd}
  \end{equation}

  \vspace{-2mm} 
\noindent  We need to demonstrate that there exists a unique dashed morphism making the full diagram commute.
  First observe that, since the underlying functor of the vertical morphisms are identity functors, the dashed morphism, 
  if it exists at all, is uniquely fixed to also be given by $F$, and so the underlying functor of the top morphisms 
  must necessarily be $F \circ \iota$.
  Hence we really have a diagram as follows 
\vspace{-2mm} 
  \begin{equation}
    \label{AnalyzedCoconeDiagram}
    \begin{tikzcd}[row sep=small]
      && &
      \big(
        \mathcal{C}, \sqcup, \otimes_{{}_{\mathcal{C}}}
      \big)   \mathrlap{\,,}
      \\
      \big(
        \mathrm{Mod}_{\mathbb{C}}
        ,\,
        \otimes
      \big)
      \ar[
        rr,
        "{ \iota }"
      ]
      \ar[
        urrr, 
        bend left=10, 
        "{ F \circ \iota }"
      ]
      &&
      \big(
        \mathrm{Fam}_{\mathbb{C}}
        ,\,
        \sqcup
        ,\,
        \boxtimes
      \big)
      \ar[ur, dashed, "{ F }"{description}]
      \\
      \mathrm{Mod}_{\mathbb{C}}
      \ar[
        rr,
        "{ \iota }"
      ]
      \ar[u]
      &&
      \big(
        \mathrm{Fam}_{\mathbb{C}}
        ,\, 
        \sqcup
      \big)
      \ar[u]
      \ar[uur, bend right=14, "{ F }"{swap}]
        \end{tikzcd}
  \end{equation} 
  \vspace{-2mm} 

  \noindent 
  whose underlying diagram of functors commutes. Therefore we are reduced to showing that the dashed morphism 
  is well-defined as a morphism in $\mathrm{AnyMonCat}$, which means to show that it intertwines the external 
  tensor product $\boxtimes$ on $\mathrm{Fam}_{\mathbb{C}}$ with the given tensor product 
  $\otimes_{{}_{\mathcal{C}}}$ on $\mathbb{C}$. This is verified by the following sequence of natural isomorphisms:
  \vspace{-2mm} 
  $$
    \def\arraystretch{1.7}
    \begin{array}{lll}
      F
      \big(
      \mathscr{V}_{X}
      \boxtimes
      \mathscr{V}'_{X'}      
      \big)
      &
      \;\simeq\;
      F
      \bigg(\,
        \underset{(x,x') \in X \times X'}{\coprod}
        \,
        \iota
        \big(
        \mathscr{V}_x \otimes \mathscr{V}_{x'}
        \big)
      \!\! \bigg)
      &
      \proofstep{
        by \eqref{ExternalTensorProductAsCoproductOfFiberwiseTensor}
      }
      \\
  &    \;\simeq\;
      \underset{(x,x') \in X \times X"}{\coprod}
      \,
      F
      \Big(
        \iota
        \big(
          \mathscr{V}_x \otimes \mathscr{V}_{x'}
        \big)
      \Big)
      &
      \proofstep{
        since $F$ preserves coproducts, by assumption
        \eqref{CoconeDiagram}
      }
      \\
 &     \;\simeq\;
      \underset{(x,x') \in X \times X"}{\coprod}
      \,
      F \big(\iota(\mathscr{V}_x)\big)
      \otimes_{{}_{\mathcal{C}}}
      F \big(\iota(\mathscr{V}'_{x'})\big)
      &
      \proofstep{
        via strong monoidal structure on 
        $F \circ \iota$,
        by \eqref{AnalyzedCoconeDiagram}
      }
      \\
   &   \;\simeq\;
      \Big(\,
        \underset{x \in X}{\coprod}
        \,
        F \big(\iota(\mathscr{V}_x)\big)
      \Big)
      \otimes_{{}_{\mathcal{C}}}
      \Big(\,
        \underset{x' \in X'}{\coprod}
        \,
        F \big(\iota(\mathscr{V}'_{x'})\big)
      \Big)
      &
      \proofstep{
        since $\otimes_{{}_{\mathcal{C}}}$ 
        is distributive, by assumption
        \eqref{CoconeDiagram}
      }
      \\
   &   \;\simeq\;
      F
      \Big(\,
        \underset{x \in X}{\coprod}
        \,
        \big(\iota(\mathscr{V}_x)\big)
      \Big)
      \otimes_{{}_{\mathcal{C}}}
      F
      \Big(\,
        \underset{x' \in X'}{\coprod}
        \,
        \big(\iota(\mathscr{V}'_{x'})\big)
      \Big)
      &
      \proofstep{
        since $F$ preserves coproducts, by assumption
        \eqref{CoconeDiagram}
      }
      \\
    &  \;\simeq\;
      F(\mathscr{V}_{X})
      \,\otimes_{{}_{\mathcal{C}}}\,
      F(\mathscr{V}'_{X'})
      &
      \proofstep{
        by \eqref{ExternalTensorProductAsCoproductOfFiberwiseTensor}.
      }
    \end{array}
  $$

  \vspace{-7mm}
\end{proof}

\begin{remark}[Doubly closed monoidal structure on vector bundles and its ``bunched'' classical/quantum logic]
\label{DoublyClosedMonoidalCategories}
 
 \noindent
 {\bf (i)}  
 The category $\mathrm{Fam}_{\mathbb{C}}$ (Ex. \ref{CategoryOfComplexVectorBundlesOverSets}) also carries 
 a cartesian product --- the ``external cartesian product'' (Ex. \ref{ExternalCartesianProduct}).
    
 \noindent
 {\bf (ii)}    
 Since both this and the external tensor product are closed (cf. \cite[Prop. 2.3]{QS}), 
 jointly they make for a ``doubly closed monoidal category'' \cite[\S 3]{OHearnPym99}\cite[\S 2.2]{OHearn03} which one may want 
 to think of as providing categorical semantics for both classical propositional logic as well as for the multiplicative 
 fragment of linear logic. Since the antecedents in such a mixed classical/quantum logic are no longer plain lists of classical
 products of classical contexts, but more generally nested (``bunched'') trees obtained by alternatively using the (external)
 tensor product, this idea of combined classical/linear logic has originally been advertised as a logic of ``bunched implications'' \cite{OHearnPym99}\cite{Pym02} and came to be known as {\it bunched logic}, for short (e.g. \cite{ZBHYY21}, where also the
 quantum aspect of bunched logic is considered).
 
 \noindent   
 {\bf (iii)}   
 However, there have all along been subtle technical difficulties with promoting the broad idea of bunched logic to a 
 satisfactory formal language; these problems have been highlighted in \cite[p. 5]{Pym08} and further in
 \cite[pp. 19 and Rem. 1.4.1]{Riley22}. The claim of \cite{Riley22} is that all these problems are finally resolved by enhancing
 linear type theory to Linear Homotopy Type Theory, {\tt LHoTT}, which we may thus understand as the first working {\it universal} quantum programming language \cite{QS}. On the side of the categorical 
 semantics this requires promoting the doubly closed monoidal category of vector bundles over sets to a suitably doubly 
 monoidal model category presenting $\infty$-local systems over general homotopy types. This generalization of the present discussion is addressed in \cite{SS26-Global}.
\end{remark}

\subsection{For flat vector bundles over general spaces}
\label{PushoutForLocalSystemsOverGeneralSpaces}

We generalize the previous discussion from vector bundles over discrete 
spaces to flat vector bundles over arbitrary base spaces. 
The previous discussion in \cref{PushoutPropertyForVectorBundlesOverDiscreteSpaces} was essentially a variation of the theme that 
$\mathrm{Fam}_{\mathbb{C}}$ (Ex. \ref{CategoryOfComplexVectorBundlesOverSets}) is the free coproduct 
completion of $\mathrm{Mod}_{\mathbb{C}}$ (Ex. \ref{CategoriesOfIndexedSetsOfObjects}). In looking for a first homotopy-theoretic 
generalization of this notion, one may observe that coproducts are, of course, just the colimits over diagrams of the shape of 
a discrete category in $\mathrm{Set} \hookrightarrow \mathrm{Cat}$.

\medskip
Therefore, we are naturally led to ask more generally for completion of categories (notably of $\mathrm{Mod}_{\mathbb{C}}$) 
under colimits over diagrams of the shape of skeletal groupoids $\mathrm{Grpd}_{\mathrm{skl}} \hookrightarrow{\;} \mathrm{Cat}$: 
This should combine formation of coproducts (indexed by the set of connected components of a given skeletal groupoid) with the 
formation of {\it quotients by group actions} indexed by the automorphisms group of any connected component). 
More precisely, we should ask here for {\it homotopy quotients} over group actions. This is what we make precise in Def. \ref{HomotopyQuasiCoproducts} below.

\medskip

\noindent
{\bf Groupoids.}
In all of the following we write $\mathrm{Grpd}$ (Def. \ref{GroupoidsAndCategories}) for the 1-category of small strict 
groupoids, regarded as a monoidal category under the cartesian product.

\begin{example}[Basic examples of groupoids and notation (e.g {\cite[\S 1.2]{EquBundles}})]
  \label{BasicExampleOfGroupoids}
  For $(G, \mu, \mathrm{e}) \in \mathrm{Grp}$, we write
\begin{itemize}[leftmargin=.5cm]
  \item
  $
    \mathbf{B}G 
    \,:\defneq\,
    (G \rightrightarrows \mathrm{pt})
    \,\in \mathrm{Grpd}\,
  $ 
  for the {\it delooping groupoid} of $G$, with composition given by reverse group multiplication
  \vspace{-1mm} 
  \begin{equation}
    \label{DeloopingGroupoids}
    \mathbf{B}G
    \;:=\;
    \left\{
    \adjustbox{raise=4pt}{
    \begin{tikzcd}[row sep=small, column sep=large]
      &
      \mathrm{pt}
      \ar[dr, "{ g_{23} }"{sloped}]
      \\
      \mathrm{pt}
      \ar[ur, "{ g_{12} }"{sloped}]
      \ar[rr, "{ \mu(g_{23},\,g_{12}) }"]
      &&
      \mathrm{pt}
    \end{tikzcd}
    }
    \right\}
    \,,
  \end{equation}
  so that (left) $G$-actions are equivalently functors out of the delooping:
  \begin{equation}
    \label{GActionsAsFunctors}
    \begin{tikzcd}[sep=0pt]
    G 
      \ar[rr, "{ \mathrm{homom.} }"{swap}]
     &&
    \mathrm{Hom}_{\mathcal{C}}(\mathscr{V},\,\mathscr{V})
    \\[+2pt]
    \hline
    \\[+3pt]
    \mathbf{B}G
    \ar[rr, "{ \mathrm{funct.} }"]
    &&
    \mathcal{C}
    \\[-3pt]
    \mathrm{pt} &\mapsto& \mathscr{V}
    \,.
    \end{tikzcd}
  \end{equation}
  \item
  $
    \mathbf{E}G
    \,:\defneq\,
    (G \times G \rightrightarrows G)
    \in \mathrm{Grpd}
  $
  for the action groupoid of $G$ acting on itself by left multiplication, so that we have a forgetful functor (see \cite[\S 2.3]{EquBundles} for more background)
  \begin{equation}
    \label{ProjectionFromEGToBG}
    \begin{tikzcd}[sep=0pt]
      \mathbf{E}G
      \ar[rr, "{ q }"]
      &&
      \mathbf{B}G
      \\
      g 
      \ar[d, "{ g_{12} }"]
      &\mapsto& 
      \mathrm{pt}
      \ar[d, "{ g_{12} }"]
      \\[+15pt]
      \mu(g_{12}, g)
      &\mapsto&
      \mathrm{pt}
    \end{tikzcd}
  \end{equation}
  and a remaining $G$-action by right inverse multiplication
  \begin{equation}
    \label{CanonicalGActionOnEG}
    \begin{tikzcd}[sep=0pt]
      \mathbf{B}G
      \ar[rr]
      &&
      \mathrm{Grpd}
      \\
      \mathrm{pt}
      \ar[d, "{ g }"]
      &\mapsto&
      \mathbf{E}G
      \ar[d, "{ \mu(\mbox{-},g^{-1}) }"]
      \\[+15pt]
      \mathrm{pt}
      &\mapsto&
      \mathbf{E}G
    \end{tikzcd}
  \end{equation}
  whose colimiting cocone is $q \,:\,\mathbf{E}G \to \mathbf{B}G$ \eqref{ProjectionFromEGToBG}. 

\item $\HomotopyQuotient{W}{G}$ for the {\it action groupoid} of the (left) action $G \acts W$ of a group $G$ on a set $W$, whose set of objects is $W$ and whose morphisms are given by the group translations:
\begin{equation}
  \label{ActionGroupoid}
  \HomotopyQuotient{W}{G}
  \;\;\defneq\;\;
  \left\{
  \adjustbox{raise=4pt}{
  \begin{tikzcd}[
    row sep=10pt,
  ]
    &[-10pt]
    g_1 \!\cdot\! w
    \ar[dr, "{ g_2 }"]
    &[-20pt]
    \\
    w
    \ar[
      ur, "{ g_1 }"
    ]
    \ar[
      rr,
      "{ \mu(g_2, g_1) }"{swap}
    ]
    &&
    \mu(g_2,  g_1) \!\cdot\! w
  \end{tikzcd}
  }
  \right\}
  \,.
\end{equation}
  Notice that the previous examples are special cases of action groupoids:
  \[
    \mathbf{B}G
    \,=\,
    \HomotopyQuotient{\mathrm{pt}}{G}
    ,\hspace{.8cm}
    \mathbf{E}G
    \,=\,
    \HomotopyQuotient{G}{G}
    \,.
  \]
   In particular, the terminal map $W \to \mathrm{pt}$ induces for every action groupoid a canonical (Kan-)fibration
   \begin{equation}
     \begin{tikzcd}[
       row sep=2pt
     ]
       \HomotopyQuotient{W}{G}
       \ar[
         r,
         "{
           \scalebox{.7}{
             fibration
           }
         }"{swap}
       ]
       &
       \mathbf{B}G
       \\
       w
       \ar[r, phantom, "{ \mapsto }"]
       \ar[
         d,
         "{
           g
         }"
       ]
       & 
       \mathrm{pt}
       \ar[
         d,
         "{
           g
         }"
       ]
       \\[15pt]
       g \cdot w
       \ar[r, phantom, "{ \mapsto }"]
       & 
       \mathrm{pt}
     \end{tikzcd}
   \end{equation}

  \item
  $\mathrm{CoDisc}(S) \,:=\, (S \times S \rightrightarrows) \,\in\, \mathrm{Grpd}$ for the {\it pair groupoid} 
on some $S \,\in\, \mathrm{Set}$, i.e., the groupoid whose objects are the elements of $S$ and which has a unique
morphism between any pair of objects. For example:
  \begin{equation}
    \label{CodiscreteGroupoids}
    \mathrm{CoDisc}\big(\{1,2,3,4\}\big)
    \;\;
    \defneq
    \;\;
    \left\{\!\!
    \adjustbox{raise=4pt}{
    \begin{tikzcd}[scale=2em]
      2 
      \ar[r, <->, bend left=16]
      \ar[dr, <->]
      & 
      3
      \ar[d, <->, bend left=16]
      \\
      1
      \ar[r, <->, bend right=16]
      \ar[u, <->, bend left=16]
      \ar[ur, <->, crossing over]
      & 
      4
    \end{tikzcd}
    }
  \!\!  \right\}
    \mathrlap{\,.}
  \end{equation}
  These codiscrete groupoids serve as {\it contractible resolutions of the point}, since their terminal functor 
  is an equivalence (either in the sense of categorical equivalence or in the sense of homotopy equivalence):
  \begin{equation}
    \label{ContractabilityOfCodiscreteGroupoids}
    \begin{tikzcd}
    \mathrm{CoDisc}(S)
      \ar[rr, "{ \mathrm{equivalence} }"{swap}]
      &&
    \ast
    \end{tikzcd}
  \end{equation}
  Notice that $\mathbf{E}G$ \eqref{ProjectionFromEGToBG} is isomorphic to a codiscrete groupoid: 
  $\mathbf{E}G \,\simeq\, \mathrm{CoDisc}(G)$.

\item $\Pi_1(X)$ for the {\it fundamental groupoid} of topological space $X$ (e.g. \cite[\S 6]{Higgins71}), whose objects are the elements $x \in X$, whose morphisms are homotopy classes $[\gamma]$ (relative their endpoints) of continuous paths $\gamma : [0,1] \to X$ 
 and whose composition operation is given by concatenation of paths:
\begin{equation}
  \label{FundamentalGroupoid}
  \Pi_1(X)
  \;\;
  \defneq
  \;\;
  \left\{\!\!\!\!\!\!
  \adjustbox{raise=10pt}{
  \begin{tikzcd}[decoration=snake]
    &
    x_2
    \ar[
      dr,
      decorate,
      bend left=5,
      "{ [\gamma_{\,23}] }"
    ]
    \\[-15pt]
    x_1
    \ar[
      ur,
      decorate,
      bend left=5,
      "{ [\gamma_{12}] }"
    ]
    \ar[
      rr,
      decorate,
      bend right=5,
      "{ 
        [\mathrm{conc}(
          \gamma_{\,23}
          ,\,
          \gamma_{12} 
        )]
      }"{swap, yshift=-2pt}
    ]
    &&
    x_3
  \end{tikzcd}
  }
  \!\!\!\!\!\!\!  
  \right\}
  \,.
\end{equation}
\end{itemize}
If $X$ is connected with any choice of basepoint $x_0$, then the evident inclusion of the delooping groupoid \eqref{BasicExampleOfGroupoids} of the fundamental group $\pi_1(X,x_o)$ is an equivalence of groupoids:
\[
  \begin{tikzcd}[
    row sep=1pt,
    decoration=snake
  ]
    \mathbf{B}\pi_1(X, x_0)
    \ar[
      r,
      "{
        \scalebox{.7}{
          equivalence
        }
      }"{swap}
    ]
    &
    \Pi_1(X)
    \\
    \mathrm{pt}
    \ar[
      r,
      phantom,
      "{ \mapsto }"
    ]
    \ar[
      d,
      "{ [\gamma] }"
    ]
    &
    x_0
    \ar[
      d,
      decorate,
      "{ [\gamma] }"{xshift=2pt}
    ]
    \\[20pt]
    \mathrm{pt}
    \ar[
      r,
      phantom,
      "{ \mapsto }"
    ]
    &
    x_0    
  \end{tikzcd}
\]
\end{example}

\begin{definition}[Skeletal groupoids, cf. e.g. {\cite[p. 91]{MacLane97}\cite[\S 2.6]{Richter20}}]
\label{SkeletalGroupoid}
  A groupoid is called {\it skeletal} if it is a disjoint union of delooping groupoids \eqref{DeloopingGroupoids}:
  \begin{equation}
    \label{SkeletalGroupoidDeclaration}
    \mbox{$\mathcal{X}$ is skeletal}
    \hspace{1cm}
      \Leftrightarrow
    \hspace{1cm}
    \mathcal{X}
      \;\underset{\mathrm{iso}}{\simeq}\;
    \underset{
      x \in \mathrm{Obj}(\mathcal{X})
    }{\coprod}
    \mathbf{B}
    \big(
      \mathcal{X}(x,x)
    \big)
    \,.
  \end{equation}
\end{definition}
It is a standard fact  (assuming the \emph{axiom of choice} in the underlying set theory, as usual) that every groupoid is 
adjoint equivalent (in the sense of equivalence of categories) to a skeletal one (Def. \ref{SkeletalGroupoid}). But for the purposes here it is useful to rephrase this as follows, 
using a 1-category theoretic ``model'' for the notion of equivalence, in view of \eqref{ContractabilityOfCodiscreteGroupoids}:

\begin{lemma}[Connected is delooping times codiscrete]
  \label{ConnectedGroupoidIsomorphicToDeloopingTimesCodiscrete}
  Every connected groupoid is isomorphic to the product of a delooping groupoid \eqref{DeloopingGroupoids} 
  with a codiscrete groupoid \eqref{CodiscreteGroupoids}:
  \vspace{-2mm} 
  $$
    \left.
    \def\arraystretch{1.2}
    \begin{array}{l}
      \mathcal{X} 
        \,\in\,
      \mathrm{Grpd}
      \\
      \pi_0(\mathcal{X}) \,\simeq\, \ast
      \\
      x_0 \,\in\, \mathrm{Obj}(X)
    \end{array}
    \right\}
    \hspace{1cm}
    \vdash
    \hspace{1cm}
    \mathcal{X}
    \;\;\underset{\mathrm{iso}}{\simeq}\;\;
    \mathrm{CoDisc}\big(\mathrm{Obj}(\mathcal{X})\big)
    \times
    \mathbf{B}\big(
     \mathcal{X}(x_0,x_0)
    \big).
  $$
\end{lemma}
\begin{proof}
   Choose for each object $x \in \mathrm{Obj}(\mathcal{X})$ a morphism $\gamma : x_{0} \xrightarrow{\;} x$ 
   (which exists by the assumption that $\mathcal{X}$ is connected). This gives the following isomorphism:
  \vspace{-1mm} 
  $$
    \begin{tikzcd}[row sep=0pt, column sep=small]
      \mathcal{X} \quad 
      \ar[
        rr, 
        "{ \sim }"
      ]
      &&
      \mathrm{CoDisc}\big(
        \mathrm{Obj}(\mathcal{X})
      \big)
      \,\times\,
      \mathbf{B}
      \big(
        \mathcal{X}(x_0,\,x_0)
      \big)\,.
      \\
      x
      \ar[
        d,
        "{ f }"
      ]
      &\longmapsto&
      (x,\,\mathrm{pt})
      \ar[
        d,
        "{
          \scalebox{.9}{$
            \scalebox{1.2}{$($}
              (x,\,x'),
              \,
              \gamma_{x'}^{-1}
                \circ
              f
                \circ 
              \gamma_x
            \scalebox{1.2}{$)$}
          $}
        }"
      ]
      \\[25pt]
      x'
      &\longmapsto&
      (x',\,\mathrm{pt})
    \end{tikzcd}
  $$
 \vspace{-.8cm}

\end{proof}

\medskip

\noindent
{\bf Free homotopy quasi-coproduct completion.}

\smallskip 
In mild variation of \cite[\S 1.3]{HT95}\footnote{Our notion of homotopy quasi-coproducts (Def. \ref{HomotopyQuasiCoproducts}) is a 
special case of the notion of {\it quasi-coproducts} of \cite[\S 1.3]{HT95} in that we require the respective actions not just to be free, 
but to be free {\it qua} the Borel construction. But our definition is also slightly stronger in that we in addition require categories 
with homotopy quasi-coproducts be tensored over $\mathrm{Grpd}$. This may be understood as imposing the further requirement that homotopy 
quasi-coproducts over constant diagrams are well-behaved.}, we say:

\begin{definition}[Homotopy quasi-coproducts]
\label{HomotopyQuasiCoproducts}
A {\it category with homotopy quasi-coproducts} is a category $\mathcal{C}$ 
\begin{itemize}
\item[{\bf (i)}] equipped with a tensoring over $\mathrm{Grpd}$
\vspace{-2mm} 
\begin{equation}
  \label{GrpdTensoringOfCategoryWithHomotopyQuasiCoproducts}
  \begin{tikzcd}
    \mathrm{Grpd}
    \times
    \mathcal{C}
    \ar[rr, "{ (\mbox{-})\cdot(\mbox{-}) }"]
    &&
    \mathcal{C}
  \end{tikzcd}
\end{equation}

\vspace{-2mm} 
\item[{\bf (ii)}]
which has all colimits over diagrams of shapes of skeletal groupoids of the following form:
\begin{equation}
  \label{HomotopyQuasiCoproductDiagram}
  \begin{tikzcd}[column sep=large]
    \mathcal{X}
    \ar[rr]
    \ar[d, equals]
    &&
    \mathcal{C}
    \\
    \underset{i \in I}{\coprod}
    \,
    \mathbf{B}G_i
    \ar[
      rr,
      "{
        \left(
          \mathbf{E}G_i,\, \mathscr{V}_{(-)}
        \right)_{i \in I}\;\;
      }"
    ]
    &&
    \mathrm{Grpd} \times \mathcal{C}
    \ar[
      u, 
      "{
        (\mbox{-})\cdot(\mbox{-})
      }"{swap}
    ]
  \end{tikzcd}
\end{equation}
This means, equivalently, that $\mathcal{C}$ has 
\begin{itemize}
\item[{\bf (a)}] all set-indexed coproducts,
\item [{\bf (b)}] all ``Borel constructions'', namely all quotients of diagonal actions of any group $G$ on objects of the form $(\mathbf{E}G) \cdot \mathscr{V}$, where $\mathscr{V}$ is equipped with any $G$-action \eqref{GActionsAsFunctors} and where $\mathbf{E}G$ carries the canonical action \eqref{CanonicalGActionOnEG}.
\end{itemize}
\end{itemize}
We write $\mathrm{CoCartCat}^{{}^{h q}}$ for the (very large) category (Def. \ref{GroupoidsAndCategories}) whose objects are categories with homotopy quasi-coproducts and whose morphisms are functors between them preserving this structure.
\end{definition}

\begin{definition}[Free homotopy quasi-coproduct completion]
\label{FreeHomotopyQuasiCoproductCompletion}
Given a category $\mathcal{C}$ we say that its {\it free homotopy quasi-coproduct completion} is a full inclusion of $\mathcal{C}$ in a category with homotopy quasi-coproducts (Def. \ref{HomotopyQuasiCoproducts}) such that any functor out of the latter which preserves the $\mathrm{Grpd}$-tensoring \eqref{GrpdTensoringOfCategoryWithHomotopyQuasiCoproducts} 
and the homotopy quasi-coproducts \eqref{HomotopyQuasiCoproductDiagram} is already fixed by its restriction to $\mathcal{C}$.
\end{definition}

We will show that the following construction realizes this free homotopy quasi-coproduct completion, at least for cocomplete categories:
\begin{definition}[Category of local systems with coefficients in any category]
\label{CategoryOfLocalSystemsWithCoefficientsInAnyCategory}
For $\mathcal{C}$ a category, we write
\begin{equation}
  \label{LocWithCoefficientsInAnyC}
  \mathrm{Loc}_{\mathcal{C}}
  \;\;
    :\defneq
  \;\;
  \underset{
    \mathcal{X}\,\in\,
    \mathrm{Grpd}
  }{\int}
  \,
  \mathcal{C}^{\mathcal{X}}
\end{equation}
for the Grothendieck construction (Def. \ref{GrothendieckConstruction}) on the pseudo-functor of functor categories from groupoids into $\mathcal{C}$, as in \eqref{CategoryOfLocalSystemsInIntroduction}.
We regard this as equipped with the full inclusion of objects of $\mathcal{C}$ regarded as constant functors on the terminal groupoid
\begin{equation}
  \label{IncludingConstantLocalSystems}
  \begin{tikzcd}[sep=0pt]
    \mathcal{C}
    \ar[rr, hook, "{ \iota }"]
    &&
    \mathrm{Loc}_{\mathcal{C}}
    \\
    \mathscr{V} 
      &\mapsto&
    \mathscr{V}_{\mathrm{pt}}
  \end{tikzcd}
\end{equation}
and with the $\mathrm{Grpd}$-tensoring given by 
\vspace{-2mm} 
\begin{equation}
  \label{GrpdTensoringOfLocalSystems}
  \begin{tikzcd}[row sep=0pt, column sep=small]
    \mathrm{Grpd}
    \times
    \mathrm{Loc}_{\mathcal{C}}
    \ar[rr, "{ (\mbox{-})\cdot(\mbox{-}) }"]
    &&
    \mathrm{Loc}_{\mathcal{C}}
    \\
    \big(
      \mathcal{X},\,\mathscr{W}_{\mathcal{Y}}
    \big)
    &\longmapsto&
    \big(
      (\mathrm{pr}_{\mathcal{Y}})^\ast
      \mathscr{V}
    \big)_{\mathcal{X} \times \mathcal{Y}}
    \\
    \phantom{a} && \phantom{a}
  \end{tikzcd}
\end{equation}
\end{definition}

\begin{example}[Group representations as local systems]
  \label{GroupActionsAsLocalSystems}
  For each $G \,\in\, \mathrm{Grp}$ there is a full inclusion of the category of $G$-actions on objects of $\mathcal{C}$
  into the category of local systems \eqref{LocWithCoefficientsInAnyC}, given by the correspondence \eqref{GActionsAsFunctors}:
  \vspace{-2mm} 
  \begin{equation}
    \label{GroupActionAsLocalSystems}
    \begin{tikzcd}[
      column sep=20pt,
      row sep=-2pt
    ]
      G\mathrm{Act}(\mathcal{C})
      \ar[
        r, 
        "{ \sim }"
      ]
      &
      \mathcal{C}^{\mathbf{B}G}
      \ar[
        r, 
        hook
      ]
      &
      \mathrm{Loc}_{\mathcal{C}}\;.
      \\
      G \acts \mathscr{V}
      &\longmapsto&
      \mathscr{V}_{\mathbf{B}G}
    \end{tikzcd}
  \end{equation}

  \vspace{-2mm} 
\noindent  However (if there is a zero object $0 \in \mathcal{C}$ in $\mathrm{Loc}_{\mathcal{C}}$, or at least a terminal object)
  after regarding it inside $\mathrm{Loc}_{\mathcal{C}}$, then
  any such group representation may be ``decomposed'' into:
  
  (i) the underlying group $G$, 
  
  (ii) the underlying object 
  $\mathscr{V} \,\in\, \mathcal{C}$ and 
  
  (iii) the action itself, in that it fits into a pullback square of this form:
  \vspace{-2mm} 
  $$
    \begin{tikzcd}[row sep=small] 
      \mathscr{V}_{\mathrm{pt}}
      \ar[d]
      \ar[r]
      \ar[dr, phantom, "{ \scalebox{.7}{(pb)} }"]
      &
      \mathscr{V}_{\mathbf{B}G}
      \ar[d]
      \\
      0_{\mathrm{pt}}
      \ar[r]
      &
      0_{\mathbf{B}G}
    \end{tikzcd}
    \;\;\;\;\;\;\;
    \in
    \;\;\;
    \mathrm{Loc}_{\mathcal{C}}
    \,.
  $$

  \vspace{-2mm} 
\noindent  At least for $\mathcal{C}$ such that $\mathrm{Loc}_{\mathcal{C}}$ embeds continuously into an $\infty$-topos, 
such squares exhibit $\mathscr{V}_{\mathbf{B}G}$ as the {\it homotopy quotient} of $\mathscr{V}$ by its $G$-action, 
and the map $\mathscr{V}_{\mathbf{B}G} \xrightarrow{\;} 0_{\mathbf{B}G}$ as the $\mathscr{V}$-fiber bundle associated 
to the universal $G$-principal bundle $0_{\mathbf{E}G} \xrightarrow{\;} 0_{\mathbf{B}G}$; 
see \cite[\S 2.2]{OrbifoldCohomology}\cite[Prop. 0.2.1]{EquBundles}.
\end{example}

We are mainly interested in the specialization of Def. \ref{CategoryOfLocalSystemsWithCoefficientsInAnyCategory}
to the case that $\mathcal{C}$ is a cocomplete monoidal category such as $\mathrm{Mod}_{\mathbb{C}}$.
When $\mathcal{C}$ is cocomplete, then it is canonically tensored over $\mathrm{Set}$ 
and all base change operations $f^\ast$ \eqref{CategoryOfLocalSystemsInIntroduction} have left adjoints $f_!$ 
(by left Kan extension):
\vspace{-4mm} 
\begin{equation}
 \label{ExtraStructureOnLocalSystemsWithCocompleteCoefficients}
 \hspace{-1cm} 
 \mbox{
    $\mathcal{C}$ cocomplete
  }
  \hspace{1cm}
    \vdash
  \hspace{1cm}
    \begin{tikzcd}[row sep=-2pt, column sep=small]
      \mathrm{Set}
      \times
      \mathcal{C}
      \ar[
        rr,
        "{
          (\mbox{-})
          \cdot
          (\mbox{-})
        }"
      ]
      &&
      \mathcal{C}
      \\
  \scalebox{0.8}{$    (S,\,\mathscr{V}) $}
      &\longmapsto&
    \scalebox{0.8}{$    \underset{s \in S}{\coprod} \mathscr{V} $}
    \end{tikzcd}
  \hspace{.7cm}
  \mbox{and}
  \hspace{.7cm}
  \begin{tikzcd}[row sep=0pt, column sep=small]
    \mathrm{Grpd}
    \ar[rr, "{ \mathcal{C}^{(\mbox{-})} }"]
    &&
    \mathrm{Cat}_{\mathrm{adj}}
    \\
    \mathcal{X}
    \ar[d, "{ f }"]
    &\mapsto&
    \mathcal{C}^{\mathcal{X}}
    \ar[
      d,
      shift right=8pt,
      "{
        f_!
      }"{swap}
    ]
    \ar[
      from=d,
      shift right=8pt,
      "{
        f^\ast
      }"{swap}
    ]
    \ar[
      d,
      phantom,
      "{
        \scalebox{.7}{$\dashv$}
      }"
    ]
    \\[+30pt]
    \mathcal{Y}
    &\mapsto&
    \mathcal{C}^{\mathcal{Y}}
  \end{tikzcd}
\end{equation}
\begin{lemma}[Pushforward along quotient projection on $\mathbf{E}G$]
  \label{PushforwardAlongQuotientProjectionOnEG}
  If $\mathcal{C}$ is cocomplete,  then the push-forward operation \eqref{ExtraStructureOnLocalSystemsWithCocompleteCoefficients}
  along $q : \mathbf{E}G \xrightarrow{\;} \mathbf{B}G$ \eqref{ProjectionFromEGToBG} is given by
  \vspace{-2mm} 
  $$
    \begin{tikzcd}[row sep=0pt, column sep=small]
      \mathcal{C}^{\mathbf{E}G}
      \ar[
        rr,
        "{ q_! }"
      ]
      &&
      \mathcal{C}^{\mathbf{B}G}
      \\
      \mathscr{V}_{\mathbf{E}G}
      &\longmapsto&
      (
        G \cdot \mathscr{V}_{\mathrm{e}}
      )_{\mathbf{B}G}
      \mathrlap{
        \;\;\simeq\;
        (G \cdot 1)_{\mathbf{B}G}
        \,\boxtimes\,
        \mathscr{V}_{\{\mathrm{e}\}}
        \,.
      }
    \end{tikzcd}
  $$
\end{lemma}  

\noindent
  (On the right, 
  using the tensoring \eqref{ExtraStructureOnLocalSystemsWithCocompleteCoefficients},
  the $G$-action 
  \eqref{GActionsAsFunctors} is via the left multiplication action of $G$ on itself, which on the far right we are
  transparently re-expressing, when $\mathcal{C}$ is monoidal, through the external tensor product
  \eqref{ExternalTensorProductOfLocalSystems}.)
\begin{proof}
  We may check the $(q_! \dashv q^\ast)$ hom-isomorphism: 
  First, speaking equivalently in terms of group actions via Ex. \ref{GroupActionsAsLocalSystems},
  since $G \cdot \mathscr{V}_{\mathrm{e}}$ carries the free group action on the underlying object, the $G$-equivariant morphisms $G \cdot \mathscr{V}_{\mathrm{e}} \xrightarrow{\;} \mathscr{W}$ are in natural bijection with the underlying such morphisms $f : \mathscr{V}_{\!\!\mathrm{e}} \xrightarrow{\;} \mathscr{W}$ in $\mathcal{C}$.
  These, in turn, are in natural bijection with morphisms from $\mathscr{V}_{\mathbf{E}G}$ to $q^\ast \mathscr{W}$, namely with natural transformations given as follows:
  \vspace{-2mm} 
  $$
   \begin{tikzcd}[row sep=-4pt, column sep=small]
     \mathbf{E}G \quad 
     \ar[rr]
     &&
    \qquad  \mathcal{C}
     &
     \\[10pt]
     \mathrm{e}
     \ar[d]
     &\longmapsto&
     \mathscr{V}_{\!\!\mathrm{e}}
     \ar[
       d, 
       "{
         \mathscr{V}_{\!(\mathrm{e},g)}
       }"{swap}
     ]
     \ar[
       r,
       "{
         f
       }"
     ]
     &
     \mathscr{W}
     \ar[
       d,
       "{
         \mathscr{W}_{\!\mu(g',g^{-1})}
       }"
     ]
     \\[+30pt]
     g
     &\longmapsto&
     \mathscr{V}_{\!\!g}
     \ar[r, "{ \exists! }"]
     &
     \mathscr{W}
     \mathrlap{\,.}
   \end{tikzcd}
  $$
  
  \vspace{-2mm}
  \noindent Here the bottom morphisms clearly exist uniquely for all $g \in G$, thus establishing the claimed bijection.
\end{proof}

\medskip
\noindent
{\bf External tensor product on local systems.}
The following Prop. \ref{ExternalTensorProductOnLocalSystems} is fairly immediate (for more details compare \cite[\S 2.3]{SS26-Global}). In generalization of \eqref{DistributivityOfMonoidalStructure},
the following Def. \ref{GroupoidMonoidalCategory} is a lightweight version of the notion of {\it monoidal enriched  categories} which we use in this section here in order not to overburden the elementary discussion:
\begin{definition}
\label{GroupoidMonoidalCategory}
 A {\it $\mathrm{Grpd}$-monoidal category} is a monoidal category $(\mathcal{C}, \otimes, 1)$ equipped with 
  \begin{itemize}
    \item[{\bf (i)}]
      a $\mathrm{Grpd}$-tensoring 
      \eqref{GrpdTensoringOfCategoryWithHomotopyQuasiCoproducts}
      $\begin{tikzcd}
        \mathrm{Grpd}
        \times
        \mathcal{V}
        \ar[r, "{ (\mbox{-})\cdot(\mbox{-}) }"]
        &
        \mathcal{C}
      \end{tikzcd}$
    \item[{\bf (ii)}]
     natural isomorphism
     \vspace{-2mm} 
     \begin{equation}
       \label{StructureIsosForGrpdMonoidality}
       \left.
       \def\arraystretch{1.4}
       \begin{array}{l}
         \mathcal{X} \,\in\, \mathrm{Grpd},
         \\
         \mathscr{V},\,
         \mathscr{W}
         \,\in\,
         \mathcal{C}
       \end{array}
      \! \right\}
       \hspace{1.3cm}
         \vdash
       \hspace{1.3cm}
       \big(
       \mathcal{X}
       \cdot
       \mathscr{V}
       \big)
       \otimes
       \mathscr{W}
       \;\;
         \simeq
       \;\;
       \mathcal{X}
       \cdot
       \big(
       \mathscr{V}
       \otimes
       \mathscr{W}
       \big).
     \end{equation}
  \end{itemize}
\end{definition}
\begin{definition}[Homotopy quasi-distributive categories]
\label{HomotopyQuasiDistributiveCategories}
 We say that a $\mathrm{Grpd}$-monoidal category $(\mathcal{C},\, \otimes_{\mathcal{C}},\, 1_{\mathcal{C}} )$ (Def. \ref{GroupoidMonoidalCategory})
 which has homotopy quasi-coproducts (Def. \ref{HomotopyQuasiCoproducts}) is {\it homotopy quasi-distributive} if the tensor product is 
 compatible in each variable with the homotopy quasi-coproducts \eqref{HomotopyQuasiCoproductDiagram} in that the canonical comparison 
 maps are isomorphisms:
 \vspace{-2mm} 
 $$
   \begin{tikzcd}[row sep=-3pt]
   \underset{\longrightarrow}{\lim}
   \Big(
     \mathcal{X} 
       \xrightarrow{\;}
     \mathcal{C}
       \xrightarrow{
         \scalebox{.7}{$
           (-) 
             \otimes_{\mathcal{C}} 
           \mathscr{W}
         $}
       }
     \mathcal{C}
   \Big)
   \ar[
     r,
     "{}",
     "{ \sim }"{swap}
   ]
   &
   \Big(
   \underset{\longrightarrow}{\lim}
   \big(
     \mathcal{X} 
       \xrightarrow{\;}
     \mathcal{C}
   \big)   
   \Big)
   \otimes_{\mathcal{C}}
   \mathscr{W}\;,
   \\
   \underset{\longrightarrow}{\lim}
   \Big(
     \mathcal{X} 
       \xrightarrow{\;}
     \mathcal{C}
       \xrightarrow{
         \scalebox{.7}{$
           \mathscr{W}
             \otimes_{\mathcal{C}} 
           (-) 
         $}
       }
     \mathcal{C}
   \Big)
   \ar[
     r,
     "{}",
     "{ \sim }"{swap}
   ]
   &
   \mathscr{W}
   \otimes_{\mathcal{C}}
   \Big(
   \underset{\longrightarrow}{\lim}
   \big(
     \mathcal{X} 
       \xrightarrow{\;}
     \mathcal{C}
   \big)   
   \Big).
   \end{tikzcd}
 $$
\end{definition}
\begin{proposition}[External tensor product on local systems]
\label{ExternalTensorProductOnLocalSystems}
Given a cocomplete closed monoidal category $(\mathcal{C}, 1,  \otimes)$, the category of $\mathcal{C}$-valued local systems (Def. \ref{CategoryOfLocalSystemsWithCoefficientsInAnyCategory}) becomes a homotopy quasi-distributive category
(Def. \ref{HomotopyQuasiDistributiveCategories}) under the external tensor product
\vspace{-1mm} 
\begin{equation}
  \label{ExternalTensorProductOfLocalSystems}
  \begin{tikzcd}[row sep=0pt, column sep=small]
    \mathrm{Loc}_{\mathcal{C}}
    \times
    \mathrm{Loc}_{\mathcal{C}}
    \ar[
      rr,
      "{
        \boxtimes
      }"
    ]
    &&
    \mathrm{Loc}_{\mathcal{C}}
    \\
    \big(
      \mathscr{V}_{\mathcal{X}}
      ,\,
      \mathscr{W}_{\mathcal{Y}}
    \big)
    &\longmapsto&
    \Big(
      \big(
        (\mathrm{pr}_{\mathcal{X}})^\ast
        \mathscr{V}
      \big)
      \otimes
      \big(
        (\mathrm{pr}_{\mathcal{Y}})^\ast
        \mathscr{W}
      \big)
    \Big)_{ \mathcal{X} \times \mathcal{Y} }
  \end{tikzcd}
\end{equation}

\vspace{-2mm} 
\noindent  with respect to the canonical  $\mathrm{Grpd}$-tensoring
  \eqref{GrpdTensoringOfLocalSystems}, hence in particular such that the inclusion \eqref{IncludingConstantLocalSystems} is strong monoidal
  \begin{equation}
    \label{InclusionOfPointSystemsStrongMonoidal}
    \iota\big(
      \mathscr{V}
      \otimes
      \mathscr{W}
    \big)
    \;\;
      \simeq
    \;\;
    \iota(\mathscr{V})
    \boxtimes
    \iota(\mathscr{W})\;.
  \end{equation}
\end{proposition}
\begin{proof}
  Observing that the $\mathrm{Grpd}$-tensoring
  \eqref{GrpdTensoringOfLocalSystems} equals the restriction of the external tensor product to unit systems
  \begin{equation}
    \label{GrpdTensoringOnLocalSystems}
    \left.
    \def\arraystretch{1.4}
    \begin{array}{l}
      \mathcal{X}
      \,\in\,
      \mathrm{Grpd}
      \\
      \mathscr{W}_{\!\mathcal{X}}
      \,\in\,
      \mathrm{Loc}_{\mathcal{C}}
    \end{array}
    \right\}
    \hspace{1.2cm}
    \vdash
    \hspace{1.2cm}
    \mathcal{X}
    \cdot
    \mathscr{W}_{\mathcal{Y}}
    \;\;
    =
    \;\;
    1_{\mathcal{X}}
    \,\boxtimes\,
    \mathscr{W}_{\mathcal{Y}}    
    \,,
  \end{equation}  

  \vspace{-2mm} 
  \noindent
  the  $\mathrm{Grpd}$-monoidality structure
  \eqref{StructureIsosForGrpdMonoidality}
  is given identically by  \eqref{GrpdTensoringOnLocalSystems}. Moreover, under the assumption that the tensor product 
  $\otimes$ on $\mathcal{C}$ is closed, hence a left adjoint in each variable and as such colimit-preserving, it follows 
  that also the external tensor product preserves all colimits in each variable (we spell this out in greater generality
  in \cite[Prop. 2.36]{SS26-Global}), hence 
  in particular it preserves quasi-coproducts.
\end{proof}

\medskip
\noindent
{\bf Local systems as the free homotopy quasi-coproduct completion of vector spaces.}

\smallskip 
The following Lemmas \ref{EveryLocalSystemIsCoproductOfGrpdTensoringOfGroupRepresentation}, \ref{EveryGroupRepresentationILocalSystemsIsQuasiColimit} 
should be thought of as saying that in $\mathrm{Loc}_{\mathcal{C}}$ every object is {\it homotopy equivalent} to a coproduct of
{\it homotopy quotients}, where the tensoring with the contractible groupoids $\mathrm{CoDisc}(S)$, $\mathbf{E}G$ 
\eqref{ContractabilityOfCodiscreteGroupoids} is used to model these homotopy-theoretic notions in 1-category theoretic terms.

\begin{lemma}[Every local system is coproduct of Grpd tensoring of group representation]
  \label{EveryLocalSystemIsCoproductOfGrpdTensoringOfGroupRepresentation}
  Given an object $\mathscr{V}_{\mathcal{X}} \,\in\, \mathrm{Loc}_{\mathcal{C}}$ \eqref{LocWithCoefficientsInAnyC}, it is isomorphic 
  to a coproduct of tensorings \eqref{GrpdTensoringOfLocalSystems} with coskeletal groupoids (Ex. \ref{BasicExampleOfGroupoids})
  of group representations \eqref{GroupActionAsLocalSystems}:
  \vspace{-2mm} 
  $$
    \mathscr{V}_{\mathcal{X}}
    \;\in\;
    \mathrm{Loc}_{\mathcal{C}}
    \hspace{1.2cm}
      \vdash
    \hspace{1.2cm}
    \mathscr{V}_{\mathbf{B}G}
    \;\underset{\mathrm{iso}}{\simeq}\;
    \underset{
      i \in \pi_0(\mathcal{X})
    }{\coprod}
    \Big(
      \mathrm{CoDisc}\big(
        \mathrm{Obj}(\mathcal{X}_i)
      \big)
      \cdot
      \big(
       \iota_{x_i}^\ast
       \mathscr{V}
      \big)_{\mathbf{B}\mathcal{X}(x_i,x_i)}
    \Big)
    \,.
  $$
\end{lemma}
\begin{proof}
  On underlying groupoids, this is the statement of Lem. \ref{ConnectedGroupoidIsomorphicToDeloopingTimesCodiscrete}.
  Therefore, it is sufficient to see that any local system $\mathscr{V}_{(-)}$ on $\mathrm{CoDisc}(X) \cdot \mathbf{B}G$ 
  is isomorphic to one pulled back from $\mathbf{B}G$, via the following natural isomorphism
\vspace{-2mm} 
  $$
    \begin{tikzcd}[row sep=1pt, column sep=small]
      \mathrm{CoDisc}(X)
      \times
      \mathbf{B}G
      \quad 
      \ar[rr]
      &&
     \qquad \qquad  \mathcal{C}
      \\[5pt]
      (x,\mathrm{pt})
      \ar[
        d,
        "{
          \scalebox{.7}{$
            \scalebox{1.2}{$($}
              (x,x'),g
            \scalebox{1.2}{$)$}
          $}
        }"{swap}
      ]
      &\longmapsto \quad &
         \mathscr{V}_{\!x}
      \ar[
        r, 
        "{
          \mathscr{V}_{\!
            \scalebox{.7}{$
              \scalebox{1.2}{$($}(x, x_i),\, \mathrm{e}\scalebox{1.2}{$)$}
            $}
         }
        }"{swap}
      ]
      \ar[
        d,
        "{
          \mathscr{V}_{\!
            \scalebox{.7}{$
              \scalebox{1.2}{$($}(x, x'),\, g
              \scalebox{1.2}{$)$}
            $}
          }
        }"{swap}
      ]
      &[20pt]
      \mathscr{V}_{\!x_i}
      \ar[
        d,
        "{
          \mathcal{V}_{\! 
            \scalebox{.7}{$
            \scalebox{1.2}{$($}\mathrm{id}_{x_i},\, g\scalebox{1.2}{$)$}
            $}
          }
        }"
      ]
      \\[+20pt]
      (x',g)
      &\longmapsto \quad&
      \mathscr{V}_{\!x}
      \ar[
        r, 
        "{
          \mathscr{V}_{\!
            \scalebox{.7}{$
            \scalebox{1.2}{$($}(x', x_i),\, \mathrm{e}\scalebox{1.2}{$)$}
            $}
         }
        }"
      ]
      &[20pt]
      \mathscr{V}_{\!x_i}
      \mathrlap{\,,}
    \end{tikzcd}
  $$
  for any choice of basepoint 
  $x_i \,\in\, X$.
\end{proof}

\begin{lemma}[Group representation as homotopy quasi-coproduct]
  \label{EveryGroupRepresentationILocalSystemsIsQuasiColimit}
  Given a cocomplete category $\mathcal{C}$, and $G \,\in\, \mathrm{Grp}$, any group action $\rho : \mathbf{B}G \to \mathcal{C}$ 
  \eqref{GActionsAsFunctors} is, when regarded as an object of $\mathrm{Loc}_{\mathcal{C}}$ (Rem. \ref{GroupActionsAsLocalSystems}), 
  a homotopy quasi-coproduct (Def. \ref{HomotopyQuasiCoproducts}) in that it is the colimit in $\mathrm{Loc}_{\mathcal{C}}$ 
  over the following diagram:
  $$
    \begin{tikzcd}[row sep=1pt, column sep=small]
      \mathbf{B}G
      \ar[
        rr, 
        "{  }"
      ]
      &&
      \mathrm{Loc}_{\mathcal{C}}
      \\
      \mathrm{pt}
      \ar[
        d,
        "{ g }"
      ]
      &\longmapsto \quad &
      (\mathbf{E}G)
        \cdot
      \mathscr{V}_{\mathrm{pt}}
      \ar[
        d,
        "{
          \mu(\mbox{-},\,g^{-1})
          \,\cdot\,
          \rho(g)
        }"
      ]
      \\[+20pt]
      \mathrm{pt}
      &\longmapsto \quad &
      (\mathbf{E}G)
        \cdot
      \mathscr{V}_{\mathrm{pt}}      
    \end{tikzcd}
  $$
\end{lemma}
\begin{proof}
  By the assumption that $\mathcal{C}$ is cocomplete, the colimit exists and is given (Prop. \ref{ColimitsInAGrothendieckConstruction})
  on underlying groupoids by the quotient coprojection
  $
    q \,:\,
    \mathbf{E}G
    \xrightarrow{\;}
    \mathbf{B}G
  $
  \eqref{ProjectionFromEGToBG} and on $\mathcal{C}$-components by the colimit over the following diagram:
  $$
    \begin{tikzcd}[row sep=1pt, column sep=small]
   \phantom{AAA}   \mathbf{B}G \qquad 
      \ar[rr]
      &&
     \qquad \qquad \qquad  \qquad \mathcal{C}
      \\[6pt]
      \mathrm{pt}
      \ar[
        d,
        "{ g }"{swap}
      ]
      &\longmapsto \quad&
      q_!
      \big(
        \mathbf{E}G
        \cdot
        \mathscr{V}_{\mathrm{pt}}
      \big)
      \ar[
        d,
        "{
          q_!
          \big(
            \mu(-,\, g^{-1})
            \cdot
            \rho(g)_{\mathrm{pt}}
          \big)
        }"
      ]
      &\simeq&
      (G \cdot 1)_{\mathbf{B}G}
      \,\boxtimes\,
      \mathscr{V}_{\mathrm{pt}}
      \ar[
        d,
        "{
          \big(
            \mu(-,\, g^{-1})
            \cdot 
            1
          \big)_{\mathbf{B}G}
          \boxtimes
          \mathscr{V}_{\!\rho(g)}
        }"
      ]
      \\[+35pt]
      \mathrm{pt}
      &\longmapsto \quad&
      q_!
      \big(
        \mathbf{E}G
        \cdot
        \mathscr{V}_{\mathrm{pt}}
      \big)      
      &\simeq&
      (G \cdot 1)_{\mathbf{B}G}
      \boxtimes
      \mathscr{V}_{\mathrm{pt}}
      \,,
    \end{tikzcd}
  $$
  where on the right we used Lem. \ref{PushforwardAlongQuotientProjectionOnEG}. Since this colimit is taken in a functor category, 
  $\mathcal{C}^{\mathbf{B}G}$, it is computed on underlying objects in $\mathcal{C}$,  where it is following cocone
  $$
    \begin{tikzcd}[row sep=1pt, column sep=huge]
      G \cdot \mathscr{V}
      \ar[
        rr,
        "{
          \mu(-,\, g^{-1})
          \cdot
          \rho(g)
        }"
      ]
      \ar[
        dr,
        "{\rho}"{swap}
      ]
      &&
      G \cdot \mathscr{V}
      \ar[
        dl,
        "{\rho}"
      ]
      \\
      &
      \mathscr{V}
    \end{tikzcd}
  $$
  and the induced action of morphisms in $\mathbf{B}G$ is thus given by the universal dashed morphism in
  $$
    \begin{tikzcd}
      G \cdot \mathscr{V}
      \ar[d, "{ \rho }"]
      \ar[rr, "{ \mu(g,-) }"]
      &&
      G \cdot \mathscr{V}
      \ar[d, "{\rho}"]
      \\
      \mathscr{V}
      \ar[
        rr,
        dashed,
        "{
          \rho(g)
        }"
      ]
      &&
      \mathscr{V}
      \,.
    \end{tikzcd}
  $$
  Under the equivalence \eqref{GroupActionAsLocalSystems},
  this is the object $\mathscr{V}_{\mathbf{B}G} \,\in\, \mathrm{Loc}_{\mathcal{C}}$, as claimed.
\end{proof}
In conclusion:

\begin{theorem}[Quasi-coproduct completion]
  Given a cocomplete category $\mathcal{C}$, then $\mathrm{Loc}_{\mathcal{C}}$ (Def. \ref{CategoryOfLocalSystemsWithCoefficientsInAnyCategory}) 
  is its free homotopy quasi-coproduct completion (Def. \ref{FreeHomotopyQuasiCoproductCompletion}).
\end{theorem}
\begin{proof}
  First, due to the assumption that $\mathcal{C}$ is cocomplete, so is $\mathrm{Loc}_{\mathcal{C}}$ (by Prop. \ref{ColimitsInAGrothendieckConstruction}) 
  and hence, in addition to the $\mathrm{Grpd}$-tensoring \eqref{GrpdTensoringOfLocalSystems}, it in particular has all homotopy quasi-coproducts \eqref{HomotopyQuasiCoproductDiagram}.
  
  Now every object of $\mathrm{Loc}_{\mathcal{C}}$ is isomorphic to a coproduct of $\mathrm{Grpd}$-tensorings of group representations 
  (Lem. \ref{EveryLocalSystemIsCoproductOfGrpdTensoringOfGroupRepresentation}) and every group representation is a homotopy quasi-coproduct
  of a constant local system (Lem. \ref{EveryGroupRepresentationILocalSystemsIsQuasiColimit}). Therefore any functor out of 
  $\mathrm{Loc}_{\mathcal{C}}$ which preserves the $\mathrm{Grpd}$-tensoring and homotopy quasi-coproducts is already fixed by 
  its restriction to constant local systems \eqref{IncludingConstantLocalSystems}.
\end{proof}

\medskip

\noindent
{\bf Pushout-characterization of the external tensor product on local systems.}
In generalization of Def. \ref{AnyMonCat}:
\begin{definition}
\label{AnyMonCanHQ}
  Write 
  \fbox{$\mathrm{AnyMonCat}^{{}^{hq}}$} for the Grothendieck construction on the 
  following diagram of forgetful functors
\begin{equation}
  \label{CommutingSquareOfCategoriesOfCategoriesHQ}
  \begin{tikzcd}
    (\mathcal{C}, \otimes)
    \ar[rrr, phantom, "{\xmapsfrom{\qquad \qquad \qquad \qquad \quad}}"]
    \ar[ddd, phantom, "{ \xmapsto{\qquad \quad} }"{sloped}]
    &[-35pt]
    &
    &[-45pt]
    (\mathcal{C}, \sqcup^{{}^{hq}}, \otimes)
    \ar[ddd, phantom, "{ \xmapsto{\qquad \quad} }"{sloped}]
    \\[-20pt]
    &
    \mathrm{MonCat}
    \ar[from=r]
    \ar[d]
    &
   \colorbox{lightgray}{$ \mathrm{DistMonCat}^{{}^{hq}}$}
    \ar[d]
    \\
    &
    \mathrm{Cat}
    \ar[from=r]
    &
    \mathrm{CoCartCat}^{{}^{hq}}
    \\[-20pt]
    \mathcal{C}
    \ar[rrr, phantom, "{\xmapsfrom{\qquad \qquad \qquad \qquad \qquad}}"]
    &&&
    (\mathcal{C}, \sqcup^{{}^{hq}})
    \mathrlap{\,,}
  \end{tikzcd}
\end{equation}
where now $\mathrm{CoCartCat}^{{}^{hq}}$ is from Def. \ref{HomotopyQuasiCoproducts} and $\mathrm{DistMonCat}^{{}^{hq}}$ from Def. \ref{HomotopyQuasiDistributiveCategories}.
\end{definition}

Now we are ready to state and prove the first homotopy-theoretic generalization of the pushout theorem \ref{ThePushoutTheoremOverSets}:

\begin{theorem}[Pushout-characterization of the external tensor product on local systems]
\label{PushoutCharacterizationOfExternalTensorProductOnLocalSystems}
The following is a pushout diagram in $\mathrm{AnyMonCat}^{{}^{hq}}$ (Def. \ref{AnyMonCanHQ}), where the structure on the right is from Prop. \ref{ExternalTensorProductOnLocalSystems}:
\vspace{-2mm} 
\begin{equation}
  \label{ThePushoutDiagramForLocalSystems}
  \begin{tikzcd}
   \big(
     \mathrm{Mod}_{\mathbb{C}}
     ,\,
     \otimes
   \big)
   \ar[rr, "{ \iota }"]
   \ar[
     drr, 
     phantom, 
     "{
       \scalebox{.7}{\rm{(po)}}
     }"
   ]
   &&
  \colorbox{lightgray}{$ \big(
     \mathrm{Loc}_{\mathbb{C}}
     ,\,
     \sqcup^{{}^{hq}}
     ,\,
     \boxtimes
   \big)
   $}
   \\
   \mathrm{Mod}_{\mathbb{C}}
   \ar[
     rr,
     "{
       \iota
     }"
   ]
   \ar[u]
   &&
   \big(
     \mathrm{Loc}_{\mathbb{C}}
     ,\,
     \sqcup^{{}^{hq}}
   \big)
   \ar[u]
  \end{tikzcd}
\end{equation}
\end{theorem}
\begin{proof}
As before in the proof of Thm. \ref{ThePushoutTheoremOverSets}, we demonstrate the unique existence of a dashed arrow, 
now in $\mathrm{AnyMonCat}^{{}^{hq}}$, given a solid diagram as shown here:
\vspace{-2mm}
\begin{equation}
  \label{CoconeDiagramTowardsProvingPushutForLocalSystems}
  \begin{tikzcd}
   &&&
   \big(
     \mathcal{C}
     ,\,
     \sqcup^{{}^{hq}}
     ,\,
     \otimes_{{}_{\mathcal{C}}}
   \big)
   \\
   \big(
     \mathrm{Mod}_{\mathbb{C}}
     ,\,
     \otimes
   \big)
   \ar[
     urrr,
     bend left=15,
     "{ F \circ \iota }"
   ]
   \ar[rr, "{ \iota }"]
   &&
   \big(
     \mathrm{Loc}_{\mathbb{C}}
     ,\,
     \sqcup^{{}^{hq}}
     ,\,
     \boxtimes
   \big)
   \ar[
     ur,
     dashed,
   ]
   \\
   \mathrm{Mod}_{\mathbb{C}}
   \ar[
     rr,
     "{
       \iota
     }"
   ]
   \ar[u]
   &&
   \big(
     \mathrm{Loc}_{\mathbb{C}}
     ,\,
     \sqcup^{{}^{hq}}
   \big)
   \ar[u]
   \ar[
     uur,
     bend right=20,
     "{ F }"{swap}
   ]
  \end{tikzcd}
\end{equation}
\vspace{-2mm} 

\noindent 
The proof proceeds along the same  lines as before in Thm. \ref{ThePushoutTheoremOverSets}, now using the stronger 
homotopy quasi-distributivity property to factor out the richer homotopical quasi-coproduct structure of the objects. 
Namely, we need to see that the given functor $F$ already intertwines the external tensor product on $\mathrm{Loc}_{\mathcal{C}}$ 
with the tensor product on $\mathcal{C}$, and this is obtained by the following sequence of natural isomorphisms:
\vspace{-5mm} 
$$
\adjustbox{scale=.95}{$
    \def\arraystretch{2.2}
    \begin{array}{ll}
      F
      \big(
      \mathscr{V}_{\!\!\mathcal{X}}
      \boxtimes
      \mathscr{W}_{\!\!\mathcal{Y}}      
      \big)
      \\
      \;\simeq\;
      F
      \bigg(\!
      \Big(
      \coprod_i
      \,
      \mathrm{CoDisc}
      \big(
        X_i
      \big)
      \cdot
      \big(
        \iota_{x_i}^\ast\mathscr{V}
      \big)_{\mathbf{B}G_i}
      \Big)
      \boxtimes
      \Big(
      \coprod_{j}
      \,
      \mathrm{CoDisc}
      \big(
        Y_j
      \big)
      \cdot
      \big(
        \iota_{y_j}^\ast\mathscr{W}
      \big)_{\mathbf{B}G_{j}}
      \Big)
      \!\bigg)
      &
      \proofstep{
        by Lem. \ref{EveryLocalSystemIsCoproductOfGrpdTensoringOfGroupRepresentation}
      }
      \\
      \;\simeq\;
      F
      \bigg(\!
      \Big(
      \coprod_i
      \,
      \mathrm{CoDisc}
      \big(
        X_i
      \big)
      \cdot
      \underset{\underset{\mathbf{B}G_i}{\longrightarrow}}{\lim}
      \big(
        (\mathbf{E}G_i) 
        \cdot
        \mathscr{V}_{x_i}
      \big)
      \Big)
      \boxtimes
      \Big(
      \coprod_{j}
      \,
      \mathrm{CoDisc}
      \big(
        Y_j
      \big)
      \cdot
      \underset{\underset{\mathbf{B}G_i}{\longrightarrow}}{\lim}
      \big(
        (\mathbf{E}G_j) 
        \cdot
        \mathscr{W}_{y_j}
      \big)
      \Big)
      \!\bigg)
      &
      \proofstep{
        by Lem. \ref{EveryGroupRepresentationILocalSystemsIsQuasiColimit}
      }
      \\
      \;\simeq\;
      F
      \bigg(\!
      \Big(
      \coprod_{i,j}
      \,
      \mathrm{CoDisc}
      \big(
        X_i
        \times 
        Y_j
      \big)
      \cdot
      \underset{
        \underset{
          \mathclap{
          \mathbf{B}(G_i \times G_j)
          }
        }{\longrightarrow}
      }{\lim}
      \Big(
        \mathbf{E}(G_i \times G_j)
        \cdot
        \big(
          \iota(\mathscr{V}_{x_i})
          \boxtimes
          \iota(\mathscr{V}_{y_j})
        \big)
      \Big)
      \!\bigg)
      &
      \proofstep{
        by Prop. \ref{ExternalTensorProductOnLocalSystems}
      }
      \\
      \;\simeq\;
      F
      \bigg(\!
      \Big(
      \coprod_{i,j}
      \,
      \mathrm{CoDisc}
      \big(
        X_i
        \times 
        Y_j
      \big)
      \cdot
      \underset{
        \underset{
          \mathclap{
          \mathbf{B}(G_i \times G_j)
          }
        }{\longrightarrow}
      }{\lim}
      \Big(
        \mathbf{E}(G_i \times G_j)
        \cdot
        \iota
        \big(
          \mathscr{V}_{x_i}
          \otimes
          \mathscr{W}_{x_i}
        \big)
      \Big)
      \!\bigg)
      &
      \proofstep{
        by 
        \eqref{InclusionOfPointSystemsStrongMonoidal}
      }
      \\
      \;\simeq\;
      \coprod_{i,j}
      \,
      \mathrm{CoDisc}
      \big(
        X_i
        \times 
        Y_j
      \big)
      \cdot
      \underset{
        \underset{
          \mathclap{
          \mathbf{B}(G_i \times G_j)
          }
        }{\longrightarrow}
      }{\lim}
      \;
      \bigg(
        \mathbf{E}(G_i \times G_j)
        \cdot
        F
        \Big(
        \iota
        \big(
          \mathscr{V}_{x_i}
          \otimes
          \mathscr{W}_{x_i}
        \big)
        \Big)
      \!\bigg)
      &
      \proofstep{
        \hspace{-12pt}
        \def\arraystretch{.9}
        \begin{tabular}{l}
        as $F$ preserves
        \\
        hq-coproducts
        \eqref{CoconeDiagramTowardsProvingPushutForLocalSystems}
        \end{tabular}
      }
      \\
      \;\simeq\;
      \coprod_{i,j}
      \,
      \mathrm{CoDisc}
      \big(
        X_i
        \times 
        Y_j
      \big)
      \cdot
      \underset{
        \underset{
          \mathclap{
          \mathbf{B}(G_i \times G_j)
          }
        }{\longrightarrow}
      }{\lim}
      \;
      \bigg(
        \mathbf{E}(G_i \times G_j)
        \cdot
        F
        \Big(
        \iota
        \big(
          \mathscr{V}_{x_i}
        \big)
        \Big)
          \otimes_{\mathcal{C}}
        F
        \Big(
        \iota
        \big(         
          \mathscr{W}_{x_i}
        \big)
        \Big)
      \!\bigg)
      &
      \proofstep{
        \hspace{-12pt}
        \def\arraystretch{.9}
        \begin{tabular}{l}
        as $F \circ \iota$ preserves
        \\
        tensor products
        \eqref{CoconeDiagramTowardsProvingPushutForLocalSystems}
        \end{tabular}
      }
    \\
    \;\simeq\;
    \bigg(
    \coprod_i
    \mathrm{CoDisc}(X_i)
    \cdot
    \underset{\underset{\mathbf{B}G_i}{\longrightarrow}}{\lim} 
    \Big(
      \mathbf{E}G_i
      \cdot
      F
      \big(
      \iota
      (\mathscr{V}_{x_i})
      \big)
    \Big)
   \! \bigg)
    \otimes_{\mathcal{C}}
    \bigg(
    \coprod_j
    \mathrm{CoDisc}(Y_j)
    \cdot
    \underset{\underset{\mathbf{B}G_j}{\longrightarrow}}{\lim} 
    \Big(
      \mathbf{E}G_j
      \cdot
      F
      \big(
      \iota
      (\mathscr{W}_{x_j})
      \big)
    \Big)
   \! \bigg)
    &
    \proofstep{
      \hspace{-12pt}
      \def\arraystretch{.9}
      \begin{tabular}{l}
        as $\otimes_{\mathcal{C}}$ preserves
        \\
        hq-coproducts
        \eqref{CoconeDiagramTowardsProvingPushutForLocalSystems}
        \end{tabular}
    }
    \\
    \;\simeq\;
    F
    \Big(
    \coprod_i
    \mathrm{CoDisc}(X_i)
    \cdot
    \underset{\underset{\mathbf{B}G_i}{\longrightarrow}}{\lim} 
    \big(
      \mathbf{E}G_i
      \cdot
      \iota
      (\mathscr{V}_{x_i})
    \big)
    \Big)
    \otimes_{\mathcal{C}}
    F
    \Big(
    \coprod_j
    \mathrm{CoDisc}(Y_j)
    \cdot
    \underset{\underset{\mathbf{B}G_j}{\longrightarrow}}{\lim} 
    \big(
      \mathbf{E}G_j
      \cdot
      \iota
      (\mathscr{W}_{x_j})
    \big)
    \Big)
    &
    \proofstep{
      \hspace{-12pt}
      \def\arraystretch{.9}
      \begin{tabular}{l}
        as $F$ preserves
        \\
        hq-coproducts
        \eqref{CoconeDiagramTowardsProvingPushutForLocalSystems}
        \end{tabular}
    }
    \\
    \;\simeq\;
    F\big(
      \mathscr{V}_{\mathcal{X}}
    \big)
    \otimes_{\mathcal{C}}
    F\big(
      \mathscr{W}_{\mathcal{Y}}
    \big)
    &
      \proofstep{
        by Prop. \ref{ExternalTensorProductOnLocalSystems}.
      }
  \end{array}
$}
$$

\vspace{-.8cm}
\end{proof}

\section{Conclusion and Outlook}
\label{ConclusionAndOutlook}

Our Thm. \ref{PushoutCharacterizationOfExternalTensorProductOnLocalSystems} provides an answer to the question \eqref{TheSoughtAfterPushout} 
for flat vector bundles over general parameter base spaces, thereby making contact to the categorical formulation of parameterized quantum processes that is recently being used in the functional quantum programming literature (cf. \cite{QS}). 

But there is a yet more general and powerful form of this theorem, which however requires mathematical technology beyond the scope of this note. Namely, flat vector bundles are only sensitive to the homotopy 1-type of their base space, hence to its fundamental group. This is the reason that the theory of ordinary local systems can be described entirely in ordinary category/groupoid theory, 
as we have done in this subsection. But more generally we want to generalize further to \emph{higher} flat vector bundles (\emph{flat $\infty$-vector bundles}) hence to \emph{higher} local systems (\emph{$\infty$-local systems}) which are sensitive to 
the full homotopy type of the underlying parameter spaces. This is relevant for describing more fine-grained aspects of topological quantum processes, as discussed in \cite{TQP,Reality}. 

\medskip

But doing so requires the larger toolbox of ``homotopical category theory'' and for the result to be tractable in practice we need strong tools 
from \emph{simplicial model category theory}. Using these, we can give a discussion of {\it simplicial local systems} which, conceptually,
closely parallels the discussion of ordinary local systems we just gave, while constituting a model for these $\infty$-local systems. While this is well beyond the scope of the present note, we have laid out the relevant discussion in \cite{SS26-Global}.

Using the main theorem there, our Thm. \ref{PushoutCharacterizationOfExternalTensorProductOnLocalSystems} has a straightforward generalization, providing a comprehensive answer to equation \eqref{TheSoughtAfterPushout} in full homotopy theory.

\section{Appendix: Some definitions and facts}
\label{SomeDefinitionsAndFacts}

For reference in the main text, we record some basic facts from the literature
and highlight some immediate examples that we use in the main text.

\medskip

\noindent
{\bf Categories, groupoids and simplicial enrichment.} We use basic concepts from category theory (e.g. \cite{MacLane97}) and enriched
category theory (e.g. \cite{Kelly82}).

\begin{definition}[Categories and groupoids]
 \label{GroupoidsAndCategories}
 
With respect to any fixed Grothendieck universe $\mathfrak{U}$ of sets \cite[\S 3.2]{Schubert72} of which we assume at least two 
$\mathfrak{U} < \mathfrak{U}'$, cf. e.g. \cite[p. 4]{Levy18}\cite[p. 18]{Shulman08}:
 
We write
 \vspace{-2mm} 
 \begin{equation}
   \label{LocalizationAndCore}
   \begin{tikzcd}[column sep=large]
     \mathrm{Grpd}
     \ar[from=r, shift right=12pt, "{ \mathrm{Loc} }"{swap}]
     \ar[r, hook]
     \ar[from=r, shift left=12pt, "{ \mathrm{core} }"]
     \ar[r, phantom, shift left=6pt, "{ \scalebox{.7}{$\bot$} }"]
     \ar[r, phantom, shift right=6pt, "{ \scalebox{.7}{$\bot$} }"]
     &
     \mathrm{Cat}
   \end{tikzcd}
 \end{equation}
 \vspace{-2mm} 

\noindent for the full inclusion of the 1-category of $\mathfrak{U}$-small groupoids into the 1-category of $\mathfrak{U}$-small categories (e.g. \cite[\S 3]{Schubert72}), with left adjoint $\mathrm{Loc}$ being the {\it localization}-construction that universally inverts all morphisms \cite[\S 1.5.4]{GabrielZisman67}.

(The $\mathfrak{U}'$-small categories are called {\it $\mathfrak{U}$-large}, whence $\mathrm{Cat}$ in this case is ``very large'' \cite[p. 18]{Shulman08}.)
\end{definition}

\noindent
{\bf Pseudofunctors and the Grothendieck construction.}
Given a ``coherent system of categories and functors`` -- namely a pseudo-functorial diagram of categories, Def. \ref{Pseudofunctor} 
below -- the {\it Grothendieck construction} (Def. \ref{GrothendieckConstruction} below) is the natural way of merging this data into 
a single category whose morphisms subsume those of the individual categories but also transfers from one category to the other along one of the given functors.

\begin{definition}[Pseudofunctor {\cite[\S A.1]{Grothendieck60}, cf. \cite[Def. 3.10]{Vistoli05}}] 
  \label{Pseudofunctor}
  $\,$ \newline 
 \noindent {\bf (i)} For $\mathcal{B}$ a category, a {\it covariant pseudofunctor} to $\mathrm{Cat}$
 \vspace{-2mm} 
  \begin{equation}
    \label{CovariantPseudofunctor}
    \begin{tikzcd}[row sep=-5pt, column sep=20pt]
      \mbox{\bf C}_{(-)}
      \;\colon\;
      &
      \mathcal{B}
      \ar[rr]
      &&
      \mathrm{Cat}
      \\
      &
      X_1 
      \ar[d, "{f}"]
      &\longmapsto&
      \mbox{\bf C}_{X_1}
      \ar[d, "{f_!}"]
      \\[+20pt]
      &
      X_2 
      &\longmapsto& 
      \mbox{\bf C}_{X_2}
    \end{tikzcd}
  \end{equation}

  \vspace{-2mm} 
\noindent  is an assignment that sends 
  \begin{itemize} 
 \item  each object $B \in \mathrm{Obj}(\mathcal{B})$ to a category $\mathbf{C}_B$, 
 \item  each morphism $f\colon X_0 \to X_1$ to a functor $f_! \,\colon\,\mathbf{C}_{X_0} \to \mathbf{C}_{X_1}$, 
 \item  each pair of composable morphisms $X_0 \overset{f_{01}}{\to} X_1 \overset{f_{12}}{\to} X_2$ to a natural 
  isomorphism $(f_{12})_! \circ (f_{01})_! \Rightarrow (f_{12} \circ f_{01})_!$ 
  \vspace{-2mm} 
  \begin{equation}
    \label{CompositionalCoherenceIsomorphism}
   \begin{tikzcd}
     & X_1
     \ar[
       dr, 
       "{ f_{12} }",
       "{\ }"{swap, pos=.1, name=s}
     ]
     \\
     X_0
     \ar[ur, "{f_{01}}"]
     \ar[
       rr, 
       "{ f_{02} }"{swap},
       "{\ }"{name=t}
      ]
     &&
     X_2
     \\
     \ar[from=s, to=t, shorten=3pt, equals]
   \end{tikzcd}
   \quad 
   \longmapsto
   \quad 
   \begin{tikzcd}
     & 
     \mathbf{C}_{X_1}
     \ar[
       dr, 
       "{ (f_{12})_! }",
       "{\ }"{swap, pos=.1, name=s}
     ]
     \\
     \mathbf{C}_{X_0}
     \ar[ur, "{ (f_{01})_! }"]
     \ar[
       rr, 
       "{ (f_{02})_! }"{swap},
       "{\ }"{name=t}
      ]
     &&
     \mathbf{C}_{X_2}
     \\
     \ar[
       from=s, 
       to=t, 
       shift left=10pt,
       shorten=3pt, 
       Rightarrow, 
       "{ \mu_{f_{01}, f_{12}} }"{swap},
       "{\sim}"{sloped, swap }
     ]
   \end{tikzcd}
  \end{equation}

  \vspace{-1.1cm} 
 \item  and, finally, each 
  identity morphism $\mathrm{id}_X : X \to X$ to a natural isomorphism $(\mathrm{id}_X)_! \Rightarrow \mathrm{id}_{\mathbf{C}_X}$ 
  \vspace{-2mm} 
  $$
    \begin{tikzcd}
      X
      \ar[rr, "{ \mathrm{id}_X }"]
      &&
      X
    \end{tikzcd}
    \;\;\;\;
    \mapsto
    \;\;\;\;
    \begin{tikzcd}[row sep=small]
      \mathbf{C}_X
      \ar[
        rr, 
        bend left=50, 
        "{ (\mathrm{id}_B)_! }",
        "{\ }"{name=s, swap}
      ]
      \ar[
        rr, 
        bend left=00, 
        "{\ }"{name=t},
        "{ \mathrm{id}_{\mathbf{C}_X} }"{swap}
      ]
      &&
      \mathbf{C}_X
      \ar[from=s, to=t, Rightarrow, "{\sim}"{sloped}]
    \end{tikzcd}    
  $$

  \vspace{-2mm} 
\noindent  such that these natural isomorphisms satisfy evident associativity and unitality coherences.
\end{itemize} 
 \noindent {\bf (ii)}   Similarly, a contravariant pseudofunctor is such a pseudofunctor on the opposite category $\mathcal{B}^{\mathrm{op}}$.
 \vspace{-2mm} 
  \begin{equation}
    \label{ContravariantPseudofunctor}
    \begin{tikzcd}[row sep=-5pt, column sep=20pt]
      \mbox{\bf C}_{(-)}
      \;\colon\;
      &
      \mathcal{B}^{\mathrm{op}}
      \ar[rr]
      &&
      \mathrm{Cat}
      \\
      &
      X_1 
      \ar[d, "{f}"]
      &\longmapsto&
      \mbox{\bf C}_{X_1}
      \ar[from=d, "{f^\ast}"]
      \\[+20pt]
      &
      X_2 
      &\longmapsto& 
      \mbox{\bf C}_{X_2}
    \end{tikzcd}
  \end{equation}
\end{definition}

\begin{definition}[Grothendieck construction {\cite[\S VI.8]{Grothendieck71}, cf.  \cite[\S 3.1.3]{Vistoli05}}]
  \label{GrothendieckConstruction}
 $\,$  
  
\noindent {\bf (i)}  The {\it Grothendieck construction} on a covariant pseudofunctor
$\mathbf{C}_{(-)} : \mathcal{B} \longrightarrow \mathrm{Cat}$ 
\eqref{CovariantPseudofunctor} 
is the category $\int_{X \in\mathcal{B}} \mathbf{C}_X$ whose 

  \begin{itemize}
    \item objects $\mathscr{V}_X$ are pairs $(X, \mathscr{V})$ with $X \in \mathrm{Obj}(\mathcal{B})$ and $\mathscr{V} \,\in\, \mathrm{Obj}(\mathbf{C}_{X})$,
    \item morphisms $\phi_f \,\colon\, \mathscr{V}_X \longrightarrow \mathscr{W}_{\mathrm{Y}}$ are pairs $(f,\phi)$ with $f : X \longrightarrow Y$ in $\mathcal{B}$ and 
    \fbox{$\phi \,:\, f_! \mathscr{V} \longrightarrow \mathscr{W}$ in $\mathbf{C}_{Y}$}\,,
  \end{itemize}
  hence the hom-sets of the 
  covariant Grothendieck construction
  are these dependent products:
  \begin{equation}
    \label{HomSetsOfCovariantGrothendieckConstruction}
    \Big(
    \int_{\mathcal{B}}
    \mathbf{C}
    \Big)
    \Big(
      \mathscr{V}_X
      ,\,
      \mathscr{W}_Y      
    \Big)
    \;\;:\defneq\;\;
    \big(
      f \in \mathcal{B}(X,Y)
    \big)
    \times
    \mathbf{C}_Y
    \big(
      f_! X
      ,\,
      Y
    \big).
  \end{equation}

\vspace{-1mm} 
\noindent {\bf (ii)}
  Dually, the  {\it Grothendieck construction} on a contra-variant pseudofunctor $\mathbf{C}_{(-)} : \mathcal{B}^{\mathrm{op}} \longrightarrow \mathrm{Cat}$
  \eqref{ContravariantPseudofunctor} is the category $\int_{X \in\mathcal{B}} \mathbf{C}_X$ whose 

  \begin{itemize}
    \item objects $\mathscr{V}_X$ are pairs $(X, \mathscr{V})$ with $X \in \mathrm{Obj}(\mathcal{B})$ and $\mathscr{V} \,\in\, \mathrm{Obj}(\mathbf{C}_{X})$,
    \item morphisms $\phi_f \,\colon\, \mathscr{V}_X \longrightarrow \mathscr{W}_{Y}$ are pairs $(f,\phi)$ with $f : X \longrightarrow Y$ in $\mathcal{B}$ 
    and \fbox{$\phi \,:\, \mathscr{V} \longrightarrow f^\ast \mathscr{W}$ in $\mathbf{C}_{X}$}\,,
  \end{itemize}
  hence the hom-sets of the 
  contravariant Grothendieck construction
  are these dependent products:
  \begin{equation}
    \label{HomSetsOfContravariantGrothendieckConstruction}
    \Big(
    \int_{\mathcal{B}}
    \mathbf{C}
    \Big)
    \Big(
      \mathscr{V}_X
      ,\,
      \mathscr{W}_Y      
    \Big)
    \;\;:\defneq\;\;
    \big(
      f \in \mathcal{B}(X,Y)
    \big)
    \times
    \mathbf{C}_Y
    \big(
      X
      ,\,
      f^\ast Y
    \big)
    \,.
  \end{equation}

\vspace{-3mm} 
\noindent {\bf (iii)} Finally, composition of morphisms 
$\!\!
  \begin{tikzcd}[sep=15pt]
    \mathscr{V}_X
    \ar[r, "{ \phi_f }"]
    &
    \mathscr{W}_{Y}
    \ar[r, "{ \psi_g }"]
    &
    \mathscr{R}_{Z}
  \end{tikzcd}
\!\!$
in the Grothendieck construction is defined by using the  pseudo-functoriality of $\mathbf{C}_{(-)}$ 
to coherently push (or pull) morphisms into the codomain or domain category:
$$
  \psi_g \circ \phi_f
  \;:=\;
  \big( 
    \psi 
      \,\circ\, 
    g_!(\phi) 
      \,\circ\,
    \mu_{f,g}(\mathscr{V})
  \big)_{g \circ f}
  \hspace{1cm}
  \mbox{or}
  \hspace{1cm}
  \psi_g \circ \phi_f
  \;:=\;
  \big(
    \mu_{f,g}(\mathscr{R})
      \,\circ\,
    g^\ast(\psi) 
      \,\circ\, 
    \phi 
  \big)_{g \circ f}
  \,.
$$
Here one is using the coherence isomorphisms \eqref{CompositionalCoherenceIsomorphism} 
to adjust for the identification of composite functors:
\vspace{-1mm} 
$$
\hspace{-2mm} 
  \begin{tikzcd}[sep=20pt]
    (g \,\circ\, f )_!(\mathscr{V})
    \ar[
      r, 
      "{ \mu_{f,g} }"{yshift=1pt}, 
      "{\sim}"{swap}
    ]
    &
    g_!\big(f_! (\mathscr{V}) \big)
    \ar[
      r,
      "{ g_!(\phi) }"
    ]
    &
    g_!\big( \mathscr{W} \big)
    \ar[r, "{ \psi }"]
    &
    \mathscr{R}
  \end{tikzcd}
  \hspace{.4cm}
  \mbox{or}
  \hspace{.4cm}
  \begin{tikzcd}[sep=20pt]
    \mathscr{V}
    \ar[
      r,
      "{ \phi }"
    ]
    &
    f^\ast\big( \mathscr{W} \big)
    \ar[
      r, 
      "{ f^\ast(\psi) }"
    ]
    &
    f^\ast
    \big(
      g^\ast(
        \mathscr{R}
      )
    \big)
    \ar[
      r,
      "{
        \mu_{f,g}
      }"{yshift=1pt}
    ]
    &
    (g \,\circ\, f)^\ast
    \mathscr{R}\;.
  \end{tikzcd}
$$
\end{definition}
\begin{remark}[Grothendieck fibration]
A key aspect of the Grothendieck construction is that it is a {\it fibered category} over the original diagram shape, 
and as such an equivalent incarnation of the pseudo-functor that induced it. While important, here we do not need this 
aspect and will regard the Grothendieck construction as a plain category, this being the domain category of 
the corresponding Grothendieck fibration.
\end{remark}

\begin{example}[Categories of indexed sets of objects {\cite[\S 3]{Benabou85}}, Free coproduct completion {\cite[\S 2]{HT95}}]
\label{CategoriesOfIndexedSetsOfObjects}
$\,$
\newline
  For $\mathcal{C}$ any category,  there is the contravariant pseudofunctor (Def. \ref{Pseudofunctor}) on $\mathrm{Set}$ 
  which to a set $S$ assigns the $S$-fold product category of $\mathcal{C}$ with itself:
  \vspace{-2mm} 
  \begin{equation}
    \label{PseudofunctorOfProductCategories}
    \hspace{-1cm}
    \mathcal{C} \,\in\, \mathrm{Cat}
    \hspace{1cm}
      \vdash
    \hspace{1cm}
    \begin{tikzcd}[row sep=-6pt, column sep=small]
      \mathrm{Set}^{\mathrm{op}}
      \ar[rr]
      &&
      \mathrm{Cat}
      \\[3pt]
      S
      \ar[
        dd,
        "{ f }"
      ]
      &\longmapsto&
      \mathrm{Func}(S,\,\mathcal{C})
      \ar[
        from=dd,
        "{ f^\ast }"{swap}
      ]
      \ar[r, phantom, "{ \defneq }"]
      &
      \mathcal{C}^S
      \;\simeq\;
      \underset{s \in S}{\prod}
      \mathcal{C}
      \\[+20pt]
      \\
      T &\longmapsto&
      \mathrm{Func}(T,\,\mathcal{C})
      \ar[r, phantom, "{ \defneq }"]
      &
      \mathcal{C}^T
      \;\simeq\;
      \underset{t \in T}{\prod}
      \mathcal{C}
      \,.
    \end{tikzcd}
  \end{equation}

  \vspace{-3mm} 
  \noindent  
  equivalently, the 
  functor category into $\mathcal{C}$ out of the discrete category on $S$:
  \vspace{-2mm} 
  $$
    \begin{tikzcd}[row sep=-4pt, column sep=small]
    \mathrm{Func}(
      S
      ,\,
      \mathcal{C}
    )
    \ar[rr, "{ \sim }"]
    &&
    \underset{s \in S}{\prod}
    \mathcal{C}
    \\
    (
      s 
      \,\mapsto\,
      \mathscr{V}_s
    )
    &\longmapsto&
    (
      \mathscr{V}_s
    )_{s \in S}
    \,,
   \end{tikzcd}
  $$

  \vspace{-2mm} 
\noindent  and whose base change functors are given by precomposition with, hence re-indexing by, the given map of sets:
  \vspace{-2mm} 
  $$
  \hspace{-1cm} 
    f \,:\, S \longrightarrow T
    \hspace{1cm}
      \vdash
    \hspace{1cm}
    \begin{tikzcd}[row sep=-4pt, column sep=small]
      \mathrm{Func}(S,\,\mathcal{C})
      \ar[
        from=rr,
        "{ f^\ast }"{swap}
      ]
      &&
      \mathrm{Func}(T,\,\mathcal{C})
      \\
      \big(
        \mathscr{V}_{f(s)}
      \big)_{s \in S}
      &\longmapsfrom&
      (
        \mathscr{V}_{t}
      )_{t \in T}
      \,.
    \end{tikzcd}
  $$

   \vspace{-2mm} 
\noindent
  Accordingly, the Grothendieck construction (Def. \ref{GrothendieckConstruction}) on this pseudofunctor,

  \noindent
  \fbox{$\underset{S \in \mathrm{Set}}{\int} \!  \mathcal{C}^S$} has the following description:
  \begin{itemize}[leftmargin=.5cm]
    \item objects $\mathscr{V}_{\!S}$ are dependent pairs consisting of a set $S \in \mathrm{Set}$ and an $S$-tuple 
    $\big( \mathscr{V}_s \big)_{s \in S}$ of objects in $\mathcal{C}$,

    \item morphisms $\phi_f \,:\, \mathscr{V}_S \longrightarrow \mathscr{W}_T$ are $S$-tuples 
    $
      \big(
        \phi_s
        \,:\,
        \mathscr{V}_s
        \xrightarrow{\;\;}
        \mathscr{W}_{f(s)}
      \big)_{s \in S}
    $
    of morphisms in $\mathcal{C}$.
  \end{itemize}
  Independently of whether or how $\mathcal{C}$ has co-products,  this category has set-indexed coproducts
  ${\coprod}_i
    \mathscr{V}(i)_{S_i}$ with underlying set $\coprod_i S_i$ and components
  $\big({\coprod}_i
    \mathscr{V}(i)_{S_i}\big)_{s_j} \,=\, \mathscr{V}(j)_{s_j}$ for $s_j \in S_j$.
  \medskip
  
  But if the
  category $\mathcal{C}$ is {\it extensive}, in that it already has coproducts itself and the coproduct-functors between (products of)
  slice categories are equivalences
  \vspace{-3mm} 
  $$
    S \,\in\, \mathrm{Set}
    \hspace{1cm}
      \vdash
    \hspace{1cm}
    \begin{tikzcd}[row sep=-4pt, column sep=small]
      \underset{s \in S}{\prod}
      \mathcal{C}_{/X_s}
      \ar[
        rr,
        "{ \sim }"
      ]
      &&
      \mathcal{C}_{/ \coprod_s X_s }
      \\
  \scalebox{0.85}{$    \left(\!\!
      \begin{array}{c}
        E_s
        \\
        \downarrow
        \\
        X_s
      \end{array}
     \!\! \right)_{\!\!\! s \in S}
      $}
      &\longmapsto&
   \scalebox{0.85}{$      \left(\!\!
      \begin{array}{c}
        \coprod_{\, s}
        E_s
        \\
        \downarrow
        \\
        \coprod_{\, s} X_s
      \end{array}
      \!\!\right)
      $}
    \end{tikzcd}
  $$

   \vspace{-2mm} 
\noindent
  then the construction yields the category of bundles in $\mathcal{C}$ over sets, the latter understood via the unique coproduct-preserving 
  inclusion $\iota_{\mathrm{Set}} \,\colon\, \mathrm{Set} \hookrightarrow \mathcal{C}$, hence the comma category $(\mathrm{id}_{\mathcal{C}},\,\iota_{\mathrm{Set}})$:
   \vspace{-2mm} 
$$
  \mbox{$\mathcal{C}$ extensive}
  \hspace{1cm}
    \vdash
  \hspace{1cm}
  \underset{S \in \mathrm{Set}}{\int}
  \,
  \underset{s \in S}{\prod}
  \,
  \mathcal{C}
  \;\;\;
    \simeq
  \;\;\;
  (\mathrm{id}_{\mathcal{C}},\,\iota_{\mathrm{Set}})  
$$

 \vspace{-2mm} 
\noindent
whose morphisms $\phi_f \,:\, X_S \longrightarrow Y_T$ are commuting diagrams in $\mathcal{C}$ of this form:
\vspace{-3mm} 
$$
  \begin{tikzcd}[row sep=small]
    \underset{s \in S}{\coprod}
    X_s
    \ar[rr, "\phi"]
    \ar[d]
    &&
    \underset{t \in T}{\coprod}
    Y_t
    \ar[d]
    \\
    S 
    \ar[rr, "{f}"]
      &&
    T
    \mathrlap{\,.}
  \end{tikzcd}
$$  

 \vspace{-2mm} 
\noindent
Conversely, if $\mathcal{C}$ is not extensive, then we may understand $\int_{S \in \mathrm{Set}} \mathcal{C}^S$ as the stand-in for 
the would-be category of ``$\mathcal{C}$-fiber bundles'' over sets.
\end{example}
\begin{proposition}[{\cite[Lem. 4.2]{CarboniVitale98}}]
  \label{CategoriesBeingCoproductCompletionOfTheirConnectedObjects}
  A category $\mathcal{C}$ with all set-indexed coproducts each of whose objects is a coproduct of connected objects
  is the free coproduct completion (Ex. \ref{CategoriesOfIndexedSetsOfObjects}) of its full subcategory of connected objects 
  (i.e., of those objects $X \in \mathcal{C}$ for which $\mathcal{C}(X,-) : \mathcal{C} \to \mathcal{C}$ preserves coproducts).
\end{proposition}
\begin{proof}
  Since, by assumption, every object is already presented by an indexed set of connected objects, it remains to see that also the morphisms 
  $\big(\coprod_{\, s} X_s\big) \longrightarrow \big(\coprod_{\, t} Y_t\big)$ are in bijection to indexed sets of morphisms of connected 
  objects. This follows by
  $$
    \def\arraystretch{1.4}
    \begin{array}{ll}
      \mathcal{C}
      \big(
        \coprod_{\, s} X_s
        ,\,
        \coprod_{\, t} Y_t
      \big)
      &      \;\simeq\;
      \prod_{\, s} 
      \mathcal{C}
      \big(
        X_s
        ,\,
        \coprod_{\, t} Y_t
      \big)
      \\
    &  \;\simeq\;
      \underset{s \in S}{\prod} 
      \;
      \underset{t_s \in T}{\coprod}
      \mathcal{C}
      \big(
        X_s
        ,\,
        Y_{t_s}
      \big)
      \\
  &    \;\simeq\;
      \underset{f : S \to T}{\coprod}
      \;
      \underset{s \in S}{\prod}
      \mathcal{C}
      \big(
        X_s
        ,\,
        Y_{f(s)}
      \big)\,,
    \end{array}
  $$
  where the first bijection is by general properties of Hom-functors and the second is by the assumption that all $X_s$ are connected.
\end{proof}

\begin{example}[Induced adjunctions between Grothendieck constructions]
  \label{InducedAdjunctionBetweenGrothendieckConstructions}
  Given a contravariant pseudofunctor and a left adjoint functor into its domain
   \vspace{-2mm} 
  $$
    \begin{tikzcd}
      \mathcal{C}
      \ar[
        rr,
        shift left=5pt,
        "{
          L
        }"
      ]
      \ar[
        from=rr,
        shift left=5pt,
        "{
          R
        }"
      ]
      \ar[
        rr,
        phantom,
        "{ \scalebox{.7}{$\bot$} }"
      ]
      &&
      \mathcal{B}
    \end{tikzcd}
    \hspace{1cm}
    \begin{tikzcd}
      \mathcal{B}^{\mathrm{op}}
      \ar[r, "{ \mathbf{C}_{(-)} }"]
      &
      \mathrm{Cat}
    \end{tikzcd}
  $$

   \vspace{-2mm} 
\noindent
  there is an induced adjunction between the Grothendieck constructions on $\mathbf{C}_{(-)}$ and on $\mathbf{C}_{L(-)}$, covering the given adjunction:
  \begin{equation}
    \label{}
    \begin{tikzcd}
    \Big(\,
    \underset{
      c \in \mathcal{C} 
    }{\int}  
    \mathbf{C}_{L(c)}
    \Big)
    \ar[
      rr,
      shift left=5pt,
      "{ \hat L }"
    ]
    \ar[
      from=rr,
      shift left=5pt,
      "{ \hat R }"
    ]
    \ar[
      rr,
      phantom,
      "{ \scalebox{.7}{$\bot$} }"
    ]
    \ar[d]
    &&
    \Big(\,
    \underset{
      b \in \mathcal{B} 
    }{\int}  
    \mathbf{C}_{b}
    \Big)
    \ar[d]
    \\
    \mathcal{C}
    \ar[
      rr,
      shift left=5pt,
      "{
        L
      }"
    ]
    \ar[
      from=rr,
      shift left=5pt,
      "{
        R
      }"
    ]
    \ar[
      rr,
      phantom,
      "{ \scalebox{.7}{$\bot$} }"
    ]
    &&
    \mathcal{B}
    \end{tikzcd}
  \end{equation}
  where on components in $\mathbf{C}_{(-)}$ the functor $\widehat{L}$ is the identity 
  while $\widehat{R}$ is pullback along the underlying adjunction counit $\epsilon^{L \dashv R}  \,:\, L \circ R \to \mathrm{id}$:
  \vspace{-2mm} 
  \begin{equation}
    \label{LiftOfAdjunctionToGrothendieckConstruction}
    \begin{tikzcd}
      \mathscr{V}_{c}
      \ar[d, "{ \phi_f }"]
      \\
      \mathscr{V}'_{c'}
    \end{tikzcd}
    \;\;\;
    \overset{\widehat{L}}{\longmapsto}
    \;\;\;
    \begin{tikzcd}
      \mathscr{V}_{L(c)}
      \ar[d, "{ \phi_{L(f)} }"]
      \\
      \mathscr{V}'_{L(c')}
    \end{tikzcd}
    \hspace{1cm}
    \mbox{and}
    \hspace{1cm}
    \begin{tikzcd}
      \mathscr{V}_{b}
      \ar[d, "{ \phi_f }"]
      \\
      \mathscr{V}'_{b'}
    \end{tikzcd}
    \;\;\;
    \overset{\widehat{R}}{\longmapsto}
    \;\;\;
    \begin{tikzcd}
      \mathscr{V}_{R(b)}
      \ar[
        d, 
        "{ 
          (\epsilon^{L \dashv R}_{b})^\ast
          \phi_{L(f)} 
        }"
      ]
      \\
      \mathscr{V}'_{R(b')}
    \end{tikzcd}
  \end{equation}

\vspace{-2mm} 
\noindent  The counit of this adjunction is given by the identity component map covering the underlying counit:
  \begin{equation}
    \label{CounitOfLiftOfAdjunctionToGrothendieckConsturction}
    \epsilon^{\widehat{L}\dashv \, \widehat{R}}_{\mathscr{V}_{\mathbf{x}}}
    \;:\;
    \widehat{L}\, \widehat{R}
    \big(
      \mathscr{V}_{b}
    \big)
    \,=\,
    \big(
      \epsilon^{L \dashv \, R}_{L R(b)}
      \mathscr{V}
    \big)_{}.
  \end{equation}
\end{example}

\begin{proposition}[Colimits in a Grothendieck construction {\cite[\S 3.2, Thm. 2]{TBG91}\cite[Prop. 2.4.4]{HarpazPrasma15}}] 
\label{ColimitsInAGrothendieckConstruction}
$\,$
\newline
 \noindent {\bf (i)} The Grothendieck construction $\int_{X \in \mathcal{B}} \mathbf{C}_X$ (Def. \ref{GrothendieckConstruction})
 on a covariant pseudofunctor $\mathbf{C}_{(-)} : \mathcal{B} \xrightarrow{\phantom{-}} \mathrm{Cat}$ \eqref{Pseudofunctor} 
 is cocomplete as soon as the base category $\mathcal{B}$ 
  as well as all the fiber categories $\mathbf{C}_X$, $X \in \mathcal{B}$ are cocomplete.
  In this case the colimit of a small diagram
  $$
    \begin{tikzcd}[sep=0pt]
      I
      \ar[rr]
      && 
      \int_{X \in \mathcal{B}}\mathbf{C}_X
      \\
      i &\mapsto& \mathscr{V}(i)_{X_i}
    \end{tikzcd}
  $$

  \vspace{-2mm} 
\noindent  is given by
  \begin{equation}
    \label{FormulaForColimitInGrothendieckConstruction}
    \underset{\underset{i \in I}{\longrightarrow}}{\lim}
    \big(
      \mathscr{V}(i)_{X_i}
    \big)
    \;\;\;\;
    \simeq
    \;\;\;\;
    \Big(
     \underset{\longrightarrow}{\mathrm{lim}}_i
    \big(
      q(i)_! \mathscr{V}(i)
    \big)
    \Big)_{
      \underset{\longrightarrow}{\mathrm{lim}}_i  
      X_i
    }
    \;\;\;\;\;
    \in
    \;
    \int_{X \in \mathcal{B}}
    \mathbf{C}_X
    \,,
  \end{equation}
  where
  $$
    i\,\in\, I
    \hspace{1cm}
    \vdash
    \hspace{1cm}
    q(i)
    \;:\;
    X_i 
    \xrightarrow{\phantom{--}}
    \underset{\longrightarrow}{\lim}_i X_i
    \;\;\;
    \in
    \;
    \mathcal{B}
  $$
  denote the coprojections into the underlying colimit in $\mathcal{B}$.
  
  \noindent {\bf (ii)} The analogous dual statement holds for limits.
\end{proposition}
\begin{example}[External cartesian product]
  \label{ExternalCartesianProduct}
  Given a contravariant pseudofunctor $\mathbf{C}_{(-)} : \mathcal{B}^{\mathrm{op}} \xrightarrow{\phantom{-}} \mathrm{Cat}$
  such that both $\mathcal{B}$ as well as all the $\mathbf{C}_{(-)}$ have Cartesian products, then its Gorthendieck 
  construction has cartesian products given by
  \begin{equation}
    \label{FormulaForExternalCartesianProduct}
    \mathscr{V}_{X}
    \times
    \mathscr{W}_{Y}
    \;\;
    \simeq
    \;\;
    \Big(\!
    \big( 
      (\mathrm{pr}_{X})^\ast \mathscr{V}
    \big)
    \times
    \big( 
      (\mathrm{pr}_{Y})^\ast \mathscr{W}
    \big)
    \!\Big)_{X \times Y}
    \,.
  \end{equation}
  More explicitly, the components of the external Cartesian product are
   \vspace{-2mm} 
  $$
    \def\arraystretch{1.3}
    \begin{array}{lll}
      \big(
        \mathscr{V}_X \times \mathscr{W}_Y
      \big)_{(x,y)}
      &      \;\simeq\;
      \{(x,y)\}^\ast
      \Big(
        \big(
          (\mathrm{pr}_X)^\ast
          \mathscr{V}
        \big)
        \times
        \big(
          (\mathrm{pr}_Y)^\ast
          \mathscr{W}
        \big)
      \Big)
      \\
   &   \;\simeq\;
      \Big(
        \big(
          \{(x,y)\}^\ast
          (\mathrm{pr}_X)^\ast
          \mathscr{V}
        \big)
        \times
        \big(
          \{(x,y)\}^\ast
          (\mathrm{pr}_Y)^\ast
          \mathscr{W}
        \big)
      \Big)
      \\
    &  \;\simeq\;
      \big(
        \{x\}^\ast \mathscr{V}
      \big)
      \times
      \big(
        \{y\}^\ast \mathscr{W}
      \big)
      \\
  &    \;\simeq\;
      \mathscr{V}_{\!x}
      \times
      \mathscr{W}_{\!y}
    \end{array}
    \hspace{.7cm}
    \begin{tikzcd}
      \{x\}
      \ar[d]
      &
      \{(x,y)\}
      \ar[l, "{\sim}"{swap}]
      \ar[r, "{\sim}"]
      \ar[d, hook]
      &
      \{y\}
      \ar[d, hook]
      \\
      X 
      &
      X \times Y
      \ar[r, "{ \mathrm{pr}_Y }"]
      \ar[l, "{ \mathrm{pr}_X }"{swap}]
      &
      Y
    \end{tikzcd}
  $$
\end{example}
This gives the following elementary fact, which is crucial in the main text:
\begin{proposition}[Free coproduct completion]
 If a category $\mathcal{C}$ has Cartesian products, then its free coproduct completion (Ex. \ref{CategoriesOfIndexedSetsOfObjectsWithCoproducts}) 
also  has Cartesian products and those distribute \eqref{DistributivityOfMonoidalStructure}
 over the coproducts.
\end{proposition}

\medskip

\noindent
{\bf The 2-category of categories with adjoint functors between them.}
We extract the gist of the discussion in \cite[p. 97-103]{MacLane97}.
\begin{definition}[Conjugate transformation of adjoints {\cite[p. 98]{MacLane97}}]
  \label{ConjugateTransformationOfAdjoints}
  Given a pair of pairs of adjoint functors between the same categories
  \vspace{-2mm} 
  $$
    \begin{tikzcd}
      \mathcal{C}
      \ar[
        rr,
        shift left=5pt,
        "{ L_i }"
      ]
      \ar[
        from=rr,
        shift left=5pt,
        "{ R_i }"
      ]
      \ar[
        rr,
        phantom,
        "{ \scalebox{.7}{$\bot$} }"
      ]
      &&
      \mathcal{D}
    \end{tikzcd}
    \hspace{1cm}
    i \,\in\, \{1,2\}\,,
  $$

  \vspace{-3mm} 
\noindent   then a {\it conjugate transformation} between them
 \vspace{-2mm} 
  $$
    (\lambda ,\, \rho)
    \;:\;
    \begin{tikzcd}
    (L_1 \dashv R_1) 
    \ar[
      r,
      Rightarrow
    ]
    &
    (L_2 \dashv R_2)
    \end{tikzcd}
  $$

   \vspace{-2mm} 
\noindent
  is a pair of natural transformations of the form
   \vspace{-2mm} 
  $$
    \lambda : L_1 \Rightarrow L_2
    ,\,
    \;\;\;
    \rho : R_2 \Rightarrow R_1
  $$

   \vspace{-2mm} 
\noindent
  such that they make the following square of natural transformations of hom-sets commute, 
  where the horizontal maps refer to the given hom-isomorphisms:
   \vspace{-2mm} 
  \begin{equation}
    \label{ConjugacyOfTransformations}
    \begin{tikzcd}[row sep=small]
      \mathcal{C}\big(
        L_2(-)
        ,\,
        -
       \big)
     \ar[
       rr,
       "{ \sim }"
     ]
     \ar[
       d,
       "{
         \mathcal{C}\left(
           \lambda_{(-)} 
           ,\, 
           \mathrm{id}_{(-)}
        \right)
       }"{swap}
     ]
     &&
     \mathcal{D}\big(
       -
       ,\,
       R_2(-)
     \big)
     \ar[
       d,
       "{
         \mathcal{D}\left(
           \mathrm{id}_{(-)}
           ,\, 
           \rho_{(-)} 
        \right)
       }"
     ]
     \\[+10pt]
      \mathcal{C}\big(
        L_1(-)
        ,\,
        -
       \big)
     \ar[
       rr,
       "{ \sim }"
     ]
     &&
     \mathcal{D}\big(
       -
       ,\,
       R_1(-)
     \big)     
    \end{tikzcd}
  \end{equation}

   \vspace{-2mm} 
\noindent
  Such conjugate transformations compose via composition of their components $(\lambda, \rho)$, yielding a category of adjoint functors 
  with conjugate transformations between them, which we denote as follows:
  \vspace{-2mm} 
  \begin{equation}
    \label{CategoryOfAdjointFunctorsWithConjugateTransformations}
    \mathcal{C},\, \mathcal{D}
    \,\in\,
    \mathrm{Cat}
    \;\;\;\;\;\;\;\;\;\;\;\;
    \vdash
    \;\;\;\;\;\;\;\;\;\;\;\;
    \mathrm{Cat}_{\mathrm{adj}}
    (\mathcal{C},\,\mathcal{D})
    \;\in\;
    \mathrm{Cat}
    \,.
  \end{equation}
\end{definition}

\begin{proposition}[Uniqueness of conjugate transformations {\cite[p. 98]{MacLane97}}]
  \label{UniquenessOfConjugateTransformations}
  Given $ L_i \dashv R_i \,:\, \mathcal{C} \rightleftarrows \mathcal{D}$ and $\lambda$
  in Def. \ref{ConjugateTransformationOfAdjoints}, there is a \emph{unique} $\rho$ that completes this data to a conjugate transformation.
  In other words, the forgetful functor from \eqref{CategoryOfAdjointFunctorsWithConjugateTransformations} to the functor category is 
  a fully faithful sub-category inclusion:
  \vspace{-1mm} 
  \begin{equation}
    \label{FullEmbeddingOfCatAdhCD}
    \hspace{1.5cm} 
    \begin{tikzcd}[row sep=-4pt, column sep=small]
    \mathllap{
    \mathcal{C},\, \mathcal{D}
    \;\in\;
    \mathrm{Cat}
    \;\;\;\;\;\;\;\;\;\;\;\;
    \vdash
    \;\;\;\;\;\;\;\;\;\;\;\;    
    }
    \mathrm{Cat}_{\mathrm{adj}}(\mathcal{C},\,\mathcal{D})
    \ar[rr, hook]
    &&
    \mathrm{Cat}(\mathcal{C},\,\mathcal{D})
    \\
    (L_1 \dashv R_1)
    \ar[d, "{ (\lambda,\, \rho) }"]
    &\longmapsto&
    L_1
    \ar[
      d,
      "{ \lambda }"
    ]
    \\[23pt]
    (L_2 \dashv R_2)
    &\longmapsto&
    L_2
    \mathrlap{\,.}
    \end{tikzcd}
  \end{equation}
\end{proposition}

\begin{proposition}[Horizontal composition of conjugate transformations {\cite[p. 102]{MacLane97}}]
  \label{HorizontalCompositionOfConjugateTransformations} 
  The horizontal composition $(-)\cdot(-)$ of 
  the underlying natural transformations of a pair of conjugate transformations Def. \ref{ConjugateTransformationOfAdjoints} is itself a conjugate transformation, so that the composition functor on functor categories restricts along the inclusions \eqref{FullEmbeddingOfCatAdhCD}:
  \vspace{-1mm} 
  $$
    \hspace{3cm} 
    \begin{tikzcd}[sep=0pt]
      \mathllap{
    \mathcal{C}
    ,\,
    \mathcal{D}
    ,\,
    \mathcal{E}
    \;\in\;
    \mathrm{Cat}
    \;\;\;\;\;\;\;\;\;\;\;
    \vdash
    \;\;\;\;\;\;\;\;\;\;\;      
      }
      \mathrm{Cat}_{\mathrm{adj}}(\mathcal{D},\,\mathcal{E})
      \times
      \mathrm{Cat}_{\mathrm{adj}}(\mathcal{C},\,\mathcal{D})
      \ar[
        rr
      ]
      &&
      \mathrm{Cat}_{\mathrm{adj}}(\mathcal{C},\,\mathcal{E})
      \\
     \scalebox{0.9}{$   \big(
        (\lambda,\,\rho)
        ,\,
        (\lambda',\,\rho')
      \big)
      $}
      &\longmapsto&
   \scalebox{0.9}{$   \big(
        \lambda' \cdot \lambda
        ,\,
        \rho \cdot \rho'
      \big).
      $}
    \end{tikzcd}
  $$
\end{proposition}

Via Prop. \ref{HorizontalCompositionOfConjugateTransformations}, we have: 

\begin{definition}[2-category of categories, adjoint functors and conjugate transformations {\cite[p. 102]{MacLane97}}] 
\label{CatAdj}
  Write
  \vspace{-2mm} 
  \begin{equation}
    \label{FunctorFromCatAdjToCat}
    \begin{tikzcd}
      \mathrm{Cat}_{\mathrm{adj}}
      \ar[r]
      &
      \mathrm{Cat}
    \end{tikzcd}
  \end{equation}
  \vspace{-2mm} 

 \noindent   
 for the (very large) locally full sub-2-category of $\mathrm{Cat}$ whose
  \begin{itemize}
    \item objects are categories,
    \item hom-categories are those \eqref{CategoryOfAdjointFunctorsWithConjugateTransformations}
    of adjoint functors with conjugate transformations between them.
  \end{itemize}
\end{definition}

\begin{proposition}[Bivariant pseudofunctors,
{cf. \cite[Lem. 9.1.2]{Jacobs98}\cite[Prop. 2.2.1]{HarpazPrasma15}\cite[pp. 10]{CagneMellies20}}] 
\label{RecognitionOfBivariantPseudofunctors}
  Given a covariant pseudofunctor $\mathbf{C}_{(-)}$ (Def. \ref{Pseudofunctor}) such that each component 
  functor $f_! : \mathbf{C}_X \longrightarrow \mathbf{C}_{Y}$ has a right adjoint
  \vspace{-2mm} 
  \begin{equation}
    \label{BivariantPseudofunctor}
    \begin{tikzcd}[row sep=-2pt, column sep=20pt]
      \mbox{\bf C}_{(-)}
      \;\colon\;
      &
      \mathcal{B}
      \ar[rr]
      &&
      \mathrm{Cat}
      \\
      &
      X_1 
      \ar[d, "{f}"]
      &\longmapsto&
      \mbox{\bf C}_{X_1}
      \ar[d, shift right=7pt, "{f_!}"{swap}]
      \ar[from=d, shift right=7pt, "{f^\ast}"{swap}]
      \ar[d, phantom, "{ \scalebox{.8}{$\dashv$} }"]
      \\[+20pt]
      &
      X_2 
      &\longmapsto& 
      \mbox{\bf C}_{X_2}
    \end{tikzcd}
  \end{equation}

  \vspace{-2mm} 
\noindent  then:
  \begin{itemize}
    \item[{\bf (i)}] it factors essentially uniquely through $\mathrm{Cat}_{\mathrm{adj}}$ \eqref{FunctorFromCatAdjToCat},
    \item[{\bf (ii)}]  hence it induces a contravariant pseudofunctor with component functors $f^\ast$,
    \item[{\bf (iii)}] 
    such that the Grothendieck construction (Def. \ref{GrothendieckConstruction}) on the covariant pseudofunctor is equivalent to 
    that on the corresponding contravariant pseudofunctor via the functor that is the identity on objects and on morphisms is 
    the hom-isomorphism of the given adjoint pairs:
    \vspace{-2mm} 
    $$
      \begin{tikzcd}[row sep=-2pt]
        \mathscr{f_! \mathscr{V}}
        \ar[r, "{\widetilde{\phi}}"]
        &
        \mathscr{W}
        \\
        X \ar[r, "f"]
        &
        Y
      \end{tikzcd}
      \;\;\;\;\;\;\;\;\;\;
      \leftrightarrow
      \;\;\;\;\;\;\;\;\;\;
      \begin{tikzcd}[row sep=-2pt]
        \mathscr{\mathscr{V}}
        \ar[r, "{{\phi}}"]
        &
        f^\ast \mathscr{W}
        \\
        X \ar[r, "f"]
        &
        Y
        \mathrlap{\,.}
      \end{tikzcd}
    $$

     \vspace{-2mm} 
\noindent
    Therefore, both construction are still unambiguously denoted by $\int_{X \in \mathcal{B}} \mathbf{C}_{X}$.
    \end{itemize}
\end{proposition}
\begin{proof}
  The first statement is a direct consequence of Prop. \ref{UniquenessOfConjugateTransformations}, the second then follows 
  by Prop. \ref{HorizontalCompositionOfConjugateTransformations} and finally the third by the property \eqref{ConjugacyOfTransformations} 
  in Def. \ref{ConjugateTransformationOfAdjoints}.
\end{proof}

In refinement of Ex. \ref{CategoriesOfIndexedSetsOfObjects}, we have:

\begin{example}[Categories of indexed sets of objects with coproducts]
\label{CategoriesOfIndexedSetsOfObjectsWithCoproducts}
If a category $\mathcal{C}$ already has all coproducts, then the pseudofunctor \eqref{PseudofunctorOfProductCategories}
of its product categories has left adjoint component functors given by forming coproducts over fibers of base maps
\vspace{-4mm} 
\begin{equation}
\label{BaseChangeFunctorsBetweenProductCategories}
  f \,:\, S \longrightarrow T
  \hspace{1cm}
    \vdash
  \hspace{1cm}
  \begin{tikzcd}[row sep=-3pt, column sep=10pt]
 \scalebox{0.8}{$   \big(
      \mathscr{V}_s
    \big)_{ \in S}
    $}
    &\xmapsto{\qquad}&
   \scalebox{0.8}{$  \bigg(\,
      \underset{
        s \in f^{-1}(\{t\})
      }{\coprod}
      \mathscr{V}_s
    \bigg)_{t \in T}
    $}
    \\
    \mathrm{Func}
    \big(
      S
      ,\,
      \mathcal{C}
    \big)
    \ar[
      rr,
      shift left=5pt,
      "{ f_! }"
    ]
    \ar[
      from=rr,
      shift left=5pt,
      "{ f^\ast }"
    ]
    \ar[
      rr,
      phantom,
      "{ \scalebox{.7}{$\bot$} }"
    ]
    &&
    \mathrm{Func}
    \big(
      T
      ,\,
      \mathcal{C}
    \big)    
    \\[6pt]
  \scalebox{0.8}{$   \big(
      \mathscr{V}_{f(s)}
    \big)_{s \in S}
    $}
    \ar[
      rr,
      phantom,
      "{ \xmapsfrom{\qquad} }"
    ]
    &&
 \scalebox{0.8}{$    \big(
      \mathscr{V}_{t}
    \big)_{t \in T}
    $}
  \end{tikzcd}
\end{equation}
\end{example}

\newpage

\noindent
{\bf Data availability statement.}
\newline
\noindent
Data sharing is not applicable to this article as no new data were created or analyzed in this study.

\smallskip

\noindent
{\bf Conflict of interest statement.}
\newline
\noindent
The authors declare that they have no conflict of interest.

\medskip


\end{document}